%% file: Std_Model_arXiv.tex
\documentclass[a4paper,reqno]{article}

\input{Packages.tex}
\input{technical.tex}

\title{Boundary structure of the standard model coupled to gravity\thanks{A. S. C. acknowledges partial support of SNF Grant No. 200020 192080 and of the Simons Collaboration
on Global Categorical Symmetries. G. C. acknowledges partial support of SNF Grant No P500PT 203085. This
research was (partly) supported by the NCCR SwissMAP, funded by the Swiss National Science Foundation.}}
\author[1]{Giovanni Canepa}
\author[2]{Alberto S. Cattaneo}
\author[3]{Filippo Fila-Robattino}
\author[4]{Manuel Tecchiolli}

\renewcommand\footnotemark{}

\affil[1]{Universit\"at Wien, \href{mailto:giovanni.canepa.math@gmail.com}{giovanni.canepa.math@gmail.com}}
\affil[2,3,4]{Institut f\"ur Mathematik, Universit\"at Z\"urich }
\affil[2]{\href{mailto:cattaneo@math.uzh.ch}{cattaneo@math.uzh.ch}}
\affil[4]{\href{mailto:manuel.tecchiolli@math.uzh.ch}{manuel.tecchiolli@math.uzh.ch}}
\affil[3]{Scuola Internazionale Superiore di Studi Avanzati, Trieste, \href{mailto:ffilarob@sissa.it}{ffilarob@sissa.it}}

\date{}

\begin{document}
\maketitle

\begin{abstract}
In this article a description of the reduced phase space of the standard model coupled to gravity  is given. For space or time-like boundaries this is achieved as the reduction of a symplectic space with respect to a coisotropic submanifold and with the BFV formalism. For light-like boundaries the reduced phase space is described as the reduction of a symplectic manifold with respect to a set of constraints. Some results about the Poisson brackets of sums of functionals are also proved.
\end{abstract}

\tableofcontents
\section{Introduction}
The goal of this paper is to provide a description of the Reduced Phase Space (RPS) of the standard model coupled to gravity. We will describe it in two different ways, one through the reduction of a symplectic manifold by a coisotropic submanifold, and the other using the cohomological approach given by the BFV formalism---after Batalin--Fradkin--Vilkowisky \cite{BV3,BV2} (only in the case of space or time-like boundaries).

If we consider a globally hyperbolic space-time and a Cauchy surface $\Sigma$ on it, then the reduced phase space of a theory describes the set of possible initial conditions on $\Sigma$. In other words, not all possible $n$-uples of fields on a Cauchy surface produce, under evolution in time, a solution of the field equations on the space-time, but only a subset of them satisfying some conditions. The RPS is exactly this subset and it is usually described as a quotient. 

In more general cases we can still consider a space-time with boundary, not necessarily space-like, and consider the RPS on the boundary. The importance of this space is related to the possibility of analyzing in this way the structures of the theory at null-infinity (see for example  \cite{Sachs1962,Penrose1980, Torre1985,PS2017} and references therein) and to the problem of locality after quantization. Indeed, in order to encode cutting and gluing in the BV formalism, we need a description of the boundary fields given by the BFV formalism \cite{CMR2012,CMR2, Schaetz:2008}. The BV formalism---after Batalin--Vilkowisky\cite{BV1}--- is a necessary tool to quantize---using the path integral quantization---a gauge theory, as it is the case for General Relativity and the standard model. 
Furthermore the advantages of describing the RPS as a BFV theory are also related to the possibility of using the AKSZ construction \cite{AKSZ} in order to construct a BV theory on a cylindrical manifold, thus highlighting the true structure of the symmetries of the theory. This construction has already been applied to gravity in \cite{CCS2020b}. A BFV formulation also allows to consider manifolds with corners and induce on the corners some data containing information about the structure of observables around punctures of the boundary (see \cite{CC2023} for gravity and \cite{OS2020} for the connection with similar works by the quantum gravity community).

The structure and methods of this paper follow those of  \cite{CCS2020,CCF2022, CCT2020, CFHT23}. In particular we combine here some of the results of the aforementioned papers, about the RPS of gravity and gravity combined with a scalar field, a Yang--Mills field and a spinor field, in order to get a description of the standard model by adding the corresponding interaction terms. 

As in the previous articles we use the coframe formulation of General Relativity \cite{thi2007} (sometimes called Palatini--Cartan or Einstein--Cartan) and the Kijowski and Tulczijew (KT) construction \cite{KT1979} which produces a symplectic space, called \emph{geometric phase space}, and a set of constraints.

Namely, one starts by taking the variation of the action of the theory and from the variation we can distinguish between the Euler--Lagrange equations and a boundary term stemming from possible integration by parts. It is then possible to reinterpret the boundary term as a one form on the space of the restriction of the fields to the boundary. By taking the variation of this one form we then obtain a closed two form. If this form is non-degenerate, it will constitute the symplectic form of the geometric phase space (formed by the restriction of the fields to the boundary). Otherwise, if the two form is degenerate but its kernel is regular, it is possible to obtain the geometric phase space as the quotient by the kernel of the aforementioned degenerate two-form. 
Subsequently starting from the Euler--Lagrange equations one can restrict them to the boundary and distinguish between evolution equations (containing  derivatives in the direction normal to the boundary) and constraints. We then look for some structural constraint on the geometric phase space so that these constraints can be expressed in the new boundary variables (i.e. after the reduction described above).
Then the RPS is defined as the reduction of the geometric phase space with respect to these constraints.

If the constraints are of first class, i.e. the Poisson brackets of the constraints are proportional to 
themselves (or zero), they define a coisotropic submanifold. If the boundary metric is non-degenerate (i.e. a space-like or time-like boundary) it was proven in \cite{CCS2020} that the constraints of gravity do actually define a coisotropic submanifold. Later it was shown in \cite{CCF2022} that this is also the case when some matter or gauge fields are added. On the other hand, note that when the boundary metric is degenerate, i.e. we are considering a light-like or null boundary, the sets of constraints is no longer coisotropic \cite{CCT2020, CFHT23}.
In this note we show that this two-fold behaviour (depending on the degeneracy of the boundary metric) happens also if we combine all the fields and the interactions needed to form the coupling of the standard model and gravity. This is proved by making a careful use of the known results and the properties of the Poisson brackets.
The BFV formalism can only be written for first class constraints and in this case it is possible to recover cohomologically the RPS, as explained with more detail in Section \ref{s:BFV_all}.

In this article we focus on the problem in dimension 4. However this can be generalized for dimensions $N>4$ with some extra care. For gravity the generalization was introduced in \cite{CCS2020} and the results are similar to the case $N=4$. We do expect a similar behaviour also for the gravity coupled with fields and the standard model.

\subsection{Overview of the results}
The main result of this paper is the description of the reduced phase space of the standard model coupled to gravity. However, in order to exploit some of the previous results about gravity coupled to different gauge and matter fields, we first develop some results in Poisson geometry. In particular we prove some formulas about the Poisson brackets of sum of functions in terms of the Poisson brackets of the single addends. Namely, in Lemma \ref{lem:Poisson_Brackets_singles} we prove a result for two addends, in Theorem \ref{thm:Poisson_Brackets_couples} for three addends and in Theorem \ref{thm:Poisson_brackets_four} for four addends. Furthermore we also prove an useful result about the Classical Master equation of the sum of two actions, in some particular conditions (Theorem \ref{thm:BFV_computation}).

On the more physical point of view the main result of this article is the description of the Reduced Phase Space of the standard model coupled to gravity. For a space or time-like boundary this has been achieved in two ways, one more classical, as a reduction of a symplectic space with respect to a coisotropic submanifold (defined from a set of constraints), the second as a cohomology class of a differential graded operator using the BFV formalism.
In the light-like case only the description through the geometric phase space and the constraints can be achieved.

Let us start from the case of non-degenerate boundary, i.e. for a boundary such that the induced metric is space or time-like. 
The reduced phase space of the standard model coupled to gravity is described as follows. Using the KT construction, we first find  the geometric phase space of gravity, which turns out to be composed of the spaces of a coframe $e$ and a connection $\omega$, satisfying  the additional structural constraint \eqref{e:structural_constraint_grav}. For more details about this space we refer to the description in \cite{CCS2020}. To the space of gravity fields we then have to add the spaces corresponding to the bosons and fermions of the standard model together with their conjugate momenta. In particular we have the spaces of the following fields:
\begin{enumerate}
    \item A scalar field multiplet corresponding to the Higgs boson $\phi$ and the momentum $\Pi$, satisfying the structural constraint \eqref{e:structural_Higgs}.
    \item An $SU(3)\times SU(2) \times U(1)$ gauge field $A$ and its momentum $B$, satisfying the structural constraint \eqref{e:constraintYM}.
    \item A left handed $\psi_L$ and a right handed $\psi_R$ multiplet of spinors and their conjugate fields $\overline{\psi_L}$ and $\overline{\psi_R}$. Note that the presence of the spinors modifies the pure gravity structural constraint \eqref{e:structural_constraint_grav} into a corrected one \eqref{e:omegareprfix2spin}.
\end{enumerate}

These space of fields form a symplectic space with symplectic form
\begin{equation*}
        \Omega_{SM}=\varpi+\varpi_A+\varpi_{\psi_L}+\varpi_{\psi_R}+\varpi_H,
    \end{equation*}
where the single terms are defined as in  table \ref{tab:symplectic}.
    \bgroup
    \def\arraystretch{1.5}%
    \begin{table}[ht!]
    \setlength\tabcolsep{4pt}
        \centering
        \begin{tabular}{|c c | c c |}
    \hline
        $\varpi$ & \eqref{e:sympl_form_grav} &
        $\varpi_H$  & \eqref{e:sympl_form_H}\\
       $\varpi_A$  & \eqref{e:sympl_form_YM}&
       $\varpi_{\psi_R}$/$\varpi_{\psi_R}$ & \eqref{e:sympl_form_spin}\\
       \hline
    \end{tabular}
        \caption{Terms building up the symplectic form of the the standard model coupled to gravity and their defining equations.}
        \label{tab:symplectic}
    \end{table}
    \egroup

More details about the coupling of gravity together with each one of these fields  can be found in \cite{CCF2022}.

On this geometric phase space we can define the constraints
\begin{align*}
        L_c^{SM}&=L_c+l_c^\psi, &  P_\xi^{SM}&=P_\xi + p_\xi^A + p_\xi^\psi + p_\xi^H + p_\xi^{A,\psi} + p_\xi^{H,A}, \\
        M_\mu^{SM}&= M_\mu^A + m_\mu^{H,A} + m_\mu^{A,\psi}, & H_\lambda^{SM}&=H_\lambda + h_\lambda^A + h_\lambda^\psi + h_\lambda^H + h_\lambda^{A,\psi} + h_\lambda^{H,A} +h_\lambda^{H,\psi}
    \end{align*}
    where the single terms are defined as in  table \ref{tab:constraints}.

    \bgroup
    \def\arraystretch{1.5}%
    \begin{table}[!ht]
    \setlength\tabcolsep{4pt}
        \centering
        \begin{tabular}{|c|c c | c c | c c|c c|c c|c c|c c|}
    \hline
        $L_c^{SM}$ &
       $L_c$  & \eqref{e:constraint_gravity1}&
       $l_c^H$ & 0 
       &$l_c^A$ & 0 
       & $l_c^\psi$ & \eqref{e:constraints_spinor1} 
       & $l_c^{H,\psi}$ & 0
       & $l_c^{H,A}$ & 0
       & $l_c^{A,\psi}$ & 0\\
       $P_\xi^{SM}$ &
       $P_\xi$  & \eqref{e:constraint_gravity2}
       & $p_\xi^H$ & \eqref{e:constraint_H2}
       &$p_\xi^A$ & \eqref{e:constraints_YM2} 
       & $p_\xi^\psi$ & \eqref{e:constraints_spinor2}
       & $p_\xi^{H,\psi}$ & 0
       & $p_\xi^{H,A}$ & \eqref{e:constraint_YMH2}
       & $p_\xi^{A,\psi}$ & \eqref{e:constraints_YMS2}\\
       $H_\lambda^{SM}$ &
       $H_{\lambda}$  &\eqref{e:constraint_gravity3} 
       & $h_\lambda^H$ & \eqref{e:constraint_H3}
       &$h_\lambda^A$ & \eqref{e:constraints_YM3} 
       & $h_\lambda^\psi$ & \eqref{e:constraints_spinor3}
       & $h_\lambda^{H,\psi}$ & \eqref{e:constraints_Yukawa3}
       & $h_\lambda^{H,A}$ & \eqref{e:constraint_YMH3}
       & $h_\lambda^{A,\psi}$ & \eqref{e:constraints_YMS3}\\
       $M_\mu^{SM}$&
       & &  & &$M_\mu^A$ & \eqref{e:constraints_YM4}  
       &  & &  & 
       & $m_\mu^{H,A}$ & \eqref{e:constraint_YMH4}
       & $m_\mu^{A,\psi}$ & \eqref{e:constraints_YMS4}\\
       \hline
    \end{tabular}
        \caption{Terms building up the constraints in the standard model coupled to gravity and their defining equations.}
        \label{tab:constraints}
    \end{table}
    \egroup

These constraints are roughly speaking given by the restriction of the Euler--Lagrange equation of the Lagrangian of the standard model, together with gravity.

The main result of this paper than states that for a non-degenerate boundary metric this set of constraints are of first class and then define a coisotropic submanifold. All the details and the precise expression of the Poisson brackets of the constraints are collected in Theorem \ref{thm:first-class-constraints_SM}. The Reduced Phase Space is then given by the coisotropic reduction of the Geometric Phase Space with respect to the zero set of these constraints. 

Alternatively, resorting to the BFV formalism, we can express the reduced phase space with an isomorphism with the zeroth cohomology class of the differential graded vector field $Q^{SM}$, the Hamiltonian vector field of the BFV action $\mathcal{S}^{SM}$. The precise details of this approach are collected in Theorem \ref{thm:BFVaction_SM}.

In order to prove the aforementioned result we worked constructively by proving similar results for gravity coupled with all the possible pairs of gauge and matter fields considered for the standard model.

When dealing with a light-like boundary the results are different. Indeed, using the KT construction for gravity alone, we find a geometric phase space still composed by the spaces of the coframe $e$ and the connection $\omega$, but with different structural constraints, given by \eqref{e:structural_constraint_grav_deg} and \eqref{e:structural_constraint2_grav_deg}. For more details about this space we refer to \cite{CCT2020}. In order to get the geometric phase space of the standard model coupled to gravity, we then have to add the spaces of the fields. This goes as in the non-degenerate case written above, and in presence of spinors the pure gravitational constraint \eqref{e:structural_constraint_gravity_degenerate} is modified into a corrected one \eqref{e:structural_constraint_spinor_deg}. More details can be found in \cite{CFHT23}. On the geometric phase space one then defines the constraints $ L_c^{SM},  P_\xi^{SM},H_\lambda^{SM}$ and $M_\mu^{SM}$ as above and the additional constraint
\begin{align*}
    R_{\tau}^{SM}= R_{\tau} + r_{\tau}^{\psi}.
\end{align*}
defined as in table \ref{tab:constraints_deg}.
    \bgroup
    \def\arraystretch{1.5}%
    \begin{table}[!ht]
    \setlength\tabcolsep{4pt}
        \centering
        \begin{tabular}{|c|c c | c c | c c|c c|c c|c c|c c|}
    \hline
        $R_{\tau}^{SM}$ &
       $R_{\tau}$  & \eqref{e:constraint_R}&
       $r_{\tau}^H$ & 0 
       &$r_{\tau}^A$ & 0 
       & $r_{\tau}^\psi$ & \eqref{e:constraint_R_psi} 
       & $r_{\tau}^{H,\psi}$ & 0
       & $r_{\tau}^{H,A}$ & 0
       & $r_{\tau}^{A,\psi}$ & 0\\
       \hline
    \end{tabular}
        \caption{Terms building up the additional constraint  in the degenerate case and their defining equations.}
        \label{tab:constraints_deg}
    \end{table}
    \egroup

This set of constraints, in contrast with the non-degenerate case is now not of first class, as proved in Theorem \ref{thm:no-first-class-constraints_SM}). This is a consequence of the geometry of the light-like boundary (and mathematically of the degeneracy of the metric).
As a consequence of this fact, in the degenerate case, providing a BFV resolution of the quotient becomes unfeasible.

\subsection{How to read this paper}\label{s:howtoread}
In this section we introduce the notation used in the paper and its logical structure. 

Let us begin with the notation used for describing the geometric phase spaces of the theories considered here.
All the theories are constructed starting from gravity and then adding the corresponding matter or gauge fields. The geometric phase space of gravity is a symplectic space denoted by $(F^{\partial}, \varpi)$, i.e. with a roman $F^{\partial}$ for the space and $\varpi$ for the corresponding symplectic form.\footnote{Note that we use $\varpi$ instead of the more common $\omega$ for symplectic forms in order not to confuse it with the connection $\omega$ of the gravity coframe formulation.} 
The reduced phase space is then defined through the use of constraints that will be denoted by roman capitalized letters with an index denoting the corresponding (odd) Lagrange multiplier. 

Then we add the gauge or matter fields. These theories will be denoted using letters out of the following table:
\begin{table}[!ht]
\centering
\begin{tabular}{|c|c|}
\hline
    Additional field & Index \\
\hline
    Scalar field & $\phi$  \\
    Higgs field (multiplet of scalar fields) & $H$ \\
    Yang--Mills gauge field & $A$ \\
    Spinor field & $\psi$\\
\hline
\end{tabular}
\end{table}

These letters will be used as apexes as follows: the geometric phase space of a theory with field  with index $f$ will be a direct sum $F^{\partial}\oplus F_f$ with symplectic form $\Omega_f=\varpi+\varpi_f$. The constraints defining the RPS are then denoted as $X^f= X + x^f$ where $X$ is the corresponding constraint for gravity, $x^f$ is the additional part and $X^f$ is the whole constraint for the coupled theory. In case we have two or more matter or gauge fields, we simply stack indices and add additional terms corresponding to the interactions. As an example we can consider the following constraint (from Section \ref{s:YM+spinor}): 

\noindent
\begin{center}
\begin{tikzpicture}[framed]
        \node (All) at (4.5, 1) {$H^{A, \psi}_{\lambda}$};
        \node (ug) at (5.5, 1) {$=$};
        \node (grav) at (6.5, 1) {$H_{\lambda}$};
        \node (p1) at (7, 1) {$+$};
        \node (YM) at (7.5, 1) {$h^{A}_{\lambda}$};
        \node (p2) at (8, 1) {$+$};
        \node (spin) at (8.5, 1) {$h^{\psi}_{\lambda}$};
        \node (p3) at (9, 1) {$+$};
        \node (YMS) at (9.7, 1) {$h^{A, \psi}_{\lambda}$};
        \draw[->] (5,2.4) node [anchor=south]{Constraint of theory with a Yang--Mills field and a spinor field} to[out=-135, in=135] (All);
        \draw[->] (4.5,0.2) node [anchor=east]{$\lambda$ is the Lagrange multiplier} to[out=0, in=-90] (All);
        \draw[->] (7.5,2) node [anchor=south]{Pure gravity constraint} to[out=-90, in=90] (grav);
        \draw[->] (10.5,1.6) node [anchor=south]{Additional part containing YM fields} to[out=-135, in=45] (YM);
        \draw[->] (7.5,0.2) node [anchor=north]{Additional part containing spinor fields} to[out=45, in=-90] (spin);
        \draw[->] (11,-0.2) node [anchor=north]{Additional part containing interaction} to[out=90, in=-45] (YMS);
\end{tikzpicture}

\end{center}

In Section \ref{s:appendix_Poisson_brackets}, where we present more general results not connected to a specific theory or constraint, we do not indicate the Lagrange multiplier associated to the constraint and we use the letters $X$ and $Y$ to denote generic constraints. The Hamiltonian vector fields will be denoted with blackboard bold letters and indices corresponding to the matter/gauge fields as defined in \eqref{e:Hamiltonian_VF_single} and \eqref{e:Hamiltonian_VF_interaction}.

For the BFV formalism we use the same conventions but here the BFV space of fields will be denoted by a calligraphic $\mathcal{F}$  and the actions with a calligraphic $\mathcal{S}$ with the appropriate indices. The symplectic form of the BFV space will be denoted by $\Omega^{BFV}$ with the same convention on the indices as before. Further technical notation is introduced in Section \ref{s:BFV_brackets_theorem} and in Appendix \ref{a:def_maps_and_spaces}, however this is not necessary for the general understanding of the main results of this paper but only for the proofs.

\subsection{Structure of the paper}

This paper is structured as follows. In Section \ref{s:appendix_Poisson_brackets} we show and prove some useful identities for the Poisson brackets of sums of functionals. Furthermore we also prove here a result for the BFV brackets.

In Section \ref{s:previousResults} we recall the previous results about the reduced phase space of gravity (Section \ref{s:previousResults_gravity}) and of gravity coupled with a scalar field (Section \ref{s:previousResults_scalar}), with a Yang--Mills field (Section \ref{s:YM}) and a spinor field (Section \ref{s:spinor}).

Section \ref{s:interaction_terms} is the fundamental building block of the paper in which we describe the reduced phase space of the theories composed of gravity and pairs of gauge and matter field. Respectively we have gravity together with a Yang--Mills field and a spinor field in Section \ref{s:YM+spinor}, then we consider the Higgs and Yang--Mills fields in Section \ref{s:YM+scalar} and lastly a spinor field together with a scalar field, forming the Yukawa interaction in \ref{s:scalar+spinor}.

Finally in Section \ref{s:standard_model} we describe the reduced phase space of gravity coupled to the standard model as a coisotropic reduction. 

Section \ref{s:BFV_all} is dedicated to the BFV formulations of all the theory listed above. 

In Section \ref{s:degenerate_section} we consider the same theories but for a boundary with induced degenerate metric. In this case as well we divide the work starting from recalling the previous results (Sections \ref{s:degenerate_gravity}, \ref{s:degenerate_scalar}, \ref{s:degenerate_YM} and \ref{s:degenerate_spinor}), considering the single interactions (Sections \ref{s:degenerate_YMS}, \ref{s:degenerate_YMH} and \ref{s:degenerate_yuk}) and describing the standard model  (Section \ref{s:degenerate_SM}).

\section{Properties of the Poisson brackets}
\label{s:appendix_Poisson_brackets}
The goal of this section is to show how to compute the Poisson brackets of particular sum of constraints of which we already know some partial results.

\subsection{Sum of two terms}
As a warm up let us start from constraints composed of two terms. Let $(F, \varpi)$ be a symplectic space and $F_f , \varpi_f$ a space and a two form such that $$(F\oplus F_f, \Omega_{f}=\varpi+\varpi_f)$$ is still symplectic. In our case $(F, \varpi)$ is the geometric phase space of gravity while the index $f$ denotes the spaces of some matter fields. 

Let us suppose that we have two constraints of the following form:
\begin{align*}
    X^{f}= X + x^f, \qquad Y^{f}= Y + y^f
\end{align*}
where $X,Y \in C^{\infty}(F)$ and $x^f,y^f \in C^{\infty}(F\oplus F_f)$.

The first step towards the computation of the brackets is to find the Hamiltonian vector fields of the functions $X^{f}$ and $Y^{f}$, i.e. vector fields satisfying 
\begin{align*}
    \iota_{\mathbb{X}^{f}} \Omega_{f} = \delta X^{f}, \qquad \iota_{\mathbb{Y}^{f}} \Omega_{f} = \delta Y^{f}, 
\end{align*}
Let us concentrate on $X$, for $Y$ is exactly the same.
Let us denote by $\mathbb{X}$ the Hamiltonian vector field of $X$ with respect to $\varpi$:
\begin{align*}
    \iota_{\mathbb{X}} \varpi = \delta X.
\end{align*}
Then, if we look for an Hamiltonian vector field of the form
\begin{align*}
    \mathbb{X}^{f} = \mathbb{X} + \mathbb{x}^{f}
\end{align*}
we get that $\mathbb{x}^{f}$ must satisfy
\begin{align*}
    \iota_{\mathbb{X}}\varpi_{f}+ \iota_{\mathbb{x}^{f}} (\varpi+\varpi_{f}) = \delta x^{f}.
\end{align*}
Then we have the following result:
\begin{lemma}\label{lem:Poisson_Brackets_singles}
Using the notation introduced above, the Poisson brackets of the functions $X^{f}$ and $Y^{f}$ with respect to the symplectic form $\Omega_{SM}$ are given by the following expression:
\begin{align*}
    \{X^{f},Y^{f}\}_{f}&=\{X, Y\} + \iota_{\mathbb{y}^{f}}\delta (X+ x^{f}) + \iota_{\mathbb{x}^{f}}\delta (Y+ y^{f})- \iota_{\mathbb{x}^{f}}\iota_{\mathbb{y}^{f}}\Omega_{f}+\iota_{\mathbb{X}}\iota_{\mathbb{Y}}\varpi_{f}
\end{align*}
where $\{\bullet, \bullet\}_f$ means that we are considering the Poisson brackets with respect to $\varpi+\varpi_f$ and 
$\{\bullet, \bullet\}$ means that we are considering the Poisson brackets with respect to $\varpi$.
\end{lemma}
\begin{proof}
    If we expand the Poisson brackets $\{X^{f},Y^{f}\}_{f}$ we get eight terms as follows:
    \begin{align*}
        \iota_{\mathbb{X}+\mathbb{x}^{f}}\iota_{\mathbb{Y}+\mathbb{y}^{f}}(\varpi+\varpi_f)&=  
        \iota_{\mathbb{X}}\iota_{\mathbb{Y}}\varpi +\iota_{\mathbb{X}}\iota_{\mathbb{y}^{f}}\varpi
        +\iota_{\mathbb{x}^{f}}\iota_{\mathbb{Y}}\varpi
        +\iota_{\mathbb{x}^{f}}\iota_{\mathbb{y}^{f}}\varpi\\
        &\phantom{=}
        +\iota_{\mathbb{X}}\iota_{\mathbb{Y}}\varpi_f +\iota_{\mathbb{X}}\iota_{\mathbb{y}^{f}}\varpi_f
        +\iota_{\mathbb{x}^{f}}\iota_{\mathbb{Y}}\varpi_f
        +\iota_{\mathbb{x}^{f}}\iota_{\mathbb{y}^{f}}\varpi_f.
    \end{align*}
    It is then straightforward to note that the first term in this expression corresponds to $\{X, Y\}$; the sum of the  second, fourth, sixth and eighth corresponds to $\iota_{\mathbb{y}^{f}}\delta (X+ x^{f})$; the sum of the  third, fourth, seventh and eighth corresponds to $+ \iota_{\mathbb{x}^{f}}\delta (Y+ y^{f})$. Then we have to subtract the eighth since we summed up it two times and to add the fifth which was left.
\end{proof}

This lemma was tacitly used in \cite{CCF2022} and will be generalized in the next section.

\subsection{Sum of three terms}
Let us now consider constraints composed of the sum of three terms, the pure gravity one and two matter/gauge fields. As before, let $(F, \varpi)$, $(F\oplus F_f , \varpi+\varpi_f)$, $(F\oplus F_g, \varpi+\varpi_g)$  be symplectic spaces such that $$(F\oplus F_f\oplus F_g , \Omega_{fg}=\varpi+\varpi_f+\varpi_g)$$ is still symplectic.
Let us suppose that we have two constraints of the following form:
\begin{align*}
    X^{f,g}= X + x^f + x^g + x^{f,g}, \qquad Y^{f,g}= Y + y^f + y^g + y^{f,g}
\end{align*}
where $X,Y \in C^{\infty}(F)$, $x^f,y^f \in C^{\infty}(F\oplus F_f)$, $x^g,y^g\in C^{\infty}(F\oplus F_g)$ and $x^{f,g},y^{f,g}\in C^{\infty}(F\oplus F_f\oplus F_g)$. The Hamiltonian vector fields of these constraints will satisfy
\begin{align*}
    \iota_{\mathbb{X}^{f, g}} \Omega_{fg} = \delta X^{f, g}, \qquad \iota_{\mathbb{Y}^{f, g}} \Omega_{fg} = \delta Y^{f, g}, 
\end{align*}
Let us concentrate on $X$, and look for a vector field of the form 
$  \mathbb{X}^{f, g} = \mathbb{X} + \mathbb{x}^{f}+\mathbb{x}^{g}+\mathbb{x}^{f, g}$
where $\mathbb{X},\mathbb{x}^{f},\mathbb{x}^{g}$ are such that
\begin{align}\label{e:Hamiltonian_VF_single}
    \iota_{\mathbb{X}} \varpi = \delta X, \quad 
    \iota_{\mathbb{X}}\varpi_{g}+ \iota_{\mathbb{x}^{g}} (\varpi+\varpi_{g}) = \delta x^{g}, \quad
    \iota_{\mathbb{X}}\varpi_{f}+ \iota_{\mathbb{x}^{f}} (\varpi+\varpi_{f}) = \delta x^{f}
\end{align}
as we have seen in the previous section.
We conclude that $\mathbb{x}^{f, g}$ satisfies
\begin{align}\label{e:Hamiltonian_VF_interaction}
    \iota_{\mathbb{x}^{f, g}} (\varpi+ \varpi_{g}+ \varpi_{f}) + \iota_{\mathbb{x}^{g}} \varpi_{f}+ \iota_{\mathbb{x}^{f}} \varpi_{g} = \delta x^{f, g}.
\end{align}

Having found the Hamiltonian vector fields of the two functions $X^{f,g}$ and $Y^{f,g}$, we have the following results:

\begin{theorem}\label{thm:Poisson_Brackets_couples}
Using the notation introduced above, the Poisson bracket of the functions $X^{f,g}$ and $Y^{f,g}$ with respect to the symplectic form $\Omega_{fg}$ are given by the following expression:
\begin{align*}
    \{X^{f, g},Y^{f, g}\}_{fg}&=\{X+x^{f}, Y+y^{f}\}_f+\{X+x^{g}, Y+y^{g}\}_g-\{X,Y\}\\
    &+ \iota_{\mathbb{y}^{f, g}}\delta (X+ x^{f}+x^{g}+x^{f, g}) + \iota_{\mathbb{x}^{f, g}}\delta (Y+ y^{f}+y^{g}+y^{f, g})- \iota_{\mathbb{x}^{f, g}}\iota_{\mathbb{y}^{f, g}}\Omega_{fg}\\
    &+ \iota_{\mathbb{X}+\mathbb{x}^{g}}\iota_{\mathbb{Y}+\mathbb{y}^{g}}\varpi_{f}+ \iota_{\mathbb{X}+\mathbb{x}^{f}}\iota_{\mathbb{Y}+\mathbb{y}^{f}}\varpi_{g} +\iota_{\mathbb{x}^{f}}\iota_{\mathbb{y}^{g}}\Omega_{fg}+\iota_{\mathbb{x}^{g}}\iota_{\mathbb{y}^{f}}\Omega_{fg} -\iota_{\mathbb{X}}\iota_{\mathbb{Y}}(\varpi_{f}+\varpi_{g}) 
\end{align*}
where $\{\bullet, \bullet\}_f$ means that we are considering the Poisson brackets with respect to $\varpi+\varpi_f$,
$\{\bullet, \bullet\}_g$ means that we are considering the Poisson brackets with respect to $\varpi+\varpi_g$,
$\{\bullet, \bullet\}$ means that we are considering the Poisson brackets with respect to $\varpi$,
and
$\{\bullet, \bullet\}_{fg}$ means that we are considering the Poisson brackets with respect to $\Omega_{fg}$.
\end{theorem}

\begin{proof}
We recall that the Poisson brackets of two functions can be expressed in terms of their Hamiltonian vector field and the symplectic form as
\begin{align*}
    \{X^{f, g},Y^{f, g}\}_{fg}&= \iota_{\mathbb{X}^{f, g}}\iota_{\mathbb{Y}^{f, g}} \Omega_{fg}\\
    &=\iota_{(\mathbb{X}+\mathbb{x}^{f}+\mathbb{x}^{g}+\mathbb{x}^{f, g})}\iota_{(\mathbb{Y}+\mathbb{y}^{f}+\mathbb{y}^{g}+\mathbb{y}^{f, g})}(\varpi+ \varpi_f+ \varpi_g).
\end{align*}
    As we have done in the proof of Lemma  \ref{lem:Poisson_Brackets_singles}, it is sufficient to use linearity in each of the sums in order to get 48 terms that can be rearranged as in the statement.
\end{proof}

This result is particularly useful, since in our applications we already know the terms appearing on the first row of the RHS and some of the other terms will vanish trivially.

\subsection{Sum of four terms}
\label{s: brackets triple_int}
    Similarly we can consider what happens when we compute the brackets of constraints with three matter/gauge fields. We assume that there is no triple interaction as it is the case for the standard model. Using the same notation as in the previous section, we want to compute the brackets between 
    \begin{align*}
        X^{f,g,h}&= X + x^f + x^g +x^h + x^{f,g}+ x^{f,h}+ x^{g,h} \\ 
        Y^{f,g,h}&= Y + y^f + y^g +y^h + y^{f,g}+ y^{f,h}+ y^{g,h}
    \end{align*}
    with respect to a symplectic form 
    \begin{align*}
        \Omega_{fgh}= \varpi + \varpi_f +  \varpi_g + \varpi_h.
    \end{align*}
    The first step is to compute the Hamiltonian vector field of $X^{f,g,h}$ with respect to $\Omega_{fgh}$. Call such a vector field $\mathbb{X}^{f,g,h}= \mathbb{X} + \mathbb{x}^{f}+\mathbb{x}^{g}+\mathbb{x}^{h}+\mathbb{x}^{f,g}+\mathbb{x}^{f,h}+\mathbb{x}^{g,h}+\mathbb{x}^{f,g,h}$ where the first seven summands satisfy \eqref{e:Hamiltonian_VF_single} and \eqref{e:Hamiltonian_VF_interaction}. Then it can be shown that $\mathbb{x}^{f,g,h}$ must satisfy 
    \begin{align*}
    \iota_{\mathbb{x}^{f,g,h}}\Omega_{fgh}= 
    - \iota_{\mathbb{x}^{f,g}}\varpi_{h}
    - \iota_{\mathbb{x}^{g,h}}\varpi_{f}
    - \iota_{\mathbb{x}^{f,h}}\varpi_{g}
    \end{align*}
    The following result will simplify the computation of the constraints.
    \begin{theorem}\label{thm:Poisson_brackets_four}
        In the setting above, suppose that $\iota_{\mathbb{x}^{f,g}}\varpi_{h}=\iota_{\mathbb{x}^{g,h}}\varpi_{f}= \iota_{\mathbb{x}^{f,h}}\varpi_{g}=0$. Then we have the following formula:
        \begin{align*}
    \{X^{f, g, h},Y^{f, g, h}\}_{fgh}&=\{X^{f, g},Y^{f, g}\}_{fg}+\{X^{f, h},Y^{f, h}\}_{fh}+\{X^{g, h},Y^{g, h}\}_{gh} \\
    & \phantom{=} - \{X^f, Y^f\}_f -\{X^g, Y^g\}_g -\{X^h, Y^h\}_h+\{X,Y\}\\
    & \phantom{=} + \sum^{cycl}_{ijk=fgh}  (\iota_{\mathbb{x}^{i,j}}\iota_{\mathbb{y}^{k}}+\iota_{\mathbb{y}^{i,j}}\iota_{\mathbb{x}^{k}})(\varpi + \varpi_i + \varpi_j)\\
    &\phantom{=} + \sum^{cycl}_{ijk=fgh} (\iota_{\mathbb{x}^{i,j}}\iota_{\mathbb{y}^{i,k}}+\iota_{\mathbb{y}^{i,j}}\iota_{\mathbb{x}^{i,k}})(\varpi + \varpi_i) + \sum^{cycl}_{ijk=fgh} (\iota_{\mathbb{x}^{i}}\iota_{\mathbb{y}^{j}}+\iota_{\mathbb{y}^{i}}\iota_{\mathbb{x}^{j}})\varpi_k
\end{align*}
where the sums are cyclic sums over the indices $f,g,h$.
    \end{theorem}
\begin{proof}
    In order to prove the theorem we can resort to Theorem \ref{thm:Poisson_Brackets_couples} by defining $\varpi_a= \varpi_f +  \varpi_g$, $\varpi_b=\varpi_h$, $x^a= x^f + x^g + x^{f,g}$, $x^b=x^h$ and $x^{a,b}= x^{f,h}+ x^{g,h}$.
    Using the hypothesis, we get that $\mathbb{x}^{f,g,h}=0$ and hence $\mathbb{x}^{a,b}=\mathbb{x}^{g,h}+\mathbb{x}^{f,h}$. Then the formula can be proved by carefully expanding the term in the provided formula and summing the equal terms.
\end{proof}

\subsection{Results for the BFV brackets}\label{s:BFV_brackets_theorem}

Let us now prove another useful technical result for the proofs about BFV structures. 
Suppose that we have a BFV structure $(\mathcal{F}, \mathcal{S}, \varpi_{BFV})$ and let $\mathcal{S}= S_0+ S_1$ where $S_0= \sum_i X_i$ is the sum of the constraints of the theory, $S_1$ is the part of the BFV action linear in the antifields.
Let also $\varpi_{BFV}= \varpi_c + \varpi_b$ where
and $\varpi_c$ is the classical boundary symplectic form and $\varpi_b$ depends only on ghosts and antifields.

\begin{theorem}\label{thm:BFV_computation}
Let $(\mathcal{F}, \mathcal{S}, \varpi_{BFV})$ a BFV structure and let us denote $S_0$ the part of $\mathcal{S}$ linear in the ghost fields and $S_1$ the part linear in the antifields. Let $(\widetilde{\mathcal{F}}, \widetilde{\varpi})=(\mathcal{F}\oplus\mathcal{F}_2, \varpi+\varpi_2)$ be a symplectic space, $s$ be some functional on $\widetilde{\mathcal{F}}$ linear in the ghosts and $Q=Q_0+Q_1+q_0+q_1$ be a vector field such that
\begin{align*}
    \iota_{Q_0}\varpi_{BFV}&= \delta S_0, &  \iota_{Q_1}\varpi_{BFV}&= \delta S_1,\\ 
    \iota_{q_0}(\varpi_{BFV}+\varpi_2)&= \delta s - \iota_{Q_0} \varpi_2, & \iota_{q_1}(\varpi_{BFV}+\varpi_2)&= - \iota_{Q_1} \varpi_2.
\end{align*}
 Then $\{\mathcal{S}+s, \mathcal{S}+s\}=0$ if
    \begin{align}\label{e:computation_BFV1}
 \iota_{Q_0+q_0}\iota_{Q_0+q_0}(\varpi_c+\varpi_2)+2\iota_{Q_0+q_0}\iota_{Q_1+q_1}\varpi_b=0,\\
  \iota_{Q_1}\iota_{q_0}\varpi_c=0. \label{e:computation_BFV2}
 \end{align}
\end{theorem}

Note that the last equation defining $q_1$ can be simplified as $\iota_{q_1}\varpi_2= - \iota_{Q_1} \varpi_2$ and $\iota_{q_1}(\varpi_c+\varpi_b)=0.$\footnote{Note that this is a possible solution of the equation defining $q_1$, and since $\varpi_{BFV}+ \varpi_2$ is symplectic, this is also the unique solution to this equation.}
\begin{proof}
From the definition of Poisson bracket, we have
\begin{align*}
\{\mathcal{S},\mathcal{S}\}=\iota_{Q_0}\iota_{Q_0}\varpi_c+2\iota_{Q_0}\iota_{Q_1}\varpi_c+\iota_{Q_0}\iota_{Q_0}\varpi_b+2\iota_{Q_0}\iota_{Q_1}\varpi_b+\iota_{Q_1}\iota_{Q_1}\varpi_c+\iota_{Q_1}\iota_{Q_1}\varpi_b=0.
\end{align*}
Note that we automatically have $\iota_{Q_0}\iota_{Q_0}\varpi_b=0$ since $Q_0$ has no antifields terms.
Now, the equation $\{\mathcal{S},\mathcal{S}\}=0$ can be split in three smaller ones dividing it by the number of antifields appearing in each term:
  \begin{align*}
  \iota_{Q_0}\iota_{Q_0}\varpi_c+2\iota_{Q_0}\iota_{Q_1}\varpi_b=0.\\
    2\iota_{Q_0}\iota_{Q_1}\varpi_c+\iota_{Q_1}\iota_{Q_1}\varpi_b=0,\\
\iota_{Q_1}\iota_{Q_1}\varpi_c=0
\end{align*}
respectively for terms with no antifields, linear in antifields and quadratic in the antifields.

Let us now expand $\{\mathcal{S}+s, \mathcal{S}+s\}:$
 \begin{align*}
     \{\mathcal{S}+s, \mathcal{S}+s\}&=\iota_{Q_0+q_0}\iota_{Q_0+q_0}(\varpi_c+\varpi_2)+2\iota_{Q_0+q_0}\iota_{Q_1+q_1}(\varpi_c+\varpi_2)+\iota_{Q_0+q_0}\iota_{Q_0+q_0}\varpi_b\\
     &\phantom{=}+2\iota_{Q_0+q_0}\iota_{Q_1+q_1}\varpi_b+\iota_{Q_1+q_1}\iota_{Q_1+q_1}(\varpi_c+\varpi_2)+\iota_{Q_1+q_1}\iota_{Q_1+q_1}\varpi_b.
\end{align*}
Using the hypothesis \eqref{e:computation_BFV1}
and $\iota_{q_1}(\varpi_c+\varpi_b)=0$, from the definition of $q_1$, we get 
 \begin{align*}     
     \{\mathcal{S}+s, \mathcal{S}+s\}&=2\iota_{Q_0+q_0}\iota_{Q_1+q_1}(\varpi_c+\varpi_2)+\iota_{Q_0+q_0}\iota_{Q_0+q_0}\varpi_b\\
     &\phantom{=}+\iota_{Q_1+q_1}\iota_{Q_1+q_1}\varpi_2+\iota_{Q_1}\iota_{Q_1}\varpi_b
 \end{align*}

Once again we have $\iota_{Q_0+q_0}\iota_{Q_0+q_0}\varpi_b=0$ since both $Q_0$ and $q_0$ have only nonzero components in the direction of the antifields. Furthermore we have 
\begin{align*}
    \iota_{Q_1+q_1}\iota_{Q_1+q_1}\varpi_2= 0
\end{align*}
since $\iota_{Q_1+q_1}\varpi_2=0$. Finally, using again these relations we get
\begin{align*}
    2\iota_{Q_0+q_0}\iota_{Q_1+q_1}(\varpi_c+\varpi_2)+\iota_{Q_1}\iota_{Q_1}\varpi_b&= 2 \iota_{Q_0}\iota_{Q_1+q_1}\varpi_c +2\iota_{Q_0}\iota_{Q_1+q_1}\varpi_2\\
    &\phantom{=}+ 2 \iota_{q_0}\iota_{Q_1+q_1}\varpi_c +2\iota_{q_0}\iota_{Q_1+q_1}\varpi_2 +\iota_{Q_1}\iota_{Q_1}\varpi_b\\
    &= 2 \iota_{Q_0}\iota_{Q_1}\varpi_c + 2 \iota_{q_0}\iota_{Q_1}\varpi_c +\iota_{Q_1}\iota_{Q_1}\varpi_b\\
    &=  2 \iota_{q_0}\iota_{Q_1}\varpi_c.
\end{align*}
Hence we conclude using \eqref{e:computation_BFV2}.
\end{proof}

\begin{remark}\label{rem:BFVandbracketsofconstraints}
    This result will be particularly useful in computing the Classical Master Equation starting from a known BFV action and adding some terms corresponding to the additional parts in the constraints. Indeed, as we will see, the interaction parts of the constraints do not modify the brackets of the constraints themselves (see Theorems \ref{thm:first-class-constraints_YMS}, \ref{thm:first-class-constraints_Higgs} and \ref{thm:first-class-constraints_Yukawa}) and this particular condition is exactly encoded in \eqref{e:computation_BFV1}. Hence we will have only to check \eqref{e:computation_BFV2} which will be almost straightforward.
\end{remark}

\section{Previous results}\label{s:previousResults}
The description of the reduced phase space of the standard model coupled to gravity is based on some building blocks, namely the description of the RPS of Gravity in the vacuum and the RPS of gravity together with a scalar field, a Yang--Mills field and a spinor field respectively. In this section we recall the results contained in the articles \cite{CS2019, CCS2020, CCF2022}. We start with the description of gravity in vacuum.

\subsection{Gravity}\label{s:previousResults_gravity}
In this article we use the coframe formulation of General Relativity. In this formulation, instead of using the metric,
the main fields of the theory are a coframe field and a principal connection. We refer to \cite{T2019} for more detail about this theory and we present here only a summary.

The geometrical set-up is the following:
\begin{itemize}
    \item[--] An $4$-dimensional differentiable oriented pseudo-Riemannian manifold $M$ with boundary $\Sigma$;
    \item[--] A principal GL$(4,\mathbb R)$-bundle $LM$ called the \emph{frame bundle}, which can be reduced to a principal SO$(3,1)$-bundle $P$;
    \item[--] An associated vector bundle $\mathcal V\coloneqq P\times_\rho V$ called the \emph{Minkowski bundle}, where $V$ is an $4$-dimensional real pseudo-Riemannian vector space with reference metric $\eta=$diag$(1,...,-1)$ and $\rho\colon$SO$(3,1)\to$Aut$(V)$ is the fundamental representation of SO$(3,1)$.
\end{itemize}

A non-canonical identification exists between the fibers of the tangent bundle and those of the vector bundle $\mathcal V$ equipped with the reference metric $\eta$. This identification is referred to as the \emph{vielbein} or coframe, defined as the vector bundle isomorphism $e\colon TM\to\mathcal V$, which can be represented as an element of $\Omega^1_{nd}(M,\mathcal V)$, i.e. a one form on $M$ with values in $\mathcal{V}$, where the subscript $nd$ signifies the non-degenerate nature of the isomorphism. Moreover we can recover the space-time metric as $g=e^*\eta$.
\begin{remark}
    In order to shorten the notation, we will write $\Omega^{i,j}\coloneqq\Omega^i(M,\wedge^j\mathcal V)$ and $\Omega^{i,j}_\partial\coloneqq\Omega^i(\Sigma,\wedge^j\mathcal V_\Sigma)$, where the last object represents bundle-valued differential forms and will be properly defined later.
\end{remark}
The fields of the theory are therefore the coframe $e\in \Omega^1_{nd}(M,\mathcal V)$ and  the pull-back of the principal connection $\omega\in \mathcal A(P)$ via sections of $P$.\\
The action functional of the theory is 
    \begin{align}\label{e:actionPCT}
    S=\int_{M}\frac{1}{2}e^{2} F_\omega+\frac1{4!}\Lambda e^4,
    \end{align}
and its  Euler--Lagrange equations are
\label{eq:el}
\begin{align}
e F_\omega-\frac1{3!}\Lambda e^{3}&=0\label{eq:PCe}\\
e d_\omega e&= 0,\label{eq:PComega}
\end{align}
where we have omitted the wedge product.\footnote{Note that the quantities appearing in \eqref{e:actionPCT} form, up to multiplication by $\sqrt{|det(\eta)|}$ which is omitted for ease of notation, densities that can be integrated. For more details see \cite[footnotes 5 and 6]{CCS2020}.} 
Finally, the object $\Lambda\in\Omega^{0,0}$ is the so called \emph{cosmological constant}.

The fields $(e,\omega)\in\Omega^{1,1}_{nd}\times \mathcal A(P)$ and the action functional $S_{PC}$ define what we call the Palatini--Cartan theory.

The boundary structure of the theory is obtained by means of the KT construction. We recall here only the results and we refer to \cite{CCS2020} for all the details.
The output of the KT construction is a symplectic space, the geometric phase space, built out of the boundary fields, which are obtained as a pull-back of the bulk fields $(e,\omega)$ via the inclusion map $\iota\colon\Sigma\to M$, where $\Sigma$ is the boundary of $M$. We denote the restrictions of the fields to the boundary $\Sigma$ in the very same way of the bulk fields and denote by $\mathcal{V}_\Sigma$ the restriction $\mathcal{V}|_\Sigma$. Assuming $\mathcal{V}_\Sigma$ to be isomorphic to $T\Sigma\oplus\underline{\mathbb{R}}$, we fix a nowhere vanishing section $\epsilon_n$ of the summand $\underline{\mathbb{R}}$ and define $\Omega^1_{e_n}(\Sigma,\mathcal V_{\Sigma})$ to be the space of bundle maps $e\colon T\Sigma\to\mathcal{V}_\Sigma$ such that $eee e_n\neq 0$ everywhere.\footnote{Note that this is an equivalent condition to that the three components of $e$ together with $e_n$ form a basis of $\mathcal{V}_\Sigma$. Note as well that in \cite{CCS2020} $\Omega_{e_n}^1(\Sigma, \mathcal{V}_\Sigma)$ was denoted $\Omega_\text{nd}^1(\Sigma, \mathcal{V}_\Sigma)$.} The boundary action can be recovered by $g^\partial\coloneqq\iota^*g$, i.e. by pulling back along the inclusion $\iota$ the bulk metric. Equivalently we have $g^\partial = \eta (e,e)$.
The geometric phase space for gravity in the PC formulation for a space or time-like boundary is given by the bundle
\begin{align*}
    F^{\partial} \rightarrow \Omega^1_{e_n}(\Sigma,\mathcal{V}_{\Sigma}) 
\end{align*}
with fiber $\mathcal{A}_{red}(\Sigma)$ where the fields $\omega \in \mathcal{A}_{red}(\Sigma)$ satisfy the structural constraint
\begin{align}\label{e:structural_constraint_grav}
    e_n d_{\omega} e = e \sigma
\end{align}
for some $\sigma\in \Omega^{1}(\Sigma, \mathcal{V}_{\Sigma})$.
The corresponding symplectic form on  $F^{\partial} $ is given by 
\begin{align}\label{e:sympl_form_grav}
    \varpi = \int_{\Sigma} e \delta e \delta \omega .
\end{align}

For a light-like boundary the description is slightly different and will be recalled in Section \ref{s:degenerate_gravity}. From now on, unless otherwise stated we consider only space or time-like boundaries.

On this geometric phase space one can then define some constraints. Let us define
\begin{align}
L_c &= \int_{\Sigma} c e d_{\omega} e \label{e:constraint_gravity1}\\
P_{\xi} &= \int_{\Sigma}  \iota_{\xi} e e F_{\omega} + \iota_{\xi} (\omega-\omega_0) e d_{\omega} e \label{e:constraint_gravity2}\\
H_{\lambda} &= \int_{\Sigma} \lambda e_n \left(eF_\omega + \frac{1}{3!}\Lambda e^3\right) \label{e:constraint_gravity3}
\end{align}
where $c \in\Omega^{0,2}_\partial[1]$, $\xi \in\mathfrak{X}[1](\Sigma)$ and $\lambda\in \Omega^{0,0}_\partial[1]$ are (odd) Lagrange multipliers and the notation $[1]$ denotes that the fields are shifted by 1 in parity and are treated as odd variables.\footnote{Note that this does not mean necessarily that their total parity is odd.}

This set of constraints defines a coisotropic submanifold, as stated by the following theorem.
\begin{theorem}[\cite{CCS2020}] \label{thm:first-class-constraints}
 Let $g^\partial$ be non-degenerate on $\Sigma$. Then, the functionals  $L_c$, $P_{\xi}$,  $H_{\lambda}$ are well defined on ${F}^{\partial}$ and define a coisotropic submanifold  with respect to the symplectic structure $\varpi$. In particular they satisfy the following relations
 \begin{align*}
     \{L_c, L_c\} &= - \frac{1}{2} L_{[c,c]}  & \{L_c, P_{\xi}\}  &=  L_{\mathrm{L}_{\xi}^{\omega_0}c}\\
     \{L_c,  H_{\lambda}\}  &= - P_{X^{(\nu)}} + L_{X^{(\nu)}(\omega - \omega_0)_{\nu}} - H_{X^{(n)}} & \{P_{\xi}, P_{\xi}\}  &=  \frac{1}{2}P_{[\xi, \xi]}- \frac{1}{2}L_{\iota_{\xi}\iota_{\xi}F_{\omega_0}}\\
     \{P_{\xi},H_{\lambda}\}  &=  P_{Y^{(\nu)}} -L_{ Y^{(\nu)} (\omega - \omega_0)_{\nu}} +H_{ Y^{(n)}} & \{H_{\lambda},H_{\lambda}\}  &=0
 \end{align*}
where $X= [c, \lambda e_n ]$, $Y = \mathrm{L}_{\xi}^{\omega_0} (\lambda e_n)$ and $Z^{(\nu)}$, $Z^{(n)}$ are the components of $Z=X,Y$ in the basis $(e_{\nu}, e_n)$.	
\end{theorem}
Hence the RPS of gravity in the PC formulation is given by the reduction of the geometric phase space ${F}^{\partial}$ with respect to the coisotropic submanifold defined by the zero set of the constraints \eqref{e:constraint_gravity1}, \eqref{e:constraint_gravity2} and \eqref{e:constraint_gravity3}.

As announced in the introduction, we can also describe the RPS cohomologically using the BFV formalism. See Section \ref{s:BFV_all} for more details about this approach. 
The corresponding BFV theory for the PC formulation of General Relativity is recalled in the following theorem:
\begin{theorem}[\cite{CCS2020}]\label{thm:BFVaction}
Let $g^\partial$ be non-degenerate on $\Sigma$. Let $\mathcal{F}$ be the bundle
\begin{equation}
\mathcal{F} \longrightarrow \Omega_{e_n}^1(\Sigma, \mathcal{V}_{\Sigma}),
\end{equation}
with local trivialisation on an open $\mathcal{U}_{\Sigma} \subset \Omega_{e_n}^1(\Sigma, \mathcal{V}_{\Sigma})$
\begin{equation}\label{LoctrivF1}
\mathcal{F}\simeq \mathcal{U}_{\Sigma} \times \mathcal{A}_{red}(\Sigma) \oplus T^* \left(\Omega_{\partial}^{0,2}[1]\oplus \mathfrak{X}[1](\Sigma) \oplus C^\infty[1](\Sigma)\right) =: \mathcal{U}_{\Sigma} \times \mathcal{T}_{grav},
\end{equation}
where $e \in \mathcal{U}_{\Sigma}$ and $\omega \in \mathcal{A}_{red}(\Sigma)$ satisfy the modified structural constraint
\begin{align}\label{e:modifiedStructuralConstraint}
    e_n d_{\omega} e + \left(L_\xi^\omega \epsilon_n - [c,\epsilon_n]\right)^{(a)}  c^\dag_a = e \sigma_{BFV}
\end{align}
for some $\sigma_{BFV}\in \Omega^{1}(\Sigma, \mathcal{V}_{\Sigma})$.
Further denoting the ghost fields $c \in\Omega_{\partial}^{0,2}[1]$, $\xi \in\mathfrak{X}[1](\Sigma)$ and $\lambda\in \Omega^{0,0}[1]$ in degree one, $c^\dag\in\Omega_{\partial}^{3,2}[-1]$, $\lambda^\dag\in\Omega_{\partial}^{3,4}[-1]$ and $\xi^\dag\in\Omega_\partial^{1,0}[-1]\otimes\Omega_{\partial}^{3,4}$ in degree minus one, we 
define a symplectic form and an action functional on $\mathcal{F}$ respectively by

\begin{align}
\Omega^{BFV} & = \varpi + \varpi_{ghost} \label{symplectic_form_NC1} \\
\mathcal{S} & = L_c + P_{\xi}+ H_{\lambda} + \mathcal{S}_{ghost}. \label{action_NC1}
\end{align}
where 
\begin{align*}
    \varpi_{ghost}= \int_{\Sigma} & \delta c \delta c^\dag + \delta \lambda \delta \lambda^\dag + \iota_{\delta \xi} \delta \xi^\dag \\
    \mathcal{S}_{ghost}= \int_{\Sigma} &  \frac{1}{2} [c,c] c^{\dag} - \mathrm{L}_{\xi}^{\omega_0} c c^{\dag} +\frac{1}{2}\iota_{\xi}\iota_{\xi}F_{\omega_0}c^{\dag}- \mathrm{L}_{\xi}^{\omega_0} (\lambda e_n)^{(n)}\lambda^\dag - \frac{1}{2}\iota_{[\xi,\xi]}\xi^{\dag} \nonumber\\ 
    & + [c, \lambda e_n ]^{(\nu)}(\xi_\nu^{\dag}- (\omega - \omega_0)_{\nu} c^\dag) + [c, \lambda e_n ]^{(n)}\lambda^\dag - \mathrm{L}_{\xi}^{\omega_0} (\lambda e_n)^{(\nu)}(\xi_{\nu}^{\dag}- (\omega - \omega_0)_{\nu} c^\dag). 
\end{align*}
Then the triple $(\mathcal{F}, \Omega^{BFV}, \mathcal{S})$ defines a BFV structure on $\Sigma$.
\end{theorem}

Then the space of functions on the reduced phase space is isomorphic to the cohomology in degree 0 of the Hamiltonian vector field $\mathcal{Q}$ of $\mathcal{S}$.

\subsection{Scalar field}\label{s:previousResults_scalar}
Let us now consider gravity together with a scalar field.
The action of the scalar field coupled with gravity takes the form
	\begin{align*}
	    S_{scal}= S + S_{\phi}
	\end{align*}
 where
 \begin{align}\label{e:action_scalar}
     S_{\phi}&=\int_M \frac{1}{3!}e^{3}\Pi d\phi +  \frac{1}{2\cdot 4!}e^4 (\Pi,\Pi)
 \end{align}
 where $(\bullet,\bullet)$ stands for the pairing $\eta$ in $\mathcal{V}$.
The structure of the reduced phase space of this case was discussed in \cite{CCF2022} to which we refer for more details and insight.

In this case, the geometric phase space is the bundle
\begin{align*}
    {F}^{\partial}_{\phi} \rightarrow \Omega_{e_n}^1(\Sigma, \mathcal{V}_{\Sigma})\oplus \mathcal{C^{\infty}}(\Sigma)
\end{align*}
with fiber $\mathcal{A}_{red}(\Sigma)\oplus\Omega^{(0,1)}_{\partial,\text{red}}$ such that, using the same notation as before and furthermore denoting $\phi \in \mathcal{C^{\infty}}(\Sigma)$, $\Pi \in\Omega^{(0,1)}_{\partial,\text{red}}$, the structural constraints \eqref{e:structural_constraint_grav} and 
	 \begin{align}\label{e:constraintscalar}
	     (e,\Pi)=-d\phi.
	 \end{align}
  are satisfied.
  ${F}^{\partial}_{\phi}$ is a symplectic space with symplectic form
  $\Omega_{\phi}= \varpi+\varpi_{\phi}$ 
  where\footnote{It is possible to redefine $p:=\frac{1}{3!}e^3 \Pi$. The three components of $\Pi$ fixed by the structural constraint \eqref{e:constraintscalar} are in kernel of $e^3 \wedge \cdot$ acting on $\Omega^{(0,1)}_\partial $. With this substitution, the symplectic form $\varpi_\phi$ is non-degenerate, therefore there is no need for reduction. However, in some instances it is still convenient to use $\Pi$ as it provides simpler computations.}
  \begin{align}\label{e:sympl_form_scal}
      \varpi_{\phi}= \int_{\Sigma}\frac{1}{3!}\delta(e^3\Pi)\delta\phi.
  \end{align}

  Define now on this space the following functions:
    \begin{align}
		&p^{\phi}_\xi=-\int_\Sigma \frac{1}{3!}e^3\Pi\mathrm{L}_{\xi}\phi;\label{e:constraint_scalar2}\\
		&h^{\phi}_\lambda=\int_\Sigma \lambda e_n \left( \frac{1}{2}e^{2}\Pi d\phi + \frac{1}{2\cdot 3!}e^{3}(\Pi,\Pi) \right) .\label{e:constraint_scalar3}
	\end{align} 
 Then, it has been proven in \cite{CCF2022} that the functions
 $L_c$, $P^{\phi}_{\xi}=P_{\xi}+p^{\phi}_\xi$, $H^{\phi}_{\lambda}=H_{\lambda}+h^{\phi}_\lambda$ define a coisotropic submanifold. In particular the brackets of these functions are verbatim the same as in the vacuum case (Theorem \ref{thm:first-class-constraints}). Hence it is possible to describe the RPS as the quotient of the geometric phase space ${F}_{\phi}$ with respect to this coisotropic submanifold. As in the vacuum case the resolution of such quotient is given by the following BFV structure:

\begin{theorem}[\cite{CCF2022}]\label{thm:BFVaction_scalar}
		Let $\mathcal{F}_{\phi}$ be the bundle 
		\begin{equation*}
\mathcal{F}_{\phi} \longrightarrow \Omega_{e_n}^1(\Sigma, \mathcal{V}_{\Sigma})\oplus \mathcal{C^{\infty}}(\Sigma),
\end{equation*}
with local trivialisation on an open $\mathcal{U}_{\Sigma}$
		\begin{equation*}
			\mathcal{F}_{\phi}\simeq \mathcal{U}_{\Sigma} \times \mathcal{T}_{grav}\times \Omega^{(0,1)}_{\partial,\text{red}}
		\end{equation*}
		where $\mathcal{T}_{grav}$ was defined in \eqref{LoctrivF1} and the additional fields are denoted by $\Pi\in\Omega^{(0,1)}_{\partial,\text{red}}$ and $\phi\in\mathcal{C^{\infty}}(\Sigma)$ and such that  they satisfy the structural constraint \eqref{e:constraintscalar}. 
		The symplectic form and the action functional on $\mathcal{F}_S$ are respectively defined by
		\begin{align*}
			\Omega^{BFV}_{\phi} &= \Omega^{BFV}+\varpi_{\phi}, \\
			\mathcal{S}_{\phi} &= \mathcal{S} + p^{\phi}_\xi + h^{\phi}_\lambda.
		\end{align*}
		Then the triple $(\mathcal{F}_{\phi}, \Omega^{BFV}_{\phi}, S_{\phi})$ defines a BFV structure on $\Sigma$.
	\end{theorem}

\subsection{Yang--Mills} \label{s:YM}
Similarly to the scalar field case, for the theory coupling gravity and a Yang--Mills field (with group $G$ and Lie algebra $\mathfrak{g}$), the action is 
	\begin{align*}
	    S_{YM}= S + S_{A}
	\end{align*}
 where
 \begin{align}\label{e:action_YM}
     S_{A}&=\int_M \frac{1}{2}e^{2}\mathrm{Tr}(B F_A) + \frac{1}{2\cdot 4!}e^4 \mathrm{Tr}(B,B),
		\end{align}
	where $(\bullet,\bullet)$ is the canonical pairing in $\wedge^2 V$ defined in coordinates for all $C,D\in \wedge^2 V$ by $(C,D):=C^{ab}D^{cd}\eta_{ac}\eta_{bd}$ with respect to an orthonormal basis $\{u_a\}$ of $V$.

The geometric phase space is the bundle 
\begin{align*}
    F^{\partial}_{A} \rightarrow \Omega_{e_n}^1(\Sigma, \mathcal{V}_{\Sigma})\oplus \mathcal{A}^{G}_{\Sigma}
\end{align*}
with fiber $\mathcal{A}_{red}(\Sigma)\oplus \Omega^{(0,2)}_{\Sigma, red}(\mathfrak{g})$ such that, denoting $A \in \mathcal{A}^{\text{YM}}_{\Sigma}$ and  $B\in\Omega^{(0,2)}_{\Sigma, red}(\mathfrak{g})$, the structural constraints \eqref{e:structural_constraint_grav} and 
 \begin{align}\label{e:constraintYM}
	     F_A + \frac{1}{2}(e^2,B)=0
	 \end{align}
  are satisfied.
  The symplectic form on ${F}^{\partial}_{A}$ reads $\Omega_{A}= \varpi + \varpi_{A}$ where
  \begin{align}\label{e:sympl_form_YM}
      \varpi_{A}= \int_{\Sigma}\Tr\left( \frac{1}{2}\delta (eeB) \delta A\right)
  \end{align}
where the trace is defined with respect to the Lie algebra $\mathfrak{g}$.

In order to define the reduced phase space, we introduce the following functions on   ${F}^{\partial}_{A}$:
	\begin{align}
		& M^{A}_\mu:=\int_{\Sigma}\frac{1}{2}\Tr(\mu d_A (e^2B)) \label{e:constraints_YM4}; \\
		& 	p^{A}_\xi:=\int_{\Sigma}  \frac{1}{2}\iota_\xi e^2 \Tr(BF_A) + \frac{1}{2}\Tr\{\iota_{\xi}(A-A_0)d_A(e^2B)\} \label{e:constraints_YM2};\\
		& h^{A}_\lambda:=\int_{\Sigma} \lambda e_n \left(e \Tr(B F_A) + \frac{1}{2\cdot 3!} e^3 \Tr(B,B)\right) \label{e:constraints_YM3}.
	\end{align}

 Then we have the following result:

 \begin{theorem}{\cite{CCF2022}} \label{thm:first-class-constraints_YM}
		The functions $M^{A}_\mu$, $L_c$, $P^{A}_{\xi}=P_{\xi}+p^{A}_{\xi} $,  $H^{A}_{\lambda}=H_{\lambda}+h^{A}_{\lambda}$ define a coisotropic submanifold  with respect to the symplectic structure $\Omega_{A}$. In particular they have the following mutual Poisson brackets:
		
		\begin{align*}\label{brackets-of-constraints_YM}
			\{M^{A}_\mu,M^{A}_\mu\}_A&=-\frac{1}{2}M^{A}_{[\mu,\mu]}	 & \{M^{A}_\mu,L_c\}_A&=0\\
			\{M^{A}_\mu,P^{A}_\xi\}_A&=M^{A}_{\mathrm{L}_\xi^{A_0}\mu} & \{M^{A}_\mu,H^{A}_\lambda\}_A&= 0\\
            \{P^{A}_{\xi}, P^{A}_{\xi}\}_A  &=  \frac{1}{2}P^{A}_{[\xi, \xi]}- \frac{1}{2}L_{\iota_{\xi}\iota_{\xi}F_{\omega_0}}-\frac{1}{2}M^{A}_{\iota_{\xi}\iota_{\xi}F_{A_0}} & \{L_c, P^{A}_{\xi}\}_A  &=  L_{\mathrm{L}_{\xi}^{\omega_0}c}\\
            \{L_c,  H^{A}_{\lambda}\}_A & = - P^{A}_{X^{(\nu)}} + L_{X^{(\nu)}(\omega - \omega_0)_{\nu}} -H^{A}_{X^{(n)}} + M^{A}_{X^{(\nu)}(A-A_0)_{\nu}} {} & \{L_c, L_c\}_A &= - \frac{1}{2}L_{[c,c]}
             \\
			\{P^{A}_{\xi},H^{A}_{\lambda}\}_A  &=  P^{A}_{Y^{(\nu)}} -L_{ Y^{(\nu)} (\omega - \omega_0)_{\nu}} +  H^{A}_{ Y^{(n)}} - M^{A}_{Y^{(\nu)}(A-A_0)_{\nu}}   	
             & \{H^{A}_\lambda,H^{A}_\lambda\}_A&=0 		
		\end{align*}
	where $X$ and $Y$ are defined as in Theorem \ref{thm:first-class-constraints}.
	\end{theorem}
 The corresponding BFV structure resolving the quotient is given by the following theorem:

\begin{theorem}\label{thm:BFVactionYM}
	The BFV space of field $\mathcal{F}_{A}$ is the bundle
				\begin{equation*}
\mathcal{F}_{A} \longrightarrow \Omega_{e_n}^1(\Sigma, \mathcal{V}_{\Sigma})\oplus \mathcal{A}^{G}_{\Sigma}
\end{equation*}
with local trivialisation on an open $\mathcal{U}^{YM}_{\Sigma} \subset \Omega_{e_n}^1(\Sigma, \mathcal{V}_{\Sigma})\oplus \mathcal{A}^{G}_{\Sigma}$
		\begin{equation}
			\mathcal{F}_{\mathrm{YM}}\simeq \mathcal{U}^{YM}_{\Sigma} \times \mathcal{T}_{grav} \times \Omega^{(0,2)}_{\Sigma, red}(\mathfrak{g})\times T^*(\Gamma[1](\Sigma,\mathfrak{g})),
		\end{equation}
		where $\mathcal{T}_{grav}$ is the \emph{gravity fiber} defined in \eqref{LoctrivF1} and we denote by  $\mu\in \Gamma[1](\mathfrak{g})$ and its antifield by $\mu^\dagger \in \Gamma[-1](\wedge^3 T^*\Sigma \otimes \wedge^4 \mathcal{V}_{\Sigma}\otimes \mathfrak{g})$. Furthermore, $B\in\Omega^{(0,2)}_{\Sigma, red}(\mathfrak{g})$ satisfies \eqref{e:constraintYM}.
        If we define
        \begin{align*}
            \mathcal{S}^{A}_{ghost}=\int_{\Sigma} \Tr\left\{ \frac{1}{2} [\mu,\mu]\mu^\dagger - 	\mathrm{L}_\xi^{A_0}(\mu)\mu^\dagger + \frac{1}{2}\iota_\xi\iota_\xi F_{A_0}\mu^\dagger +  \left[ \mathrm{L}_\xi^{\omega_0}(\lambda e_n) ^{(\nu)}-[c,\lambda e_n]^{(\nu)}\right](A-A_0)_{\nu}\mu^\dagger \right\}
        \end{align*}
        The corresponding action functional and a symplectic form on $\mathcal{F}_{A}$ are
			\begin{align}	
				\mathcal{S}_{A}&= \mathcal{S} + p^{A}_{\xi}+ h^{A}_{\lambda}+ M^{A}_{\mu}+\mathcal{S}^{A}_{ghost}\\
				\Omega^{BFV}_A &=\Omega^{BFV} + \varpi_A  + \varpi^{A}_{ghost}\label{symplectic_form_YM} 
			\end{align}
        where $\varpi^{A}_{ghost}= \int_{\Sigma} \Tr(\delta \mu \delta \mu^\dagger).$ 
		Then the triple $(\mathcal{F}_A, \Omega^{BFV}_A, \mathcal{S}_A)$ defines a BFV structure on $\Sigma$.
  
	\end{theorem}

\subsection{Spinor}\label{s:spinor}
The precise setting in which we can couple gravity with a spinor field is described in \cite[Section 5]{CCF2022}.
We recall here just the bulk action and the  boundary structure.
The action of the spinor field coupled with gravity takes the form
	\begin{align*}
	    S_{spin}= S + S_{\psi}
	\end{align*}
 where
\begin{equation}\label{e:action_spin}
	S_{\psi} =\int_{M} i\frac{e^{3}}{2\cdot 3!}\left[ \overline{\psi} \gamma d_\omega \psi - d_\omega \overline{ \psi}\gamma \psi \right],
\end{equation}

The geometric phase space is the bundle 
\begin{align*}
    {F}^{\partial}_{\psi} \rightarrow \Omega_{e_n}^1(\Sigma, \mathcal{V}_{\Sigma})\times S(\Sigma)\times \overline{S}(\Sigma)
\end{align*}
where $S(\Sigma):=\Gamma(\Sigma,E_{\lambda}\vert_\Sigma)$, \footnote{$E_\lambda$ is defined to be the associated bundle to $\hat{P}$, $E_\lambda:=\hat{P}\times_{\lambda} W$
	where $W$ is an $N$-dimensional complex vector space and $\lambda\colon\mathrm{Spin}(N-1,1)\times W\rightarrow W$ is a non-tensorial representation of the spin group on $W$.}
with fiber $\mathcal{A}_{red}(\Sigma)$ such that
\begin{equation}\label{e:omegareprfix2spin}
		e_n\left(  d_{{\omega}}e - \frac{i}{16}\overline{\psi}\left( j_{\gamma}j_{\gamma} e^2 \gamma + \gamma j_{\gamma}j_{\gamma} e^2 \right)\psi\right)= e \widetilde{\sigma}
	\end{equation}
 for some $\widetilde{\sigma} \in \Omega^{1}(\Sigma, \mathcal{V}_{\Sigma})$.
 The symplectic form on this space is given by $\Omega_{\psi}= \varpi + \varpi_{\psi}$ where 
 \begin{align}\label{e:sympl_form_spin}
     \varpi_{\psi}= \int_{\Sigma} i\frac{e^2}{4}\left( \overline{\psi}\gamma \delta \psi - \delta \overline{\psi}	\gamma \psi \right)\delta e + i \frac{e^3}{ 3!}  \delta \overline{\psi} \gamma \delta \psi.
 \end{align}

As in the previous examples the reduced phase space is recovered as the reduction with respect to the functions
\begin{align*}
    L^{\psi}_c= L_c + l^{\psi}_c, \quad P_{\xi}^{\psi}=P_{\xi}+p_{\xi}^{\psi}, \quad H_{\lambda}^{\psi}=H_{\lambda}+h_{\lambda}^{\psi}
\end{align*}
where
\begin{align}
    l^{\psi}_c &=  \int_{\Sigma} - i \frac{e^3}{2\cdot 3!} \left( [c,\overline{\psi}]\gamma \psi - \overline{\psi} \gamma [c,\psi] \right),\label{e:constraints_spinor1}\\
    p^{\psi}_\xi&=\int_{\Sigma}- i \frac{e^3}{2\cdot 3!} \left( \overline{\psi} \gamma \mathrm{L}_{\xi}^{\omega_0}(\psi) - \mathrm{L}_\xi^{\omega_0}(\overline{\psi})\gamma \psi  \right)\label{e:constraints_spinor2}\\
	 h^{\psi}_\lambda &= \int_{\Sigma} \lambda e_n\left[ i\frac{e^2}{4}\left( \overline{\psi}\gamma d_\omega \psi - d_\omega\overline{\psi}\gamma \psi \right) \right],\label{e:constraints_spinor3}
\end{align}
having defined $[\chi,\psi]:=\frac{1}{4}j_{\gamma} j_{\gamma} \chi \psi$ for $\chi \in \Omega^{\bullet}(\Sigma, \wedge^2 \mathcal{V}_{\Sigma})$.\footnote{
 	For any $X\in V$ and for all $\alpha\in \wedge^k V$, we define for $\alpha=\frac{1}{k!}\alpha^{i_1\cdots i_k}v_{i_1}\wedge\cdots\wedge v_{i_k}$, $
			j_X \alpha:=\frac{\eta_{ab}}{(k-1)!} X^a \alpha^{b {i_2}\cdots {i_k}} 	v_{i_2}\wedge \cdots \wedge v_{i_k}$.}
The functions $L^{\psi}_c$, $ P_{\xi}^{\psi}$ and $H_{\lambda}^{\psi}$ define a coisotropic submanifold \cite{CCF2022} and their Poisson brackets are the same as in Theorem \ref{thm:first-class-constraints} (after modifying the corresponding notation).

 The BFV resolution of this quotient is given by the set of data specified in the following theorem:
      \begin{theorem}\label{thm:BFVaction spinor}
		Let $\mathcal{F}_{\psi}$ be the bundle 
        \begin{align*}
            \mathcal{F}_{\psi}\longrightarrow \Omega_{e_n}^1(\Sigma, \mathcal{V}_{\Sigma}) \times S(\Sigma) \times \overline{S}(\Sigma)
        \end{align*}
        with local trivialisation on an open $\mathcal{U}_{\Sigma}$
\begin{equation*}
\mathcal{F}_{\psi} \simeq \mathcal{U}_{\Sigma} \times \mathcal{T}_{grav}.
\end{equation*}
Further define
		\begin{align*}
			\Omega^{BFV}_{\psi} &= \Omega^{BFV} + \varpi_{\psi} \\
			\mathcal{S}_{\psi} &= \mathcal{S} + l^{\psi}_{c} + p^{\psi}_{\xi}+ h^{\psi}_{\lambda}
		\end{align*}
		Then the triple $(\mathcal{F}_{\psi}, \Omega^{BFV}_{\psi}, \mathcal{S}_{\psi})$ defines a BFV structure on $\Sigma$.
	\end{theorem}

\section{Interaction terms}\label{s:interaction_terms}
Besides the coupling between gravity and the matter fields, in the standar model there are also interactions between the matter field themselves. Since multiple instances of the same interaction occur, it is useful to spell out the details of the interaction between possible couples of types of fields, keeping the discussion as general as possible. Later on, in Section \ref{s:standard_model}, we will combine all the results contained in this section and specify the type of scalar, Yang--Mills and spinor fields we are working on. The possible interactions are the following:
\begin{enumerate}
    \item A Yang--Mills field and a spinor field;
    \item A Yang--Mills field and a scalar field;
    \item A scalar field and a spinor field.
\end{enumerate}
We will treat these three interactions respectively in Sections \ref{s:YM+spinor}, \ref{s:YM+scalar} and \ref{s:scalar+spinor}.
For each of the aforementioned interaction we describe the reduced phase space using the KT construction, following the scheme already used for the single matter fields coupled to gravity. Namely, starting from the classical action on the bulk, we derive the EL equations and the pre-symplectic form on the boundary fields. Then, if necessary we perform the reduction and find the geometric phase space. Subsequently we find the functions describing the constraints and we check if they form a coisotropic submanifold. Then the reduced phase space of the theory is found to be the quotient of the geometric phase space with respect to the coisotropic submanifold defined by the constraints. 

In order to keep the construction simple, we will make use of the notation introduced in Section \ref{s:howtoread}. In particular all the relevant quantities will be presented as sums (or products) of the quantities introduced above plus an interaction or correction term.

\subsection{Yang--Mills-Spinor}\label{s:YM+spinor}

We consider in this section the interaction of a Yang--Mills field and a spinor field together with gravity.
We denote the spinor field by $\psi\in S(M)\otimes \mathfrak{su}(N)$ and the Yang--Mills field by $A$. The action on the bulk is given by the sum of the gravity part, the Yang--Mills part \eqref{e:action_YM}, the spinor part \eqref{e:action_spin} and an interaction part:
\begin{align*}
    S_{YMS} = S+ S_{\psi} + S_A + S_{A,\psi}
\end{align*}
where 
\begin{align*}
    S_{A,\psi}= \int_M &\frac{e^{N-1}}{2(N-1)!}\left( \overline{\psi} \gamma [A, \psi] - [A, \overline{\psi}] \gamma \psi \right).
\end{align*}

where $\overline{\psi} \gamma [A, \psi]= i g_i \overline{\psi}_{I} \gamma A^{I}_J \psi^J$ and $g_i$ is a coupling constant.
The interaction term does not contain derivatives, hence the boundary structure is just the direct sum of the YM structure and Spinor structures. In particular the geometric phase space is given by

\begin{align*}
    {F}^{\partial}_{YMS} \rightarrow \Omega_{e_n}^1(\Sigma, \mathcal{V}_{\Sigma})\oplus \mathcal{A}^{G}_{\Sigma} \times S(\Sigma)\times \overline{S}(\Sigma)
\end{align*}
with fiber $\mathcal{A}_{red}(\Sigma) \oplus \Omega^{(0,2)}_{\Sigma, red}(\mathfrak{g})$ such that  \eqref{e:constraintYM} and  \eqref{e:omegareprfix2spin} are satisfied. The symplectic form on this space reads 
\begin{align*}
    \Omega_{YMS}= \varpi + \varpi_A + \varpi_\psi .
\end{align*}

On this geometric phase space we can then define the following constraints:
\begin{align*}
    L^{A, \psi}_c &= L_c + l^{\psi}_{c}; \\
    P^{A, \psi}_{\xi} &= P_{\xi}+ p^{A}_{\xi} + p^{\psi}_{\xi}+p^{A, \psi}_{\xi}; \\
    H^{A, \psi}_{\lambda} &= H_{\lambda}+ h^{A}_{\lambda} + h^{\psi}_{\lambda}+ h^{A, \psi}_{\lambda};\\
    M^{A, \psi}_\mu &= M^{A}_\mu + m^{A, \psi}_\mu
\end{align*}
where
\begin{align}
    p^{A, \psi}_\xi&=\int_{\Sigma}- i \frac{e^3}{2\cdot 3!} \left( \overline{\psi} \gamma [\iota_{\xi}{A_0},\psi] - [\iota_\xi {A_0},\overline{\psi}]\gamma \psi  \right)\label{e:constraints_YMS2}\\
	h^{A, \psi}_\lambda & = \int_{\Sigma}- \lambda e_n\left[ i\frac{e^2}{4}\left( \overline{\psi}\gamma [A, \psi] - [A,\overline{\psi}]\gamma \psi \right) \right],\label{e:constraints_YMS3}\\
    m^{A, \psi}_\mu & =  \int_{\Sigma}  - i \frac{e^3}{2\cdot 3!} \left( [\mu,\overline{\psi}]\gamma \psi - \overline{\psi} \gamma [\mu,\psi] \right)\label{e:constraints_YMS4}.
\end{align}

\begin{remark}\label{r:properties_mu}
    Note that for every field $\mu$ with values in $\mathfrak{g}$ the following identities hold:
    \begin{align*}
        [\mu,\overline{\psi}\gamma \psi]&=0 \\
        [\mu,\overline{\psi}]\gamma \psi - \overline{\psi} \gamma [\mu,\psi] &= 2 [\mu,\overline{\psi}]\gamma \psi
    \end{align*}
\end{remark}
The following theorem proves that these constraints form a coisotropic submanifold.

\begin{theorem} \label{thm:first-class-constraints_YMS}
 Assume that $g^\partial$ is non-degenerate on $\Sigma$. Then, the zero locus of the functions  $L^{A, \psi}_c$, $P^{A, \psi}_{\xi}$,  $H^{A, \psi}_{\lambda}$ and $M^{A, \psi}_\mu$ is a coisotropic submanifold with respect to the symplectic structure $\Omega_{YMS}$. Their mutual Poisson brackets read
 {\footnotesize
 \begin{align*}\label{brackets-of-constraints_YMS}
			\{M^{A, \psi}_\mu,M^{A, \psi}_\mu\}_{YMS}&=-\frac{1}{2}M^{A, \psi}_{[\mu,\mu]}	 & \{M^{A, \psi}_\mu,L^{A, \psi}_c\}_{YMS}&=0\\
			\{M^{A, \psi}_\mu,P^{A, \psi}_\xi\}_{YMS}&=M^{A, \psi}_{\mathrm{L}_\xi^{A_0}\mu} & \{M^{A, \psi}_\mu,H^{A, \psi}_\lambda\}_{YMS}&= 0\\
            \{P^{A, \psi}_{\xi}, P^{A, \psi}_{\xi}\}_{YMS}  &=  \frac{1}{2}P^{A, \psi}_{[\xi, \xi]}- \frac{1}{2}L^{A, \psi}_{\iota_{\xi}\iota_{\xi}F_{\omega_0}}-\frac{1}{2}M^{A, \psi}_{\iota_{\xi}\iota_{\xi}F_{A_0}} & \{L^{A, \psi}_c, P^{A, \psi}_{\xi}\}_{YMS}  &=  L^{A, \psi}_{\mathrm{L}_{\xi}^{\omega_0}c}\\
            \{L^{A, \psi}_c,  H^{A, \psi}_{\lambda}\}_{YMS} & = - P^{A, \psi}_{X^{(\nu)}} + L^{A, \psi}_{X^{(\nu)}(\omega - \omega_0)_{\nu}} -H^{A, \psi}_{X^{(n)}} + M^{A, \psi}_{X^{(\nu)}(A-A_0)_{\nu}} {} & \{L^{A, \psi}_c, L^{A, \psi}_c\}_{YMS} &= - \frac{1}{2}L^{A, \psi}_{[c,c]}
             \\
			\{P^{A, \psi}_{\xi},H^{A, \psi}_{\lambda}\}_{YMS}  &=  P^{A, \psi}_{Y^{(\nu)}} -L^{A, \psi}_{ Y^{(\nu)} (\omega - \omega_0)_{\nu}} +  H^{A, \psi}_{ Y^{(n)}} - M^{A, \psi}_{Y^{(\nu)}(A-A_0)_{\nu}}   	
             & \{H^{A, \psi}_\lambda,H^{A, \psi}_\lambda\}_{YMS}&=0 		
		\end{align*}
  }
 with the same notation as in Theorem \ref{thm:first-class-constraints}. 
\end{theorem}

\begin{proof}
    We prove each bracket by using Theorem \ref{thm:Poisson_Brackets_couples}.  The first step is to compute the Hamiltonian vector fields of the constraints. Using the notation and the results of Section \ref{s:appendix_Poisson_brackets}, the expressions of $\mathbb{L},\mathbb{l}^{A}, \mathbb{l}^{\psi}, \mathbb{P}\mathbb{p}^{A}, \mathbb{p}^{\psi}, \mathbb{H},\mathbb{h}^{A}, \mathbb{h}^{\psi}$ and $ \mathbb{M}^{A}$ have been computed in \cite{CCF2022} and are collected in Appendix \ref{s:appendix_Hamiltonianvf}. Hence the only components that we have to compute through \eqref{e:Hamiltonian_VF_interaction} are $\mathbb{l}^{A, \psi}$, $\mathbb{p}^{A, \psi}$, $\mathbb{h}^{A, \psi}$ and  $\mathbb{m}^{A, \psi}$. Let us start from $\mathbb{l}^{A, \psi}$. It must satisfy
\begin{align*}
     \iota_{\mathbb{l}^{A, \psi}} (\varpi+ \varpi_{A}+ \varpi_{\psi}) + \iota_{\mathbb{l}^{\psi}} \varpi_{A}+ \iota_{\mathbb{l}^{A}} \varpi_{\psi} = 0.
\end{align*}
Since $\iota_{\mathbb{l}^{\psi}} \varpi_{A}= \iota_{\mathbb{l}^{A}} \varpi_{\psi} = 0$ we conclude $\mathbb{l}^{A, \psi}=0$. Similarly we have 
\begin{align*}
     \iota_{\mathbb{p}^{A, \psi}} (\varpi+ \varpi_{A}+ \varpi_{\psi}) + \iota_{\mathbb{p}^{\psi}} \varpi_{A}+ \iota_{\mathbb{p}^{A}} \varpi_{\psi} = \delta p^{A, \psi}_{\xi}.
\end{align*}
Since $\iota_{\mathbb{p}^{\psi}} \varpi_{A}= \iota_{\mathbb{p}^{A}} \varpi_{\psi}=0$, the computation is exactly the same as for $p^{\psi}_\xi$ with $A_0$ instead of $\omega_0$. 
    Hence we get
        \begin{align*}
	    \mathbb{p}^{A,\psi}_e &= 0 &
	    \mathbb{p}^{A,\psi}_\omega &= \mathbb{V}_{p^{A, \psi}} \\
        \mathbb{p}^{A,\psi}_A &= 0 &
	    \mathbb{p}^{A,\psi}_\rho &= 0 \\
	    \mathbb{p}^{A,\psi}_{\psi} &= - [\iota_{\xi}{A_0},\psi] &
	    \mathbb{p}^{A,\psi}_{\overline{\psi}} &= -  [\iota_{\xi}{A_0},\overline\psi].
	\end{align*}
 For $\mathbb{h}^{A, \psi}$ we need some more work. We have $\iota_{\mathbb{h}^{\psi}} \varpi_{A}= \iota_{\mathbb{h}^{A}} \varpi_{\psi}=0$ and 
 \begin{align*}
     \delta h^{A, \psi}_\lambda & = \int_{\Sigma}- \lambda e_n ie \delta e\left( \overline{\psi}\gamma [A, \psi]\right)- \lambda e_n\left[ i\frac{e^2}{2}\left( \delta \overline{\psi}\gamma [A, \psi] +\overline{\psi}\gamma [\delta A, \psi]-\overline{\psi}\gamma [A, \delta \psi]  \right)
     \right].
 \end{align*}
    Hence we get:
    \begin{align*}
	    \mathbb{h}^{A,\psi}_e &= 0 &
	    \mathbb{h}^{A,\psi}_\omega &= - \frac{i}{2}\lambda e_n \overline{\psi}\gamma [A, \psi]+ \mathbb{V}_{h^{A, \psi}}\\
        \mathbb{h}^{A,\psi}_A &= 0 &
	    (\mathbb{h}^{A,\psi}_\rho)^J_I &=- \frac{1}{2}g_i \lambda e_n e^2 \overline{\psi}_I\gamma \psi^J \\
	    \frac{e^3}{3!}\gamma \mathbb{h}^{A,\psi}_{\psi} &= \frac{\lambda e_n e^2}{2}\gamma [A, \psi] &
     \frac{e^3}{3!}\mathbb{h}^{A,\psi}_{\overline{\psi}}\gamma  &= \frac{\lambda e_n e^2}{2} [A, \overline{\psi}] \gamma.
	\end{align*}
 As for $p^{A, \psi}_{\xi}$, the Hamiltonian vector field of $m^{A, \psi}_{\mu}$ can be obtained by noticing that it is equal to that of $l^{ \psi}_{c}$ by substituting $c$ with $\mu$. The result is 
         \begin{align*}
	    \mathbb{m}^{A,\psi}_e &= 0 &
	    \mathbb{m}^{A,\psi}_\omega &= \mathbb{V}_{m^{A, \psi}} \\
        \mathbb{m}^{A,\psi}_A &= 0 &
	    \mathbb{m}^{A,\psi}_\rho &= 0 \\
	    \mathbb{m}^{A,\psi}_{\psi} &= [\mu, \psi]&
	    \mathbb{m}^{A,\psi}_{\overline{\psi}} &= [\mu, \overline\psi].
	\end{align*}
 We can now compute the constraints using Theorem \ref{thm:Poisson_Brackets_couples}. Before beginning the actual computation we note that for all constraints 
\begin{align*}
    \iota_{\mathbb{X}+\mathbb{x}^{\psi}}\iota_{\mathbb{Y}+\mathbb{y}^{\psi}}\varpi_{A}+ \iota_{\mathbb{X}+\mathbb{x}^{A}}\iota_{\mathbb{Y}+\mathbb{y}^{A}}\varpi_{\psi}=0,\\
    \iota_{\mathbb{X}}\iota_{\mathbb{Y}}\varpi_{A}=0, \\ \iota_{\mathbb{X}}\iota_{\mathbb{Y}}\varpi_{\psi}=0.
\end{align*}
Furthermore it is also possible to note that 
\begin{align*}
    \iota_{\mathbb{x}^{\psi}}\iota_{\mathbb{y}^{A}}\Omega_{YMS}=0\\
    \iota_{\mathbb{x}^{A}}\iota_{\mathbb{y}^{\psi}}\Omega_{YMS}=0
\end{align*}
for all brackets except $\{H^{A, \psi}_\lambda,H^{A, \psi}_\lambda\}$
\begin{align*}
    \iota_{\mathbb{h}^{\psi}}\iota_{\mathbb{h}^{A}}\Omega_{YMS}=\iota_{\mathbb{h}^{\psi}}\iota_{\mathbb{h}^{A}}\varpi= \int_{\Sigma} e \mathbb{h}_e^{\psi} \mathbb{h}_\omega^{A}=0
\end{align*}
since $\mathbb{h}_e^{\psi} \sim \lambda$, $\mathbb{h}_\omega^{\psi} \sim \lambda$ and $\lambda^2=0$. Hence we conclude that we have
\begin{align*}
    \{X^{A, \psi},Y^{A, \psi}\}_{YMS}&=\{X+x^{A}, Y+y^{A}\}_A+\{X+x^{\psi}, Y+y^{\psi}\}_\psi-\{X,Y\}\\
    &\phantom{=}+ \iota_{\mathbb{y}^{A, \psi}}\delta (X+ x^{A}+x^{\psi}+x^{A, \psi}) + \iota_{\mathbb{x}^{A, \psi}}\delta (Y+ y^{A}+y^{\psi}+y^{A, \psi})\\
    &\phantom{=}- \iota_{\mathbb{x}^{A, \psi}}\iota_{\mathbb{y}^{A, \psi}}\Omega_{YMS}. 
\end{align*}
Using this formula we can compute the brackets, omitting the terms that are zero.
 \begin{align*}
    \{M_{\mu}^{A, \psi},M_{\mu}^{A, \psi}\}_{YMS}&=\{M_{\mu}^{A}, M_{\mu}^{A}\}_A+ 2\iota_{\mathbb{m}^{A, \psi}}\delta ( M_{\mu}^{A}+m_{\mu}^{A, \psi}) -\iota_{\mathbb{m}^{A, \psi}}\iota_{\mathbb{m}^{A, \psi}}\Omega_{YMS} \\
    &= \frac{1}{2} M_{[\mu, \mu]}^{A} +  2\int_{\Sigma} i \frac{e^3}{2\cdot 3!} \left( [\mu,[\mu,\overline{\psi}]]\gamma \psi - \overline{\psi} \gamma [\mu,[\mu,\psi]] +[\mu,\overline{\psi}]\gamma [\mu,\psi]\right)\\
     &\phantom{=} +2\int_{\Sigma} i \frac{e^3}{2\cdot 3!}  [\mu,\overline{\psi}] \gamma [\mu,\psi]-2\int_{\Sigma} i \frac{e^3}{ 3!}  [\mu,\overline{\psi}] \gamma [\mu,\psi]\\
     &= \frac{1}{2} M_{[\mu, \mu]}^{A} + \frac{1}{2} m_{[\mu, \mu]}^{A, \psi}= \frac{1}{2}M_{[\mu, \mu]}^{A, \psi}.
\end{align*}
Since $\mathbb{l}^{A, \psi}=0$ we get
 \begin{align*}
    \{M_{\mu}^{A, \psi},L_c^{A, \psi}\}_{YMS}&=\{M_{\mu}^{A}, L_c\}_A
    + \iota_{\mathbb{m}^{A, \psi}}\delta (L_c+l_c^{\psi})\\
    &= \int_{\Sigma} - i \frac{e^3}{2\cdot 3!} \left( -[c,[\mu,\overline{\psi}]]\gamma \psi-[c,\overline{\psi}]\gamma [\mu,\psi] - [\mu, \overline{\psi}] \gamma [c,\psi] + \overline{\psi} \gamma [c, [\mu, \psi]]\right)=0.
\end{align*}
where we used that $[c,[\mu,\overline{\psi}]]=[\mu,[c,\overline{\psi}]]$ and the properties in Remark \ref{r:properties_mu}. 
 \begin{align*}
    \{M_{\mu}^{A, \psi},P_{\xi}^{A, \psi}\}_{YMS}&=\{M_{\mu}^{A}, P_{\xi}+p_{\xi}^{A}\}_A+ \iota_{\mathbb{p}^{A, \psi}}\delta (M_{\mu}^{A}+m_{\mu}^{A, \psi}) + \iota_{\mathbb{m}^{A, \psi}}\delta (P_{\xi}+ p_{\xi}^{A}+p_{\xi}^{\psi}+p_{\xi}^{A, \psi})\\
    &\phantom{=}- \iota_{\mathbb{m}^{A, \psi}}\iota_{\mathbb{p}^{A, \psi}}\Omega_{YMS}\\
    &= M^{A}_{\mathrm{L}_\xi^{A_0}\mu}- \int_{\Sigma} i \frac{e^3}{3!} \left([\mu,\overline{\psi}]\gamma [\iota_{\xi} {A_0}, \psi] + [\iota_{\xi}{A_0}, \overline{\psi}] \gamma [\mu,\psi] \right)\\
    &\phantom{=}-\int_{\Sigma} i \frac{e^3}{2\cdot 3!} \left( [\mu,\overline{\psi}] \gamma \mathrm{L}_{\xi}^{\omega_0+A_0}(\psi) -\overline{\psi} \gamma \mathrm{L}_{\xi}^{\omega_0+A_0}([\mu,\psi])+ \mathrm{L}_\xi^{\omega_0+A_0}([\mu,\overline{\psi}])\gamma \psi \right)\\
    &\phantom{=}+ \int_{\Sigma} i \frac{e^3}{2\cdot 3!} \left(\mathrm{L}_\xi^{\omega_0+A_0}(\overline{\psi})\gamma [\mu,\psi]  +2\left([\mu,\overline{\psi}]\gamma [\iota_{\xi} {A_0}, \psi] + [\iota_{\xi}{A_0}, \overline{\psi} ]\gamma [\mu,\psi] \right)\right)\\
    &= M^{A}_{\mathrm{L}_\xi^{A_0}\mu}  -\int_{\Sigma} i \frac{e^3}{2\cdot 3!} \left(  -\overline{\psi} \gamma [\mathrm{L}_{\xi}^{\omega_0+A_0}\mu,\psi]+ [\mathrm{L}_\xi^{\omega_0+A_0}\mu,\overline{\psi}]\gamma \psi \right)\\
    &= M^{A}_{\mathrm{L}_\xi^{A_0}\mu}  -\int_{\Sigma} i \frac{e^3}{2\cdot 3!} \left(  -\overline{\psi} \gamma [\mathrm{L}_{\xi}^{A_0}\mu,\psi]+ [\mathrm{L}_\xi^{A_0}\mu,\overline{\psi}]\gamma \psi \right)\\
    &= M^{A}_{\mathrm{L}_\xi^{A_0}\mu}+ m^{A, \psi}_{\mathrm{L}_\xi^{A_0}\mu}= M^{A, \psi}_{\mathrm{L}_\xi^{A_0}\mu}.
\end{align*}
where we used $\mathrm{L}_{\xi}^{\omega_0+A_0}\mu=\mathrm{L}_{\xi}^{A_0}\mu$ and $[\mu,\overline{\psi}] \gamma \mathrm{L}_{\xi}^{\omega_0+A_0}(\psi)=-\overline{\psi} \gamma [\mu,\mathrm{L}_{\xi}^{\omega_0+A_0}(\psi)]$.
Similarly we get 
\begin{align*}
    \{M_{\mu}^{A, \psi},H_{\lambda}^{A, \psi}\}_{YMS}&=\{M_{\mu}^{A}, H^A_{\lambda}\}_A + \iota_{\mathbb{h}^{A, \psi}}\delta (M_{\mu}^{A}+m_{\mu}^{A, \psi}) + \iota_{\mathbb{m}^{A, \psi}}\delta (H_{\lambda}+ h_{\lambda}^{A}+h_{\lambda}^{\psi}+h_{\lambda}^{A, \psi})\\
    &\phantom{=}- \iota_{\mathbb{m}^{A, \psi}}\iota_{\mathbb{h}^{A, \psi}}\Omega_{YMS}\\
    &= -\int_{\Sigma}\frac{i\lambda e_n e^2}{2} \overline{\psi}\gamma [d_A \mu, \psi] +\int_{\Sigma}\frac{i\lambda e_n e^2}{2} \left([\mu,\overline{\psi}]\gamma [A, \psi] - [A,\overline{\psi}]\gamma [\mu, \psi] \right)\\
    &\phantom{=}+\int_{\Sigma}\frac{i\lambda e_n e^2}{4} \left([\mu,\overline{\psi}]\gamma d_{\omega} \psi- \overline{\psi}\gamma d_{\omega}([\mu, \psi])  + d_{\omega}([\mu, \overline{\psi}])\gamma \psi +d_{\omega}( \overline{\psi})\gamma [\mu, \psi]  \right)\\
    &\phantom{=}+\int_{\Sigma}\frac{i\lambda e_n e^2}{4} \left([\mu,\overline{\psi}]\gamma [A, \psi]- \overline{\psi}\gamma [A,[\mu, \psi]] + [A,[\mu, \overline{\psi}]]\gamma \psi +[A,\overline{\psi}]\gamma [\mu, \psi]  \right)\\
    &\phantom{=}-\int_{\Sigma}\frac{i\lambda e_n e^2}{2} \left([\mu,\overline{\psi}]\gamma [A, \psi] - [A,\overline{\psi}]\gamma [\mu, \psi] \right)\\
    &= -\int_{\Sigma}\frac{i\lambda e_n e^2}{2} \overline{\psi}\gamma [d_A \mu, \psi]+\int_{\Sigma}\frac{i\lambda e_n e^2}{4} \left([\mu,\overline{\psi}]\gamma d_{\omega+A} \psi- \overline{\psi}\gamma d_{\omega+A}([\mu, \psi])\right) \\
    &\phantom{=} +\int_{\Sigma}\frac{i\lambda e_n e^2}{4} \left(d_{\omega+A}([\mu, \overline{\psi}])\gamma \psi +d_{\omega+A}( \overline{\psi})\gamma [\mu, \psi]  \right)\\
    &= -\int_{\Sigma}\frac{i\lambda e_n e^2}{2} \overline{\psi}\gamma [d_A \mu, \psi]+\int_{\Sigma}\frac{i\lambda e_n e^2}{2} \overline{\psi}\gamma [d_{A+\omega} \mu, \psi]=0
\end{align*}
where we used $d_{A+\omega} \mu=d_{A} \mu.$ Using again that $\mathbb{l}^{A, \psi}=0$ we get
 \begin{align*}
    \{L_c^{A, \psi},L_c^{A, \psi}\}_{YMS}&=\{L_c, L_c\}_A+\{L_c+l_c^{\psi}, L_c+l_c^{\psi}\}_\psi-\{L_c,L_c\}\\
    &= \frac{1}{2}(L_{[c,c]}+l_{[c,c]}^{\psi})= \frac{1}{2}L^{A,\psi}_{[c,c]}.
\end{align*}
Similarly, 
\begin{align*}
    \{L_c^{A, \psi},P_{\xi}^{A, \psi}\}_{YMS}&=\{L_c, P_{\xi}+p_{\xi}^{A}\}_A+\{L_c+l_c^{\psi}, P_{\xi}+p_{\xi}^{\psi}\}_\psi-\{L_c,P_{\xi}\}+ \iota_{\mathbb{p}^{A, \psi}}\delta (L_c+ l_c^{\psi})\\
    &= L^A_{\mathrm{L}_{\xi}^{\omega_0}c}+ L^{\psi}_{\mathrm{L}_{\xi}^{\omega_0}c}- L_{\mathrm{L}_{\xi}^{\omega_0}c}+ \int_{\Sigma}  i \frac{e^3}{2\cdot 3!} \left( [c,[\iota_{\xi}A_0,\overline{\psi}]]\gamma \psi + [c,\overline{\psi}]\gamma [\iota_{\xi}A_0,\psi]\right)\\
    &\phantom{=}+ \int_{\Sigma}  i \frac{e^3}{2\cdot 3!} \left(- [\iota_{\xi}A_0,\overline{\psi}] \gamma [c,\psi] - \overline{\psi} \gamma [c,[\iota_{\xi}A_0,\psi]]\right)\\
    &= L^{A,\psi}_{\mathrm{L}_{\xi}^{\omega_0}c} 
    + \int_{\Sigma}  i \frac{e^3}{2\cdot 3!} \left( -[\iota_{\xi}A_0,[c,\overline{\psi}]]\gamma \psi + [c,\overline{\psi}]\gamma [\iota_{\xi}A_0,\psi]\right)\\
    &\phantom{=}+ \int_{\Sigma}  i \frac{e^3}{2\cdot 3!} \left(- [\iota_{\xi}A_0,\overline{\psi}] \gamma [c,\psi] + \overline{\psi} \gamma [\iota_{\xi}A_0,[c,\psi]]\right)\\
    &= L^{A,\psi}_{\mathrm{L}_{\xi}^{\omega_0}c}
\end{align*}
where we used $[c,[\iota_{\xi}A_0,\overline{\psi}]]=-[\iota_{\xi}A_0,[c,\overline{\psi}]]$ and $[\iota_{\xi}A_0,[c,\overline{\psi}]]\gamma \psi = [c,\overline{\psi}]\gamma [\iota_{\xi}A_0,\psi]$.

 \begin{align*}
    \{L_c^{A, \psi},H_{\lambda}^{A, \psi}\}_{YMS}&=\{L_c, H_{\lambda}+h_{\lambda}^{A}\}_A+\{L_c+l_c^{\psi}, H_{\lambda}+h_{\lambda}^{\psi}\}_\psi-\{L_c,H_{\lambda}\}+ \iota_{\mathbb{h}^{A, \psi}}\delta (L_c+ l_c^{\psi})\\
    &= - P^{A}_{X^{(\nu)}} + L_{X^{(\nu)}(\omega - \omega_0)_{\nu}} -H^{A}_{X^{(n)}} + M^{A}_{X^{(\nu)}(A-A_0)_{\nu}} - P^{\psi}_{X^{(\nu)}} + L^{\psi}_{X^{(\nu)}(\omega - \omega_0)_{\nu}} \\
    &\phantom{=}- H^{\psi}_{X^{(n)}}- P_{X^{(\nu)}} + L_{X^{(\nu)}(\omega - \omega_0)_{\nu}} - H_{X^{(n)}} - \int_{\Sigma} \frac{i}{2} [\lambda e_n \overline{\psi}\gamma [A, \psi], e^2] \\
    &\phantom{=}- \int_{\Sigma} \frac{i}{4}\left(-[c, \lambda e_n e^2 [A,\overline{\psi}]]\gamma \psi + [c, \overline{\psi}]\gamma \lambda e_n e^2 [A,\psi]\right)\\
    &\phantom{=}+\int_{\Sigma} \frac{i}{4}\left(- \lambda e_n e^2 [A,\overline{\psi}] \gamma [c, \psi] + \overline{\psi}\gamma [c, \lambda e_n e^2 [A,\psi]] \right)\\
    &= - P^{A}_{X^{(\nu)}} + L_{X^{(\nu)}(\omega - \omega_0)_{\nu}} -H^{A}_{X^{(n)}} + M^{A}_{X^{(\nu)}(A-A_0)_{\nu}} - P^{\psi}_{X^{(\nu)}} + L^{\psi}_{X^{(\nu)}(\omega - \omega_0)_{\nu}}\\
    & \phantom{=}- H^{\psi}_{X^{(n)}}- P_{X^{(\nu)}} + L_{X^{(\nu)}(\omega - \omega_0)_{\nu}} - H_{X^{(n)}}\\
    &\phantom{=} - \int_{\Sigma}[c, \lambda e_n]\left[ i\frac{e^2}{4}\left( \overline{\psi}\gamma [A, \psi] - [A,\overline{\psi}]\gamma \psi \right) \right]\\
    &= - P^{A}_{X^{(\nu)}} + L_{X^{(\nu)}(\omega - \omega_0)_{\nu}} -H^{A}_{X^{(n)}} + M^{A}_{X^{(\nu)}(A-A_0)_{\nu}} - P^{\psi}_{X^{(\nu)}} + L^{\psi}_{X^{(\nu)}(\omega - \omega_0)_{\nu}}\\
    &\phantom{=} - H^{\psi}_{X^{(n)}}- P_{X^{(\nu)}} + L_{X^{(\nu)}(\omega - \omega_0)_{\nu}} - H_{X^{(n)}} - p^{A, \psi}_{X^{(\nu)}} - h^{A, \psi}_{X^{(n)}} + m^{A, \psi}_{X^{(\nu)}(A-A_0)_{\nu}}\\
    &= - P^{A, \psi}_{X^{(\nu)}} + L^{A, \psi}_{X^{(\nu)}(\omega - \omega_0)_{\nu}} -H^{A, \psi}_{X^{(n)}} + M^{A, \psi}_{X^{(\nu)}(A-A_0)_{\nu}}
\end{align*}
where in the second-last passage we used that $[c,\lambda e_n]=X= X^{(\nu)}e_\nu+ X^{(n)}e_n$ and that \begin{align*}
    - p^{A, \psi}_{X^{(\nu)}}  + m^{A, \psi}_{X^{(\nu)}(A-A_0)_{\nu}}=- \int_{\Sigma}[c, \lambda e_n]^{(\nu)}e_\nu\left[ i\frac{e^2}{4}\left( \overline{\psi}\gamma [A, \psi] - [A,\overline{\psi}]\gamma \psi \right) \right]
\end{align*}
Let us now consider
 \begin{align*}
    \{P_{\xi}^{A, \psi},P_{\xi}^{A, \psi}\}_{YMS}&=\{P_{\xi}+p_{\xi}^{A}, P_{\xi}+p_{\xi}^{A}\}_A+\{P_{\xi}+p_{\xi}^{\psi}, P_{\xi}+p_{\xi}^{\psi}\}_\psi-\{P_{\xi},P_{\xi}\}\\
    &\phantom{=}+ 2\iota_{\mathbb{p}^{A, \psi}}\delta (P_{\xi}+ p_{\xi}^{A}+p_{\xi}^{\psi}+p_{\xi}^{A, \psi}) - \iota_{\mathbb{p}^{A, \psi}}\iota_{\mathbb{p}^{A, \psi}}\Omega_{YMS}.
\end{align*}
    We have
    \begin{align*}
    &\{P_{\xi}+p_{\xi}^{A}, P_{\xi}+p_{\xi}^{A}\}_A+\{P_{\xi}+p_{\xi}^{\psi}, P_{\xi}+p_{\xi}^{\psi}\}_\psi-\{P_{\xi},P_{\xi}\}\\
    &=\frac{1}{2}\left(P^{A}_{[\xi, \xi]}- L_{\iota_{\xi}\iota_{\xi}F_{\omega_0}}-M^{A}_{\iota_{\xi}\iota_{\xi}F_{A_0}} + P^{\psi}_{[\xi, \xi]}- L^{\psi}_{\iota_{\xi}\iota_{\xi}F_{\omega_0}}-P_{[\xi, \xi]}+ L_{\iota_{\xi}\iota_{\xi}F_{\omega_0}}\right)
\end{align*}
    and 
    \begin{align*}
    &2\iota_{\mathbb{p}^{A, \psi}}\delta (P_{\xi}+ p_{\xi}^{A}+p_{\xi}^{\psi}+p_{\xi}^{A, \psi}) - \iota_{\mathbb{p}^{A, \psi}}\iota_{\mathbb{p}^{A, \psi}}\Omega_{YMS}\\
    &= -\int_{\Sigma} i \frac{e^3}{2\cdot 3!} \left( [\iota_{\xi}A_0,\overline{\psi}] \gamma \mathrm{L}_{\xi}^{\omega_0+A_0}(\psi)+\overline{\psi} \gamma \mathrm{L}_{\xi}^{\omega_0+A_0}([\iota_{\xi}A_0,\psi])\right)\\
    &\phantom{=}-\int_{\Sigma} i \frac{e^3}{2\cdot 3!} \left( - \mathrm{L}_\xi^{\omega_0+A_0}(\overline{\psi})\gamma [\iota_{\xi}A_0,\psi]- \mathrm{L}_\xi^{\omega_0+A_0}([\iota_{\xi}A_0,\overline{\psi}])\gamma \psi  \right)\\
    &=-\int_{\Sigma} i \frac{e^3}{2\cdot 3!} \left( -\overline{\psi} \gamma [\iota_{\xi}A_0,\mathrm{L}_{\xi}^{\omega_0+A_0}(\psi)]+\overline{\psi} \gamma \mathrm{L}_{\xi}^{\omega_0+A_0}([\iota_{\xi}A_0,\psi])\right)\\
    &\phantom{=}-\int_{\Sigma} i \frac{e^3}{2\cdot 3!} \left( - [\iota_{\xi}A_0,\mathrm{L}_\xi^{\omega_0+A_0}(\overline{\psi})]\gamma \psi- \mathrm{L}_\xi^{\omega_0+A_0}([\iota_{\xi}A_0,\overline{\psi}])\gamma \psi  \right)\\
    &=\int_{\Sigma}- i \frac{e^3}{2\cdot 3!} \left( -\overline{\psi} \gamma [\mathrm{L}_\xi^{\omega_0+A_0}(\iota_{\xi}A_0),\psi]- [\mathrm{L}_\xi^{\omega_0+A_0}(\iota_{\xi}A_0),\overline{\psi}]\gamma \psi  \right)\\
    &= \int_{\Sigma} i \frac{e^3}{4\cdot 3!} \left( \overline{\psi} \gamma \left[\iota_{[\xi,\xi]}A_0+\iota_{\xi}\iota_{\xi}F_{A_0},\psi\right]+ \left[\iota_{[\xi,\xi]}A_0+\iota_{\xi}\iota_{\xi}F_{A_0},\overline{\psi}\right]\gamma \psi  \right)\\
    &= \frac{1}{2}p^{A,\psi}_{[\xi, \xi]}- m^{A, \psi}_{\iota_{\xi}\iota_{\xi}F_{A_0}}
    \end{align*}
    where we used that $\mathrm{L}_\xi^{\omega_0+A_0}\iota_{\xi}A_0=\mathrm{L}_\xi^{A_0}\iota_{\xi}A_0= \frac{1}{2}\iota_{[\xi,\xi]}A_0+\frac{1}{2}\iota_{\xi}\iota_{\xi}F_{A_0}.$
    Hence we get
    \begin{align*}
    \{P_{\xi}^{A, \psi},P_{\xi}^{A, \psi}\}_{YMS}&= \frac{1}{2}\left(P^{A,\psi}_{[\xi, \xi]}- L^{A,\psi}_{\iota_{\xi}\iota_{\xi}F_{\omega_0}}-M^{A,\psi}_{\iota_{\xi}\iota_{\xi}F_{A_0}}\right).
    \end{align*}

 \begin{align*}
    \{P_{\xi}^{A, \psi},H_{\lambda}^{A, \psi}\}_{YMS}&=\{P_{\xi}+p_{\xi}^{A}, H_{\lambda}+h_{\lambda}^{A}\}_A+\{P_{\xi}+p_{\xi}^{\psi}, H_{\lambda}+h_{\lambda}^{\psi}\}_\psi-\{P_{\xi},H_{\lambda}\}\\
    &\phantom{=}+ \iota_{\mathbb{h}^{A, \psi}}\delta (P_{\xi}+ p_{\xi}^{A}+p_{\xi}^{\psi}+p_{\xi}^{A, \psi}) + \iota_{\mathbb{p}^{A, \psi}}\delta (H_{\lambda}+ h_{\lambda}^{A}+h_{\lambda}^{\psi}+h_{\lambda}^{A, \psi})\\
    &\phantom{=}- \iota_{\mathbb{p}^{A, \psi}}\iota_{\mathbb{h}^{A, \psi}}\Omega_{YMS}\\
    &=  P^{A}_{Y^{(\nu)}} -L_{ Y^{(\nu)} (\omega - \omega_0)_{\nu}} +  H^{A}_{ Y^{(n)}} - M^{A}_{Y^{(\nu)}(A-A_0)_{\nu}}  +  p^{\psi}_{Y^{(\nu)}} -l^\psi_{ Y^{(\nu)} (\omega - \omega_0)_{\nu}}\\
    &\phantom{=} +  h^{\psi}_{ Y^{(n)}}- \int_{\Sigma} \frac{i \lambda e_n \mathrm{L}_\xi^{\omega_0}e^2 }{4}\left(\overline{\psi}\gamma [A,\psi]-[A,\overline{\psi}])\gamma \psi\right)\\
    &\phantom{=}- \frac{i \lambda e_n e^2 }{2}\left(\overline{\psi}\gamma[\iota_{\xi}F_{\omega_0}+\mathrm{L}_\xi^{A_0}(A-A_0), \psi]- [\iota_{\xi}F_{\omega_0}+\mathrm{L}_\xi^{A_0}(A-A_0),\overline{\psi}]\gamma \psi \right)\\
    &\phantom{=}-\int_{\Sigma} \frac{i \lambda e_n e^2 }{2}\left([A,\overline{\psi}]\gamma \mathrm{L}_\xi^{\omega_0+A_0}\psi - \mathrm{L}_\xi^{\omega_0+A_0}\overline{\psi}\gamma [A,\psi] \right)  \\
    &\phantom{=}- \int_{\Sigma}\frac{i \lambda e_n e^2 }{4}\left([\iota_\xi {A_0},\overline{\psi} ]\gamma d_{\omega+A}\psi-\overline{\psi} \gamma d_{\omega+A}[\iota_\xi {A_0},\psi]+d_{\omega+A}[\iota_\xi {A_0},\overline{\psi} ]\gamma \psi \right)\\
    &\phantom{=}- \int_{\Sigma}\frac{i \lambda e_n e^2 }{4}\left(d_{\omega+A}\overline{\psi} \gamma[\iota_\xi {A_0},\psi ] +2[A,\overline{\psi}]\gamma [\iota_\xi A_0,\psi] + 2[\iota_\xi A_0 ,\overline{\psi}]\gamma [A,\psi] \right)  
    \end{align*}
    Hence we get
    \begin{align*}
    \{P_{\xi}^{A, \psi},H_{\lambda}^{A, \psi}\}_{YMS}&=
      P^{A}_{Y^{(\nu)}} -L_{ Y^{(\nu)} (\omega - \omega_0)_{\nu}} +  H^{A}_{ Y^{(n)}} - M^{A}_{Y^{(\nu)}(A-A_0)_{\nu}}  +  p^{\psi}_{Y^{(\nu)}} -l^\psi_{ Y^{(\nu)} (\omega - \omega_0)_{\nu}}\\
    &\phantom{=} +  h^{\psi}_{ Y^{(n)}}- \int_{\Sigma}  \frac{i \mathrm{L}_\xi^{\omega_0}(\lambda e_n) e^2 }{4}\left(\overline{\psi}\gamma [A,\psi]-[A,\overline{\psi}])\gamma \psi\right)\\
    &=  P^{A}_{Y^{(\nu)}} -L_{ Y^{(\nu)} (\omega - \omega_0)_{\nu}} +  H^{A}_{ Y^{(n)}} - M^{A}_{Y^{(\nu)}(A-A_0)_{\nu}}  +  p^{\psi}_{Y^{(\nu)}} -l^\psi_{ Y^{(\nu)} (\omega - \omega_0)_{\nu}}\\
    &\phantom{=} +  h^{\psi}_{ Y^{(n)}}+  p^{A, \psi}_{Y^{(\nu)}}  +  h^{A, \psi}_{ Y^{(n)}} - m^{A, \psi}_{Y^{(\nu)}(A-A_0)_{\nu}}\\
    &= P^{A, \psi}_{Y^{(\nu)}} -L^{A, \psi}_{ Y^{(\nu)} (\omega - \omega_0)_{\nu}} +  H^{A, \psi}_{ Y^{(n)}} - M^{A, \psi}_{Y^{(\nu)}(A-A_0)_{\nu}}
\end{align*}

 \begin{align*}
    \{H_{\lambda}^{A, \psi},H_{\lambda}^{A, \psi}\}_{YMS}&=\{H_{\lambda}+h_{\lambda}^{A}, H_{\lambda}+h_{\lambda}^{A}\}_A+\{H_{\lambda}+h_{\lambda}^{\psi}, H_{\lambda}+h_{\lambda}^{\psi}\}_\psi-\{H_{\lambda},H_{\lambda}\}\\
    &\phantom{=}+ \iota_{\mathbb{h}^{A, \psi}}\delta (H_{\lambda}+ h_{\lambda}^{A}+h_{\lambda}^{\psi}+h_{\lambda}^{A, \psi}) + \iota_{\mathbb{h}^{A, \psi}}\delta (H_{\lambda}+ h_{\lambda}^{A}+h_{\lambda}^{\psi}+h_{\lambda}^{A, \psi})\\
    &\phantom{=}- \iota_{\mathbb{h}^{A, \psi}}\iota_{\mathbb{h}^{A, \psi}}\Omega_{YMS} =0
\end{align*}
because all the components of the Hamiltonian vector field of $H_{\lambda}^{A, \psi}$ are proportional to $\lambda$ and $\lambda^2=0$.
\end{proof}

\subsection{Yang--Mills-Higgs}\label{s:YM+scalar}

We consider the case of an interacting scalar field and a Yang--Mills field both coupled to gravity, with the addition of a Higgs-type potential.

First of all, let $P_{SU(n)}$ be a a $SU(n)$-principal bundle over $M$, with the fundamental representation $n\colon SU(n)\rightarrow \mathrm{End}(\mathbb{C}^n)$ and its conjugate one $\bar{n}$ with respect to the canonical hermitian structure on $\mathbb{C}^n$. 

We define the Higgs field $\phi$ (a scalar multiplet) to be a section of the associated bundle $E_n:=P_{SU(n)} \times_n \mathbb{C}^n$, while $\phi^\dagger$ is a section of $E_{\bar{n}}:=P_{SU(n)} \times_{\bar{n}} \mathbb{C}^n$.
Working in the first order formalism, we also introduce the associated momentum
$\Pi\in\Gamma(M,\mathcal{V}\otimes E_n)=:\Omega^{(0,1)}(E_n)$ and its conjugate $\Pi^\dagger\in\Omega^{(0,1)}( E_{\bar{n}})$.

\begin{remark}
    In the remainder of this section, we will identify (sections of) the Lie algebra $ \mathfrak{su}(n)$ with (sections of) the algebra of hermitian traceless matrices over $\mathbb{C}^n$, i.e.
        \begin{equation*}
            \Gamma(M,\mathfrak{su}(n))\simeq\Gamma(M,(E_n\otimes E_{\bar{n}})_{\mathrm{t,h}})=:\Gamma(M,(E_n\otimes E_{\bar{n}})').
        \end{equation*}
    Furthermore, we will consider $\phi$ and $\Pi$ to be such that the total degrees are $|\phi|=0$ and $|\Pi|=1$.%
\end{remark}

\begin{remark}
    The canonical hermitian product on $\mathbb{C}^n$ induces a hermitian product on $E_n$ (and hence on $\Gamma(E_n)$). We \emph{symmetrize} it (i.e. add its complex conjugate) to account for the reality requirement 
        \begin{align*}
            <\cdot,\cdot>\colon E_n\times E_n &\longrightarrow \mathbb{C}\\
            (\hspace{1mm}\phi\hspace{1mm},\hspace{1mm}\varphi\hspace{1mm})&\longmapsto\hspace{1mm} <\phi,\varphi>:=\frac{1}{2}(\phi^\dagger\varphi  + (-1)^{|\phi||\varphi|} \varphi^\dagger \phi).
        \end{align*}

    Furthermore, we denote the full interior product on $E_n\otimes \mathcal{V}$ by 
        \begin{align*}
            (<\cdot,\cdot>)\colon (E_n\otimes \mathcal{V})^2 &\longrightarrow \mathbb{C}\\
            (\Pi,\epsilon)&\longmapsto\hspace{1mm} (<\Pi,\epsilon>):=\frac{1}{2}\eta_{ab}(\Pi^{\dagger a}_i \epsilon^{b,i} + \text{c.c}).
        \end{align*}
\end{remark}

The last ingredient we need to write the YMH action is the covariant derivative. Letting $\alpha\in\Gamma(M,(E_n\otimes E_{\bar{n}})')$, we set $d_\alpha\phi:=d\phi + [A,\phi]$ and $d_\alpha\phi^\dagger:=d\phi^\dagger + [\alpha,\phi^\dagger]$, whilst in coordinates we have\footnote{the action of $\mathfrak{su}(n)$ on $\phi^\dagger$ is defined by requiring that $[\alpha,<\phi,\phi>]=0$, since $<\phi,\phi>$ is assumed to be a $SU(n)$ scalar.} 
    \begin{align*}
        [\alpha,\phi]^i&:=i g_H (\alpha \phi)^i = i g \alpha^i_j \phi^j ;\\
        [\alpha,\phi^\dagger]_i&=-(-1)^{|\alpha|(|\phi|+1)} i g_H \phi^\dagger_j \alpha^j_i,
    \end{align*}
where $g_H$ is a coupling constant related to the representation of $SU(n)$.

With those definitions having been established, denoting the YM field and its conjugate momentum by $A$ and $B$ as before, the desired action in dimension $N$ is given by   $S_{YMH}=S+S_{A}+S_H$ where $S_A$ is defined in \eqref{e:action_YM} and 
    \begin{align*}
        S_H=\int_M &\frac{e^3}{3!}
        <\Pi,d_A\phi> + \frac{e^4}{2 \cdot 4!} (<\Pi,\Pi>)-\frac{q}{4 \cdot 4!}e^4(<\phi,\phi>-v^2),
    \end{align*}
where $q$ is another coupling constant and $v$ represent the Higgs vacuum. 

\begin{remark}
    Notice that the first terms of $S_H$ are formally equivalent to $S_{\phi}$, after substituting $d\phi\rightarrow d_A\phi$ and generalizing $\phi$ to a $SU(n)$ multiplet. Then one can easily show that in this case the structural constraint reads 
        \begin{equation}\label{e:structural_Higgs}
            (e,\Pi)+d_A\phi=0.
        \end{equation}
    (The structural constraint for the $B$ field remains unaltered.) 
    As a consequence the interaction between the YM field and the Higgs field is contained in $S_H$.
\end{remark}

We obtain the geometric phase space  as the bundle
    \begin{equation*}
        {F}^\partial_{YMH}\rightarrow\Omega_{e_n}^1(\Sigma, \mathcal{V}_{\Sigma})\oplus \mathcal{A}_\Sigma^{SU(n)}\oplus \Gamma(\Sigma,E_n\vert_\Sigma)\times\Gamma(\Sigma,E_{\bar{n}}\vert_\Sigma),
    \end{equation*}
with fiber 
    \begin{equation*}
        \mathcal{A}_{red}(\Sigma)\oplus \Omega^{(0,2)}_{\Sigma, red}(\mathfrak{g})\oplus\Omega^{0,1}_\partial(E_{n}\vert_\Sigma)\times\Omega^{0,1}_\partial(E_{\bar{n}}\vert_\Sigma ).
    \end{equation*}
    where $\omega,B,\Pi,\Pi^\dagger$ satisfy
    \eqref{e:structural_constraint_grav}, \eqref{e:constraintYM} and \eqref{e:structural_Higgs}.
${F}^\partial_{YMH}$ is symplectic with symplectic form
    \begin{equation*}
        \Omega_{YMH}=\varpi + \varpi_{A}+\varpi_H
    \end{equation*}
where 
    \begin{equation}\label{e:sympl_form_H}
        \varpi_H = \int_\Sigma  <\delta p,\delta \phi> 
    \end{equation}
having defined $p:=\frac{e^3}{3!}\Pi$ to get rid of the unphysical components of $\Pi$, as we did in the case of the free scalar field.

Before moving on to the constraint analysis, we provide some useful identities. Let $\phi,\varphi \in \Gamma(E_n)$, $\alpha \in \Gamma(\mathfrak{su}(n))$, then: 
    \begin{align}
        <\alpha \varphi, \phi> &= (-1)^{|\alpha||\varphi|}<\varphi, \alpha\phi> \\
        <\varphi, [\alpha,\phi]>&=\frac{i g_H}{2}\Tr[(-1)^{|\alpha||\varphi|} \varphi \phi^\dagger - (-1)^{|\phi|(|\alpha|+|\varphi|)}\phi\varphi^\dagger )\alpha ]\\
        [L_\xi^{A_0},d_A]\phi&=L_\xi^{A_0} d_A\phi + d_A L_\xi^{A_0} \phi = [\iota_\xi F_{A_0} + L_\xi^{A_0}(A-A_0),\phi]\label{id: comm L_Ad_Aphi}
    \end{align}
Let us now consider  the constraints. The coupling of the scalar field to the YM field produces the expected constraints (i.e. free gravity + free YM + free scalar) plus the expected interaction terms (the one arising from the minimal coupling via the covariant derivative $d_A$). Indeed we have
    \begin{align}
        &m_\mu^{H,A}=\int_\Sigma \frac{i}{2} g_H \Tr[\mu (\phi p^\dagger - p \phi^\dagger)]=\int_\Sigma <p,[\mu,\phi]> ;\label{e:constraint_YMH4}\\
        &p_\xi^H =\int_\Sigma -<p,\mathrm{L}_\xi^{\omega_0}\phi>;\label{e:constraint_H2}\\
        &p_\xi^{H,A}=\int_\Sigma -<p, \iota_\xi[A_0,\phi]> \label{e:constraint_YMH2}\\
        &h_\lambda^H=\int_\Sigma \lambda e_n \big[ \frac{e^2}{2}<\Pi,d\phi>  + \frac{e^3}{2\cdot 3!}(<\Pi,\Pi>) - \frac{q_H}{2\cdot 3!}e^3(<\phi,\phi>-v^2)^2 \big];\label{e:constraint_H3}\\
        &h_\lambda^{H,A}=\int_\Sigma \lambda e_n \frac{e^2}{2} <\Pi,[A,\phi]>.\label{e:constraint_YMH3}
    \end{align}
\begin{remark}
    Notice that $h_\lambda^H=h_\lambda^\phi + h_\lambda^{V_H}$, where $h_\lambda^{V_H}$ is the term containing the Higgs potential term $V_H:=\frac{1}{2}q_H (<\phi,\phi>-v^2)^2$
    \begin{equation*}
        h_\lambda^{V_H}= - \int_\Sigma \lambda e_n  \frac{q_H}{24}e^3(<\phi,\phi>-v^2)^2 =\int_\Sigma \lambda e_n  \frac{e^3}{3!}V_H
    \end{equation*}   
\end{remark}    
Obtaining
    \begin{align*}
    L_c^{H,A}&=L_c; &  P_{\xi}^{H,A}&=P_\xi + p_\xi^A + p_\xi^H + p_\xi^{H,A};\\
    M_\mu^{H,A}&=M_\mu^A + m_\mu^{H,A}; &   H_{\lambda}^{H,A}&=H_\lambda + h_\lambda^A + h_\lambda^{H} + h_\lambda^{H,A}    
    \end{align*}

\begin{theorem}\label{thm:first-class-constraints_Higgs}
    Assume that $g^\partial$ is non-degenerate on $\Sigma$. Then, the zero locus of the functions  $L^{H,A}_c$, $P^{H,A}_{\xi}$,  $H^{H,A}_{\lambda}$ and $M_\mu^{H,A}$ defined above is a  coisotropic submanifold with respect to the symplectic structure $\Omega_{YMH}$. Their mutual Poisson brackets read
    {\footnotesize
         \begin{align*}
             \{L_c^{H,A}, L_c^{H,A}\}_{YMH} &= - \frac{1}{2} L^{H,A}_{[c,c]}  & \{L^{H,A}_c, P^{H,A}_{\xi}\}_{YMH}  &=  L^{H,A}_{\mathrm{L}_{\xi}^{\omega_0}c}\\
             \ \{ L^{H,A}_c, M^{H,A}_\mu \}_{YMH} & = 0   & \{L^{H,A}_c,  H^{H,A}_{\lambda}\}_{YMH}  &= - P^{H,A}_{X^{(\nu)}} + L^{H,A}_{X^{(\nu)}(\omega - \omega_0)_\nu} - H^{H,A}_{X^{(n)}} + M^{H,A}_{X^{\nu}(A-A_0)_\nu} \\
             \ \{M^{H,A}_\mu, P^{H,A}_{\xi}\}_{YMH}  &=  M^{H,A}_{\mathrm{L}_{\xi}^{\omega_0}\mu  } & \{P^{H,A}_{\xi},H^{H,A}_{\lambda}\}_{YMH}  &=  P^{H,A}_{Y^{(\nu)}} -L^{H,A}_{ Y^{(\nu)} (\omega - \omega_0)_\nu} +H^{H,A}_{ Y^{(n)}} - M^{H,A}_{Y^\nu (A-A_0)_\nu} \\
             \ \{M^{H,A}_\mu, H^{H,A}_\lambda\}_{YMH}&= 0  & \{P^{H,A}_{\xi}, P^{H,A}_{\xi}\}_{YMH}  &=  \frac{1}{2}P^{H,A}_{[\xi, \xi]}- \frac{1}{2}L^{H,A}_{\iota_{\xi}\iota_{\xi}F_{\omega_0}} - \frac{1}{2}M^{H,A}_{\iota_{\xi}\iota_{\xi}F_{\omega_0}} \\
             \ \{M_\mu^{H,A},M_\mu^{H,A}\}_{YMH}&=-\frac{1}{2}M_{[\mu,\mu]}^{H,A} & \{H^{H,A}_{\lambda},H^{H,A}_{\lambda}\}_{YMH}  &=0 
         \end{align*}
        }
 with the same notation as in Theorem \ref{thm:first-class-constraints}. 
 
\end{theorem}    

\begin{proof}
    We use the results of appendix \ref{s:appendix_Hamiltonianvf} for the components of the Hamiltonian vector fields of the non-interacting theories. In particular, we have $\mathbb{x}^H=\mathbb{x}^\phi$. The residual components are computed using the results in Section \ref{s:appendix_Poisson_brackets}.
    We start with $M_\mu^{H,A}$. One can quite easily see that
        \begin{equation*}
            \mathbb{m}^{H,A}_\phi=[\mu,\phi], \qquad \mathbb{m}^{H,A}_p=[\mu,p].
        \end{equation*}
     Now, since $l_c^{H,A}=0$ and $\iota_{\mathbb{l}^\phi  }\varpi_A=\iota_{\mathbb{l}^A}\varpi_H=0$, one finds
        \begin{equation*}
            \mathbb{l}^{H,A}=0.
        \end{equation*}
    For $P_\xi^{H,A }$ we find 
        \begin{equation*}
            \delta p_\xi^{A,H}= \int_\Sigma - <\delta p, \iota_\xi[A_o,\phi ] > + <[A_0,\iota_\xi p], \delta \phi>= \underbrace{\iota_{\mathbb{p}^\phi  }\varpi_A}_0 + \underbrace{\iota_{\mathbb{p}^A  }\varpi_H}_0 + \iota_{\mathbb{p}^{H,A}}\Omega_{YMH},
        \end{equation*}
    finding
        \begin{equation*}
            \mathbb{p}^{H,A}_\phi=-\iota_\xi[A_0,\phi], \qquad \mathbb{p}^{H,A}_p=[A_0,\iota_\xi p]; 
        \end{equation*}
    while all the other components of $\mathbb{p}^{H,A}$ vanish.    

    Regarding $H_\lambda^{H,A}$, we find that $\mathbb{h}^H=\mathbb{h}^\phi$ as expected, except for the components that inherit the Higgs potential term:
        \begin{align*}
            \mathbb{h}^H_p&=d_\omega\left(\frac{\lambda e_n}{2} e^2 \Pi\right) +\frac{\lambda e_n}{2\cdot 3!}q_h e^3 (<\phi,\phi> - v^2 )\phi;\\
            e\mathbb{h}^H_\omega&=\lambda e_n \left( e<\Pi,d\phi> + \frac{e^2}{4}<(\Pi,\Pi)> + \frac{e^2}{2}V_H \right) -\frac{\lambda}{2}e^2\Pi(\Pi,e_n) + \mathbb{V}_{h^H}.
        \end{align*}
    For $h_\lambda^{H,A}$ we obtain
        \begin{align*}
            \delta h_\lambda^{H,A}&= \int_\Sigma \lambda e_n \left[  e <\Pi,[A,\phi]>\delta e + \frac{e^2}{2}\left(<\delta \Pi, [A,\phi]> +<[A,\Pi],\delta \phi>  \right)  \right]\\
            &\phantom{=}+\int_\Sigma \lambda e_n \frac{i g_H e^2}{4} \Tr[(\Pi \phi^\dagger - \phi\Pi^\dagger)\delta A] \\
            &= \underbrace{\iota_{\mathbb{h}^\phi  }\varpi_A}_0 + \underbrace{\iota_{\mathbb{h}^A  }\varpi_H}_0 + \iota_{\mathbb{h}^{H,A}}\Omega_{YMH},
        \end{align*}
    hence finding
        \begin{align*}
            \mathbb{h}^{A,H}_e&=0 &e\mathbb{h}_\omega^{A,H}&=-\frac{\lambda e_n}{2}e<\Pi,[A,\phi]>\\
            \frac{e^3}{3!}\mathbb{h}^{A,H}_\phi&=\frac{\lambda e_n}{2}e^2 [A.\phi] & \frac{e^3}{3!}\mathbb{h}^{A,H}_\Pi&= \frac{\lambda e_n}{2}e^2 [A,\Pi]\\
            \mathbb{h}^{A,H}_A&=0  &\mathbb{h}^{A,H}_\rho&=\frac{i g_H}{4}\lambda e_n e^2 \Tr(\Pi \phi^\dagger - \phi \Pi^\dagger)
        \end{align*}
    For the computations we use  Theorem \ref{thm:Poisson_Brackets_couples} and the results of Theorem \ref{thm:first-class-constraints} and of corresponding results in presence of a scalar Higgs field and and a Yang--Mills field.

     As before we note that for all constraints 
\begin{align*}
    \iota_{\mathbb{X}+\mathbb{x}^{H}}\iota_{\mathbb{Y}+\mathbb{y}^{H}}\varpi_{A}+ \iota_{\mathbb{X}+\mathbb{x}^{A}}\iota_{\mathbb{Y}+\mathbb{y}^{A}}\varpi_{H}=0,\\
    \iota_{\mathbb{X}}\iota_{\mathbb{Y}}\varpi_{A}=0, \\ \iota_{\mathbb{X}}\iota_{\mathbb{Y}}\varpi_{H}=0,\\
    \iota_{\mathbb{x}^{H}}\iota_{\mathbb{y}^{A}}\Omega_{YMH}=0,\\
    \iota_{\mathbb{x}^{A}}\iota_{\mathbb{y}^{H}}\Omega_{YMH}=0.
\end{align*}
 Hence we can use the simplified formula
\begin{align*}
    \{X^{H,A},Y^{H,A}\}_{YMH}&=\{X+x^{A}, Y+y^{A}\}_A+\{X+x^{H}, Y+y^{H}\}_H-\{X,Y\}\\
    &\phantom{=}+ \iota_{\mathbb{y}^{H,A}}\delta (X+ x^{A}+x^{H}+x^{H,A}) + \iota_{\mathbb{x}^{H,A}}\delta (Y+ y^{A}+y^{H}+y^{H,A})\\
    &\phantom{=}- \iota_{\mathbb{x}^{H,A}}\iota_{\mathbb{y}^{H,A}}\Omega_{YMH}. 
\end{align*}
Applying it we get:
    \begin{align*}
        \{M_\mu^{A,H},M_\mu^{A,H}\}_{YMH}&=\{M_\mu^A, M_\mu^A \}_A+ 2\iota_{\mathbb{m}^{H, A}}\delta ( M_{\mu}^{A}+m_{\mu}^{H, A}) -\iota_{\mathbb{m}^{H, A}}\iota_{\mathbb{m}^{H, A}}\Omega_{YMH}\\
        &= \frac{1}{2} M_{[\mu, \mu]}^{A} + \int_\Sigma <[\mu,p],[\mu,\phi]>+ <p,[\mu,[\mu,\phi]]>-\int_\Sigma <[\mu,p],[\mu,\phi]> \\
        &=\frac{1}{2} M^A_{[\mu,\mu]}+\frac{1}{2}\int_\Sigma <p,[[\mu,\mu],\phi]> \\
        &=\frac{1}{2} M^{H,A}_{[\mu,\mu]};
    \end{align*}
    
    \begin{align*}
        \{L^{A,H}_c,M^{A,H}_\mu\}_{YMH}&=\{M_{\mu}^{A}, L_c\}_A + \iota_{\mathbb{m}^{H, A}}\delta L_c=0
    \end{align*}

    \begin{align*}
        \{ M_\mu^{H,A}, P_\xi^{H,A} \}_{YMH}&=\{M_{\mu}^{A}, P_{\xi}+p_{\xi}^{A}\}_A+ \iota_{\mathbb{p}^{H, A}}\delta (M_{\mu}^{A}+m_{\mu}^{H, A})\\
    &\phantom{=}- \iota_{\mathbb{m}^{H, A}}\iota_{\mathbb{p}^{H, A}}\Omega_{YMH}+ \iota_{\mathbb{M}^{A}}\iota_{\mathbb{P}+\mathbb{p}^{A}}\varpi_{H} +\iota_{\mathbb{M}^{A}}\iota_{\mathbb{p}^{H}}\Omega_{YMH}\\
    &= \int_\Sigma \left( < [A_0, \iota_{\xi}p],[\mu,\phi]> + <p,[\mu,[\iota_{\xi}A_0, \phi]]>\right) \\
    &\phantom{=}+M^A_{\mathrm{L}^{A_0}\mu}+\int_\Sigma \left( -<[\mu,p],\mathrm{L}_\xi^{\omega_0+A_0}\phi>-<p,\mathrm{L}_\xi^{\omega_0+A_0}[\mu,\phi]> \right)\\
    &\phantom{=}-\int_\Sigma \left(<[\mu,p],[\iota_{\xi}A_0, \phi]>+<[A_0, \iota_{\xi}p],[\mu,\phi]>\right)\\
        &= M^A_{\mathrm{L}_\xi^{A_0}\mu}- \int_\Sigma <p, [\mathrm{L}_\xi^{A_0}\mu, \phi]> =M^{H,A}_{\mathrm{L}_\xi^{A_0}\mu},
    \end{align*}
    where we noticed that in the second step the first and third line cancel and used that $$-<[\mu,p],\mathrm{L}_\xi^{\omega_0+A_0}\phi>=<p,[\mu,\mathrm{L}_\xi^{\omega_0+A_0}\phi]>.$$
    \begin{align*}
    \{ M_\mu^{H,A}, H_\lambda^{H,A} \}_{YMH}&=
    \{M_{\mu}^{A}, H^A_{\lambda}\}_A + \iota_{\mathbb{h}^{H, A}}\delta (M_{\mu}^{A}+m_{\mu}^{H, A}) + \iota_{\mathbb{m}^{H, A}}\delta (H_{\lambda}+ h_{\lambda}^{A}+h_{\lambda}^{H}+h_{\lambda}^{H, A})\\
    &\phantom{=}- \iota_{\mathbb{m}^{H, A}}\iota_{\mathbb{h}^{H, A}}\Omega_{YMH}\\
    &=  \int_{\Sigma}\frac{\lambda e_n e^2}{2}\left(<\Pi, [d_A \mu, \phi]> +<[A,\Pi], [\mu, \phi]> + <\Pi, [\mu,[A, \phi]]>\right)\\
    &\phantom{=}-\int_\Sigma \frac{\lambda e_n e^2}{2}\left(< [\mu, \Pi],d\phi>+<\Pi,d[\mu,\phi]>\right)+ \frac{\lambda e_n e^3}{3}<[\mu,\Pi], \Pi> \\
    &\phantom{=}+\int_\Sigma \frac{\lambda e_n e^3 q_H}{3}(<\phi,\phi>-v^2)<[\mu,\phi],\phi>\\
    &\phantom{=}+ \int_\Sigma \frac{\lambda e_n e^2}{2} \left(<[\mu,\Pi],[A,\phi]>+<\Pi,[A,[\mu,\phi]]>\right) \\
    &\phantom{=} -\int_\Sigma \frac{\lambda e_n e^2}{2} \left(<[\mu,\Pi],[A,\phi]>+<[A,\Pi], [\mu, \phi]>\right)=0,
        \end{align*}
    having used the fact that $[\mu, <\phi,\phi>]=[\mu, <\Pi,\Pi>]=0$ and the Jacobi identity for $A, \mu$ and $\phi$. 
    \begin{align*}
            \{L^{A,H}_c,L^{A,H}_c\}_{YMH}&=\{L^H_c, L_c^H\}_{H}+\{L_c^A, L_c^A\}_{A}-\{L_c,L_c\}\\
            &=-\frac{1}{2} L_{[c,c]}=-\frac{1}{2} L^{A,H}_{[c,c]}; 
    \end{align*}
    \begin{align*}
            \{L^{A,H}_c,P_\xi^{H,A}\}_{YMH}&=\{ L_c, P_\xi + p_\xi^H \}_H + \{ L_c, P_\xi + p_\xi^A \}_A - \{L_c, P_\xi\}+ \iota_{\mathbb{p}^{H,A}}\delta L_c= L_{\mathrm{L}^{\omega_0}c}^{H,A}
    \end{align*}
We have 
\begin{align*}
    \iota_{\mathbb{h}^{H,A}}\delta L_c= \int_\Sigma c \left[\frac{\lambda e_n}{4}<\Pi, [A,\phi]>,e^2\right]=\int_\Sigma \left[c , \lambda e_n\right] \frac{ e^2 }{4}<\Pi, [A,\phi]>.
\end{align*}
    Hence, using $[c,\lambda e_n]=X= X^{(\nu)}e_\nu+ X^{(n)}e_n$ and 
    \begin{align*}
    - p^{H, A}_{X^{(\nu)}}  + m^{H,A}_{X^{(\nu)}(A-A_0)_{\nu}}=\int_\Sigma \left[c , \lambda e_n\right]^{(\nu)}e_\nu \frac{ e^2 }{4}<\Pi, [A,\phi]>
\end{align*}
    we get
    \begin{align*}
        \{L^{A,H}_c,H_\lambda^{H,A}\}_{YMH}&=\{L_c, H_{\lambda}+h_{\lambda}^{A}\}_A+\{L_c, H_{\lambda}+h_{\lambda}^{H}\}_H-\{L_c,H_{\lambda}\}+ \iota_{\mathbb{h}^{H,A}}\delta L_c\\
        &=- P^{A}_{X^{(\nu)}} + L^{A}_{X^{(\nu)}(\omega - \omega_0)_\nu} - H^{A}_{X^{(n)}} + M^{A}_{X^{\nu}(A-A_0)_\nu}- P^{H}_{X^{(\nu)}} + L^{H}_{X^{(\nu)}(\omega - \omega_0)_\nu}\\
        &\phantom{=} - H^{H}_{X^{(n)}}+ P_{X^{(\nu)}} - L_{X^{(\nu)}(\omega - \omega_0)_\nu} + H_{X^{(n)}} - p^{H, A}_{X^{(\nu)}}  -h^{H,A}_{X^{(n)}} + m^{H,A}_{X^{(\nu)}(A-A_0)_{\nu}}\\
        &=- P^{H,A}_{X^{(\nu)}} + L^{H,A}_{X^{(\nu)}(\omega - \omega_0)_\nu} - H^{H,A}_{X^{(n)}} + M^{H,A}_{X^{\nu}(A-A_0)_\nu}
        \end{align*}
     To compute the following bracket, we make use of the following identity
     \begin{align*}
         \frac{1}{2}\iota_{[\xi,\xi]}A_0= \iota_{\xi} d \iota_{\xi} A_0 -\frac{1}{2}\iota_{\xi}\iota_{\xi}dA_0.
     \end{align*}
    \begin{align*}
        \{P_\xi^{H,A},P_\xi^{H,A}\}_{YMH}&=
        \{P_{\xi}+p_{\xi}^{A}, P_{\xi}+p_{\xi}^{A}\}_A+\{P_{\xi}+p_{\xi}^{H}, P_{\xi}+p_{\xi}^{H}\}_H-\{P_{\xi},P_{\xi}\}\\
        &\phantom{=}+ 2\iota_{\mathbb{p}^{H, A}}\delta (P_{\xi}+ p_{\xi}^{A}+p_{\xi}^{H}+p_{\xi}^{H, A}) - \iota_{\mathbb{p}^{H, A}}\iota_{\mathbb{p}^{H, A}}\Omega_{YMS}\\
        &=\frac{1}{2}P^{A}_{[\xi, \xi]}- \frac{1}{2}L^{A}_{\iota_{\xi}\iota_{\xi}F_{\omega_0}} - \frac{1}{2}M^{A}_{\iota_{\xi}\iota_{\xi}F_{A_0}} + \frac{1}{2}P^{\psi}_{[\xi, \xi]}- \frac{1}{2}L^{\psi}_{\iota_{\xi}\iota_{\xi}F_{\omega_0}}-\frac{1}{2}P_{[\xi, \xi]}\\
        &\phantom{=}+ \frac{1}{2}L_{\iota_{\xi}\iota_{\xi}F_{\omega_0}}+ \int_\Sigma < [A_0, \iota_\xi p],  \mathrm{L}^{\omega_0+A_0}_{\xi}\phi> +<p,\mathrm{L}^{\omega_0+A_0}_{\xi}[\iota_{\xi}A_0, \phi]>\\
        &\phantom{=}-\int_\Sigma  < [A_0, \iota_\xi p],[\iota_{\xi}A_0, \phi]>\\
        &=\frac{1}{2}P^{A}_{[\xi, \xi]}- \frac{1}{2}L^{A}_{\iota_{\xi}\iota_{\xi}F_{\omega_0}} - \frac{1}{2}M^{A}_{\iota_{\xi}\iota_{\xi}F_{A_0}} + \frac{1}{2}P^{\psi}_{[\xi, \xi]}- \frac{1}{2}L^{\psi}_{\iota_{\xi}\iota_{\xi}F_{\omega_0}}-\frac{1}{2}P_{[\xi, \xi]}\\
        &\phantom{=}+ \frac{1}{2}L_{\iota_{\xi}\iota_{\xi}F_{\omega_0}}+ \int_\Sigma <p,[\iota_{[\xi,\xi]}A_0, \phi]>\\
        &=\frac{1}{2}P^{H,A}_{[\xi, \xi]}- \frac{1}{2}L^{H,A}_{\iota_{\xi}\iota_{\xi}F_{\omega_0}} - \frac{1}{2}M^{H,A}_{\iota_{\xi}\iota_{\xi}F_{A_0}};
        \end{align*}
    \begin{align*}
        \{P_\xi^{H,A},H_\lambda^{H,A}\}_{YMH}&=\{P_{\xi}+p_{\xi}^{A}, H_{\lambda}+h_{\lambda}^{A}\}_A+\{P_{\xi}+p_{\xi}^{H}, H_{\lambda}+h_{\lambda}^{H}\}_H-\{P_{\xi},H_{\lambda}\}\\
        &\phantom{=}+ \iota_{\mathbb{h}^{H, A}}\delta (P_{\xi}+ p_{\xi}^{A}+p_{\xi}^{H}+p_{\xi}^{H, A}) + \iota_{\mathbb{p}^{H, A}}\delta (H_{\lambda}+ h_{\lambda}^{A}+h_{\lambda}^{H}+h_{\lambda}^{H, A})\\
        &\phantom{=}- \iota_{\mathbb{p}^{H, A}}\iota_{\mathbb{h}^{H, A}}\Omega_{YMH}\\
        &= P^{A}_{Y^{(\nu)}} -L^{A}_{ Y^{(\nu)} (\omega - \omega_0)_\nu} +H^{A}_{ Y^{(n)}} - M^{A}_{Y^\nu (A-A_0)_\nu}+ P^{H}_{Y^{(\nu)}} -L^{H}_{ Y^{(\nu)} (\omega - \omega_0)_\nu}\\
        &\phantom{=}+H^{H}_{ Y^{(n)}} - P_{Y^{(\nu)}} +L_{ Y^{(\nu)} (\omega - \omega_0)_\nu} -H_{ Y^{(n)}}+ \int_\Sigma \mathrm{L}^{\omega_0+A_0}_\xi e \frac{\lambda e_n e}{2}<\Pi, [A, \phi]>\\
        &\phantom{=}+ \int_\Sigma \frac{\lambda e_n e^2}{2}<[A,\Pi], \mathrm{L}^{\omega_0+A_0}_\xi\phi>  -
        \frac{\lambda e_n e^2}{2}<\Pi, \mathrm{L}^{\omega_0+A_0}_\xi[A,\phi]>\\
        &\phantom{=}+ \int_\Sigma <\Pi, \mathrm{L}^{\omega_0+A_0}_\xi\left(\frac{\lambda e_n e^2}{2}[A,\phi]\right)> + \frac{\lambda e_n e^2}{2}<[\iota_{\xi}A_0,\Pi],d_A\phi>\\
        &\phantom{=}+ \int_\Sigma \frac{\lambda e_n e^2}{2}\left(<\Pi,d_A[\iota_{\xi}A_0,\phi]>\right)\\
        &\phantom{=}+\int_\Sigma\frac{\lambda e_n e^3}{3}<[A,\Pi], \Pi> +\frac{\lambda e_n e^3 q_H}{3}(<\phi,\phi>-v^2)<[A,\phi],\phi>\\
         &= P^{A}_{Y^{(\nu)}} -L^{A}_{ Y^{(\nu)} (\omega - \omega_0)_\nu} +H^{A}_{ Y^{(n)}} - M^{A}_{Y^\nu (A-A_0)_\nu}+ P^{H}_{Y^{(\nu)}} -L^{H}_{ Y^{(\nu)} (\omega - \omega_0)_\nu}\\
        &\phantom{=}+H^{H}_{ Y^{(n)}} - P_{Y^{(\nu)}} +L_{ Y^{(\nu)} (\omega - \omega_0)_\nu} -H_{ Y^{(n)}}\\
        &\phantom{=}+\int_\Sigma \mathrm{L}^{\omega_0+A_0}_\xi(\lambda e_n)\frac{e^2}{2} <\Pi,[A,\phi]>\\    
            &= P^{H,A}_{Y^{(\nu)}} -L^{H,A}_{ Y^{(\nu)} (\omega - \omega_0)_\nu} +H^{H,A}_{ Y^{(n)}} - M^{H,A}_{Y^\nu (A-A_0)_\nu},
        \end{align*}
    having used  \eqref{id: comm L_Ad_Aphi}.
        \begin{align*}
            \{H_\lambda^{H,A},H_\lambda^{H,A}\}_{YMH}&=\{H_\lambda^{H},H_\lambda^{H}\}_H + \{H_\lambda^{A},H_\lambda^{A}\}_A - \{H_\lambda,H_\lambda\} \\
            &\quad + 2 \iota_{\mathbb{h}_\lambda^{H,A}}\delta H_\lambda^{H,A} - \iota_{\mathbb{h}_\lambda^{H,A}}\iota_{\mathbb{h}_\lambda^{H,A}}\Omega_{YMH}=0,
        \end{align*}
    in fact all terms in   $2 \iota_{\mathbb{h}_\lambda^{H,A}}\delta H_\lambda^{H,A} - \iota_{\mathbb{h}_\lambda^{H,A}}\iota_{\mathbb{h}_\lambda^{H,A}}\Omega_{SM}$ contain either $(\lambda e_n)^2=0$ or $\lambda e_n d_\omega(\lambda e_n)=0$.
\end{proof}

\begin{remark}
    Notice that, after choosing $U(1)$ as the gauge group and setting $q_H=0$, we obtain scalar electrodynamics coupled to gravity as a particular case of this theory.
\end{remark}

\subsection{Yukawa interaction}\label{s:scalar+spinor}
The interaction of a scalar field and a spinor field takes the name of Yukawa interaction.

As before, let us denote the spinor field by $\psi\in S(M)$ and a scalar field by $\phi \in C^{\infty}(M)$. Then an action on the bulk for the Yukawa interaction (i.e. an action correctly reproducing the classical Euler--Lagrange equations) is given by
\begin{align*}
    S_Y = S+ S_{\psi} + S_{\phi} + g_Y \int_M \frac{1}{2N!} e^N \overline{\psi}\phi \psi 
\end{align*}
where $g_Y$ is a coupling constant, $S_{\phi}$ is defined in \eqref{e:action_scalar} and $S_{\psi}$  in \eqref{e:action_spin}. Since the additional interaction term does not have derivatives, the geometric phase space will be just the direct sum of the building blocks composing this theory. In particular we have that the geometric phase space is the bundle
\begin{align*}
    {F}^{\partial}_{Y} \rightarrow \Omega_{e_n}^1(\Sigma, \mathcal{V}_{\Sigma})\oplus \mathcal{C^{\infty}}(\Sigma)\times S(\Sigma)\times \overline{S}(\Sigma)
\end{align*}
with fiber $\mathcal{A}_{red}(\Sigma)\oplus \Omega^{(0,1)}_{\Sigma,red}$ such that  \eqref{e:constraintscalar} and  \eqref{e:omegareprfix2spin} are satisfied.

The corresponding symplectic form is again just the sum of the symplectic form of the building blocks:
\begin{align*}
\Omega_{Y} = \varpi + \varpi_{\psi}+ \varpi_{\phi}.
\end{align*}

Let us now consider the constraints. By the previous computation we already know that for a spinor field and a scalar field coupled to gravity the functions building up the constraints are those corresponding to the variations $\delta e $ and $\delta \omega$, since all the other ones are evolution equations. From the expression of the action of the Yukawa interaction, we deduce that there are no other constraints than the pure gravity ones and that the only one that is modified is $H_{\lambda}$. 
Let
\begin{align}\label{e:constraints_Yukawa3}
    h^{\phi, \psi}_{\lambda}:= g_Y \int_{\Sigma} \lambda e_n \frac{1}{2(N-1)!} e^{(N-1)} \overline{\psi}\phi \psi 
\end{align}
Then the constraints for Yukawa theory are:
\begin{align*}
    L^{\phi, \psi}_c= L_c + l^{\psi}_{c}; \quad 
    P^{\phi, \psi}_{\xi}= P_{\xi}+ p^{\phi}_{\xi} + p^{\psi}_{\xi}; \quad
    H^{\phi, \psi}_{\lambda}= H_{\lambda}+ h^{\phi}_{\lambda} + h^{\psi}_{\lambda}+ h^{\phi, \psi}_{\lambda}.
\end{align*}
Once again, these constraints define a coisotropic submanifold with Poisson brackets analogous to the gravity case, as specified by the following theorem.
\begin{theorem} \label{thm:first-class-constraints_Yukawa}
 Assume that $g^\partial$ is non-degenerate on $\Sigma$. Then, the zero locus of the functions  $L^{\phi, \psi}_c$, $P^{\phi, \psi}_{\xi}$,  $H^{\phi, \psi}_{\lambda}$ defined above is  coisotropic submanifold with respect to the symplectic structure $\Omega_{Y}$. Their mutual Poisson brackets read
 \begin{align*}
     \{L^{\phi, \psi}_c, L^{\phi, \psi}_c\}_{Y} &= - \frac{1}{2} L^{\phi, \psi}_{[c,c]}  & \{L^{\phi, \psi}_c, P^{\phi, \psi}_{\xi}\}_{Y}  &=  L^{\phi, \psi}_{\mathrm{L}_{\xi}^{\omega_0}c}\\
     \{L^{\phi, \psi}_c,  H^{\phi, \psi}_{\lambda}\}_{Y}  &= - P^{\phi, \psi}_{X^{(\nu)}} + L^{\phi, \psi}_{X^{(\nu)}(\omega - \omega_0)_{\nu}} - H^{\phi, \psi}_{X^{(n)}} & \{P^{\phi, \psi}_{\xi}, P^{\phi, \psi}_{\xi}\}_{Y}  &=  \frac{1}{2}P^{\phi, \psi}_{[\xi, \xi]}- \frac{1}{2}L^{\phi, \psi}_{\iota_{\xi}\iota_{\xi}F_{\omega_0}}\\
     \{P^{\phi, \psi}_{\xi},H^{\phi, \psi}_{\lambda}\}_{Y}  &=  P^{\phi, \psi}_{Y^{(\nu)}} -L^{\phi, \psi}_{ Y^{(\nu)} (\omega - \omega_0)_{\nu}} +H^{\phi, \psi}_{ Y^{(n)}} & \{H^{\phi, \psi}_{\lambda},H^{\phi, \psi}_{\lambda}\}_{Y}  &=0
 \end{align*}
 with the same notation as in Theorem \ref{thm:first-class-constraints}. 
\end{theorem}
\begin{proof}
We use the results proved in Section \ref{s:appendix_Poisson_brackets}.
The expressions of $\mathbb{L}$,$\mathbb{l}^{\phi}$, $\mathbb{l}^{ \psi}$, $\mathbb{P}$, $\mathbb{p}^{\phi}$, $\mathbb{p}^{\psi}$, $\mathbb{H}$, $\mathbb{h}^{\phi}$ and $\mathbb{h}^{\psi}$ have been computed in \cite{CCF2022} and are collected in Appendix \ref{s:appendix_Hamiltonianvf}. 
Let us start with the constraint $L^{\phi, \psi}_c$. We first notice that $\iota_{\mathbb{l}^{\phi}} \varpi_{\psi}= \iota_{\mathbb{l}^{\psi}} \varpi_{\phi}=0$. The variation of the interaction term is $l^{\phi, \psi}_{c}=0$, hence we also conclude
\begin{align*}
    \mathbb{l}^{\phi, \psi}=0.
\end{align*}
For $P^{\phi, \psi}_{\xi}$, we work in the same way and find
\begin{align*}
    \mathbb{p}^{\phi, \psi}=0.
\end{align*}
On the other hand, for $H^{\phi, \psi}_{\lambda}$ we get $\iota_{\mathbb{h}^{\phi}}\varpi_{\psi}=0$ and 
\begin{align*}
    \iota_{\mathbb{h}^{\psi}}\varpi_{\phi}= \int_{\Sigma}\frac{1}{2}e^2\mathbb{h}_e^{\psi} \Pi \delta \phi.
\end{align*}
Note that we will not need the explicit expression.
The variation of the interaction term reads
\begin{align*}
    \delta h^{\phi, \psi}_{\lambda}:= g_Y \int_{\Sigma} \lambda e_n \left( \frac{1}{4}e^{2} \delta e \overline{\psi}\phi \psi
    +\frac{1}{2\cdot 3!} e^3 \delta\overline{\psi}\phi \psi
    -\frac{1}{2\cdot 3!} e^3 \overline{\psi}\delta \phi \psi
    -\frac{1}{2\cdot 3!} e^3 \overline{\psi}\phi \delta \psi \right). 
\end{align*}
    Hence we get:
    \begin{align*}
	    \mathbb{h}^{\phi,\psi}_e &= 0 &
	    e \mathbb{h}^{\phi,\psi}_\omega &= 0 \\
        \mathbb{h}^{\phi,\psi}_\phi &= 0 &
	    \frac{1}{ 3!}e^3\mathbb{h}^{\phi,\psi}_\Pi &=-\frac{1}{2}e^2\mathbb{h}_e^{\psi} \Pi +\frac{1}{2\cdot 3!}g_Y \lambda e_n e^3 \overline{\psi} \psi \\
	    \gamma \mathbb{h}^{\phi,\psi}_{\psi} &= -\frac{i}{2}g_Y \lambda e_n\phi \psi  &
	    \mathbb{h}^{\phi,\psi}_{\overline{\psi}} \gamma &= \frac{i}{2}g_Y \lambda e_n\overline{\psi}\phi
	\end{align*}

In order to compute the Poisson brackets of the constraints we use Theorem \ref{thm:Poisson_Brackets_couples}. 
We first note that for all constraints we have 
\begin{align*}
    \iota_{\mathbb{X}+\mathbb{x}^{\psi}}\iota_{\mathbb{Y}+\mathbb{y}^{\psi}}\varpi_{\phi}+ \iota_{\mathbb{X}+\mathbb{x}^{\phi}}\iota_{\mathbb{Y}+\mathbb{y}^{\phi}}\varpi_{\psi}=0.
\end{align*}
For the computations we use repeatedly the results of Theorem\ref{thm:first-class-constraints} and the corresponding results in presence of a scalar field and and spinor field. Not writing the zero terms we have
\begin{align*}
    \{L^{\phi, \psi}_c,L^{\phi, \psi}_c\}_Y&=\{L_c+l_c^{\phi}, L_c+l_c^{\phi}\}_{\phi}+\{L_c+l_c^{\psi}, L_c+l_c^{\psi}\}_{\psi}-\{L_c,L_c\}\\
    & +\iota_{\mathbb{l}^{\phi}}\iota_{\mathbb{l}^{\psi}}\Omega_{Y}+\iota_{\mathbb{l}^{\psi}}\iota_{\mathbb{l}^{\phi}}\Omega_{Y}\\
    &=-\frac{1}{2} L_{[c,c]} -\frac{1}{2}  l_{[c,c]}^{\psi} = -\frac{1}{2}  L^{\phi, \psi}_{[c,c]}
\end{align*}
where we used for the second line $\mathbb{l}_e^{\phi}=\mathbb{l}_e^{\psi}=0$.
Similarly
\begin{align*}
    \{L^{\phi, \psi}_c,P^{\phi, \psi}_\xi\}_Y&=\{L_c+l_c^{\phi}, P_\xi+p_\xi^{\phi}\}_{\phi}+\{L_c+l_c^{\psi}, P_\xi+p_\xi^{\psi}\}_{\psi}-\{L_c,P_\xi\}\\
    & +\iota_{\mathbb{l}^{\phi}}\iota_{\mathbb{p}^{\psi}}\Omega_{Y}+\iota_{\mathbb{l}^{\psi}}\iota_{\mathbb{p}^{\phi}}\Omega_{Y}\\
    &= L_{\mathrm{L}_{\xi}^{\omega_0}c}+ l_{\mathrm{L}_{\xi}^{\omega_0}c}^{\psi} = L^{\phi, \psi}_{\mathrm{L}_{\xi}^{\omega_0}c}
\end{align*}
where we used for the second line $\mathbb{l}_e^{\phi}=\mathbb{p}_e^{\phi}=\mathbb{l}_e^{\psi}=\mathbb{p}_e^{\psi}=0$, and
\begin{align*}
    \{P^{\phi, \psi}_{\xi}, P^{\phi, \psi}_{\xi}\}_Y &=\{P_\xi+p_\xi^{\phi}, P_\xi+p_\xi^{\phi}\}_{\phi}+\{P_\xi+p_\xi^{\psi}, P_\xi+p_\xi^{\psi}\}_{\psi}-\{P_\xi,P_\xi\}\\
    & +\iota_{\mathbb{p}^{\phi}}\iota_{\mathbb{p}^{\psi}}\Omega_{Y}+\iota_{\mathbb{p}^{\psi}}\iota_{\mathbb{p}^{\phi}}\Omega_{Y}\\
    &= \frac{1}{2}P_{[\xi, \xi]}- \frac{1}{2}L_{\iota_{\xi}\iota_{\xi}F_{\omega_0}}+
    \frac{1}{2}p^{\phi}_{[\xi, \xi]}+
    \frac{1}{2}p^{\psi}_{[\xi, \xi]}- \frac{1}{2}l^{\psi}_{\iota_{\xi}\iota_{\xi}F_{\omega_0}}\\
    &=\frac{1}{2}P^{\phi, \psi}_{[\xi, \xi]}- \frac{1}{2}L^{\phi, \psi}_{\iota_{\xi}\iota_{\xi}F_{\omega_0}}
\end{align*}
For the brackets with $H^{\phi, \psi}_{\lambda}$ we have to take into account more terms:
\begin{align*}
    \{L^{\phi, \psi}_c,  H^{\phi, \psi}_{\lambda}\}_Y  &=\{L_c+l_c^{\phi}, H_{\lambda}+h_{\lambda}^{\phi}\}_{\phi}+\{L_c+l_c^{\psi}, H_{\lambda}+h_{\lambda}^{\psi}\}_{\psi}-\{L_c,H_{\lambda}\}\\
    &+\iota_{\mathbb{h}^{\phi, \psi}}\delta (L_c+l_c^{\psi})
\end{align*}    
Let us compute the last term:
\begin{align*}
    \iota_{\mathbb{h}^{\phi, \psi}}\delta (L_c+l_c^{\psi})&= \int_{\Sigma} - c e [\mathbb{h}^{\phi, \psi}_{\omega}, e] - i \frac{e^3}{2\cdot 3!} \left( -[c,\mathbb{h}^{\phi, \psi}_{\overline{\psi}}]\gamma \psi - \mathbb{h}^{\phi, \psi}_{\overline{\psi}} \gamma [c,\psi] \right)\\
    &- i \frac{e^3}{2\cdot 3!} \left(- [c,\overline{\psi}]\gamma \mathbb{h}^{\phi, \psi}_\psi + \overline{\psi} \gamma [c,\mathbb{h}^{\phi, \psi}_\psi] \right)\\
    &= \int_{\Sigma} - \frac{1}{4} g_Y c e [ \lambda e_n  e \overline{\psi}\phi \psi, e] - i \frac{e^3}{2\cdot 3!} \left( -[c,\mathbb{h}^{\phi, \psi}_{\overline{\psi}}]\gamma \psi - \mathbb{h}^{\phi, \psi}_{\overline{\psi}} \gamma [c,\psi] \right)\\
    &- i \frac{e^3}{2\cdot 3!} \left(- [c,\overline{\psi}]\gamma \mathbb{h}^{\phi, \psi}_\psi + \overline{\psi} \gamma [c,\mathbb{h}^{\phi, \psi}_\psi] \right)\\
    &= -\int_{\Sigma}\frac{e^3}{2\cdot 3!} g_Y [c,\lambda e_n]  \overline{\psi}\phi \psi
\end{align*}
Where we used the properties listed in \cite[Appendix B.5]{CCF2022} for the bracket $[\cdot, \cdot]$ on spinors. The result is easily recognized as
 $h^{\phi, \psi}_{X^{(n)}}$. Hence we get
\begin{align*}
    \{L^{\phi, \psi}_c,  H^{\phi, \psi}_{\lambda}\}_Y&=  - P_{X^{(\nu)}} + L_{X^{(\nu)}(\omega - \omega_0)_{\nu}} - H_{X^{(n)}}
    - p^{\phi}_{X^{(\nu)}} - h^{\phi}_{X^{(n)}}
    - p^{\psi}_{X^{(\nu)}} \\
    & \quad + l^{\psi}_{X^{(\nu)}(\omega - \omega_0)_{\nu}} - h^{\psi}_{X^{(n)}}
    -h^{\phi, \psi}_{X^{(n)}}\\
    &=- P^{\phi, \psi}_{X^{(\nu)}} + L^{\phi, \psi}_{X^{(\nu)}(\omega - \omega_0)_{\nu}} - H^{\phi, \psi}_{X^{(n)}}
\end{align*} 
Analogously we compute
\begin{align*}
    \{P^{\phi, \psi}_\xi,  H^{\phi, \psi}_{\lambda}\}_Y  &=\{P_\xi+p_\xi^{\phi}, H_{\lambda}+h_{\lambda}^{\phi}\}_{\phi}+\{P_\xi+p_\xi^{\psi}, H_{\lambda}+h_{\lambda}^{\psi}\}_{\psi}-\{P_\xi,H_{\lambda}\}\\
    &+ \iota_{\mathbb{h}^{\phi, \psi}}\delta (P_\xi+p_\xi^{\psi}+p_\xi^{\phi})+\iota_{\mathbb{h}^{\psi}}\iota_{\mathbb{p}^{\phi}}\Omega_{Y}
    \end{align*}
    Let us consider the terms in the second row. We have $$\iota_{\mathbb{h}^{\psi}}\iota_{\mathbb{p}^{\phi}}\Omega_{Y}= \int_{\Sigma}\frac{1}{2}e^2 \mathbb{h}^{\psi}_e \Pi \mathbb{p}^{\phi}_{\phi}=
    -\int_{\Sigma}\frac{1}{2}e^2 \mathbb{h}^{\psi}_e \Pi \mathrm{L}_\xi \phi.$$ On the other hand we have
    \begin{align*}
    \iota_{\mathbb{h}^{\phi, \psi}}\delta (P_\xi+p_\xi^{\psi}+p_\xi^{\phi})&= \int_{\Sigma} - \frac{1}{3!}e^3 \mathbb{h}^{\phi, \psi}_{\Pi}\mathrm{L}_{\xi}\phi - i \frac{e^3}{2\cdot 3!} \left( \mathbb{h}^{\phi, \psi}_{\overline{\psi}} \gamma \mathrm{L}_{\xi}^{\omega_0}(\psi) + \mathrm{L}_\xi^{\omega_0}(\overline{\psi})\gamma \mathbb{h}^{\phi, \psi}_{\psi}  \right)\\
    &\quad 
    - \frac{i}{2\cdot 3!} \left(\mathrm{L}_{\xi}^{\omega_0}( e^3 \overline{\psi}) \gamma \mathbb{h}^{\phi, \psi}_{\psi} + \mathbb{h}^{\phi, \psi}_{\overline{\psi}}\gamma \mathrm{L}_\xi^{\omega_0}(e^3 \psi)  \right)\\
    &= \int_{\Sigma} \frac{1}{2}e^2\mathbb{h}_e^{\psi} \Pi \mathrm{L}_{\xi}\phi-\frac{1}{2\cdot 3!}g_Y \lambda e_n e^3 \overline{\psi} \psi\mathrm{L}_{\xi}\phi    \\
    &\quad + \frac{e^3}{2\cdot 3!} \left( \frac{1}{2}g_Y \lambda e_n\overline{\psi}\phi \mathrm{L}_{\xi}^{\omega_0}(\psi) - \frac{1}{2}g_Y \lambda e_n\mathrm{L}_\xi^{\omega_0}(\overline{\psi})\phi \psi  \right)\\
    &\quad 
    -  \frac{1}{2\cdot 3!} \left(\frac{1}{2}g_Y \lambda e_n\mathrm{L}_{\xi}^{\omega_0}( e^3 \overline{\psi}) \phi \psi -\frac{1}{2}g_Y \lambda e_n\overline{\psi}\phi\mathrm{L}_\xi^{\omega_0}(e^3\psi)  \right)\\
    &= \int_{\Sigma} - g_Y \frac{1}{2\cdot 3!}  \lambda e_n \mathrm{L}_{\xi}^{\omega_0}( e^3 \overline{\psi}\phi \psi)+\frac{1}{2}e^2\mathbb{h}_e^{\psi} \Pi \mathrm{L}_{\xi}\phi\\
    &= \int_{\Sigma} g_Y \frac{1}{2\cdot 3!} \mathrm{L}_{\xi}^{\omega_0} (\lambda e_n ) e^3 \overline{\psi}\phi \psi+\frac{1}{2}e^2\mathbb{h}_e^{\psi} \Pi \mathrm{L}_{\xi}\phi.
    \end{align*}
    where we used $\mathrm{L}_\xi^{\omega_0} \gamma=0.$
    The second term in the last row cancels out with the one computed above, while the first is exactly $h^{\phi, \psi}_{ Y^{(n)}}$. Hence collecting all the terms we get
    \begin{align*}
    \{P^{\phi, \psi}_\xi,  H^{\phi, \psi}_{\lambda}\}_Y &=  P_{Y^{(\nu)}} - L_{Y^{(\nu)}(\omega - \omega_0)_{\nu}} + H_{Y^{(n)}}
    + p^{\phi}_{Y^{(\nu)}} + h^{\phi}_{Y^{(n)}}
    + p^{\psi}_{Y^{(\nu)}}\\
    &\quad  - l^{\psi}_{Y^{(\nu)}(\omega - \omega_0)_{\nu}} + h^{\psi}_{Y^{(n)}}
    +h^{\phi, \psi}_{Y^{(n)}}\\
    &= - P^{\phi, \psi}_{Y^{(\nu)}} + L^{\phi, \psi}_{ Y^{(\nu)} (\omega - \omega_0)_{\nu}} -H^{\phi, \psi}_{ Y^{(n)}}.
\end{align*}

Finally we consider
 \begin{align*}
    \{H^{\phi, \psi}_{\lambda},H^{\phi, \psi}_{\lambda}\}_Y &=\{H_{\lambda}+h^{\phi}_{\lambda}, H_{\lambda}+h^{\phi}_{\lambda}\}_{\phi}+\{H_{\lambda}+h^{\psi}_{\lambda}, H_{\lambda}+h^{\psi}_{\lambda}\}_{\psi}-\{H_{\lambda},H_{\lambda}\}\\
    &+ 2\iota_{\mathbb{h}^{\phi, \psi}}\delta (H_{\lambda}+ h^{\phi}_{\lambda}+h^{\psi}_{\lambda}+h^{\phi,\psi}_{\lambda})- \iota_{\mathbb{h}^{\phi, \psi}}\iota_{\mathbb{h}^{\phi, \psi}}\Omega_{Y}+2\iota_{\mathbb{h}^{\phi}}\iota_{\mathbb{h}^{\psi}}\Omega_{Y}.
\end{align*}
All the terms of the first line are zero, as proved in the previous theorems. Furthermore notice that $H_{\lambda}+ h^{\phi}_{\lambda}+h^{\psi}_{\lambda}+h^{\phi,\psi}_{\lambda}$ is proportional to $\lambda$, as well as $\mathbb{h}^{\phi, \psi}$, $\mathbb{h}^{\phi}$ and  $\mathbb{h}^{\psi}$. Since $\lambda^2=0$, we conclude that $\{H^{\phi, \psi}_{\lambda},H^{\phi, \psi}_{\lambda}\}=0.$
\end{proof}

\section{Standard Model}\label{s:standard_model}
We now consider the full theory of gravity coupled to a Yang--Mills, a Higgs scalar and a spinor fields as defined in the previous sections. One recovers the usual structure of the standard model by choosing the appropriate gauge group and its representation on the various fields. 
The action reads:
    \begin{equation*}
        S_{SM}=S+S_A+S_{\psi}+S_H+S_{\psi,A}+S_{H,\psi}.
    \end{equation*}

In order to have a well defined Yukawa interaction term $S_{H,\psi}$, the space of Dirac spinors is split as the direct sum of left-handed and right-handed Weyl spinors, i.e.
    \begin{equation*}
        S=S_R\oplus S_L,
    \end{equation*}
corresponding respectively to the projections of $S$ onto the eigenspaces of $\gamma_5:=-\Pi_{a=0}^3 \gamma_a$. In particular, defining the projectors $P_{R/L}:=\frac{\mathbb{1}\pm \gamma^5}{2}$, for any $\psi\in S$, we obtain
    \begin{equation*}
        \psi_L=P_L(\psi)\in S_L \qquad \psi_R=P_R(\psi)\in S_R.
    \end{equation*}
We choose $\psi_R$ to be in the same representation space of $G$ as the Higgs field $\phi$,\footnote{In particular $\psi_R\in\Gamma(S_R\otimes E_n)$.} while $\psi_L$ is set to be in the trivial one. Hence, the Yukawa interaction term is
    \begin{equation*}
        S_{H,\psi}=\int_M \frac{e^4}{2\cdot 4!} g_Y \left[\bar{\psi}_L < \phi, \psi_R> - <\phi,\bar{\psi}_R>\psi_L \right].
    \end{equation*}
On the boundary $\Sigma$, the geometric phase space is given by 
    \begin{equation*}
        F^\partial_{SM}\rightarrow \Omega^1_{e_n}(\Sigma,\mathcal{V}_{\Sigma}) \oplus \mathcal{A}_\Sigma^G \oplus \Gamma(\Sigma, (S_R\otimes E_n)) \oplus \Gamma(\Sigma, (\bar{S}_R\otimes E_n)) \oplus \Gamma(\Sigma,{S}_L) \oplus\Gamma(\Sigma,\bar{S}_L) \oplus \Gamma(\Sigma,E_n),
    \end{equation*}
with fiber given by $\mathcal{A}_\Sigma \oplus \Omega^{(0,2)}_{\partial,\mathrm{red}}(\mathfrak{g}) \oplus \Omega^{(0,1)}_{\partial,\mathrm{red}}(E_n) $ satisfying \eqref{e:omegareprfix2spin},\eqref{e:constraintYM} and \eqref{e:constraintscalar}.

The corresponding symplectic structure is simply given by 
    \begin{equation*}
        \Omega_{SM}=\varpi+\varpi_A+\varpi_{\psi_L}+\varpi_{\psi_R}+\varpi_H,
    \end{equation*}
where
    \begin{equation*}
        \varpi_{\psi_R}=\int_\Sigma \frac{i}{4}e^2 (<\delta \bar{\psi}_R  \gamma, \psi_R> - <\bar{\psi}_R  \gamma,\delta \psi_R>)\delta e + \frac{i}{3!}e^3 <\delta\bar{\psi}_R \gamma,\delta\psi_R>
    \end{equation*}
while, defining $x^{\psi_{L/R}}:=x^{\psi_{L}}+x^{\psi_{R}}$, the constraints are
    \begin{align*}
        L_c^{SM}&=L_c+l_c^{\psi_{L/R}}, \\
        P_\xi^{SM}&=P_\xi + p_\xi^A + p_\xi^{\psi_{L/R}} + p_\xi^H + p_\xi^{A,{\psi_{L/R}}} + p_\xi^{H,A}, \\
        M_\mu^{SM}&= M_\mu^A + m_\mu^{H,A} + m_\mu^{A,{\psi_{L/R}}}, \\
        H_\lambda^{SM}&=H_\lambda + h_\lambda^A + h_\lambda^{\psi_{L/R}} + h_\lambda^H + h_\lambda^{A,{\psi_{L/R}}} + h_\lambda^{H,A} +h_\lambda^{H,{\psi_{L/R}}}.
    \end{align*}
Given the absence of any triple interaction between the fields, it is easy to check that there does not exist a non-vanishing $\mathbb{x}^{H,A,\psi}$ satisfying
\begin{align*}
    \iota_{\mathbb{x}^{H,A,\psi}}\Omega_{SM}= 
    - \iota_{\mathbb{x}^{H,A}}\varpi_{\psi}
    - \iota_{\mathbb{x}^{H,\psi}}\varpi_{A}
    - \iota_{\mathbb{x}^{A,\psi}}\varpi_{H},
\end{align*}
then the Hamiltonian vector fields associated to constraints are just the sums of the ones found in the previous sections. 
    \begin{align*}
        \mathbb{P}^{SM}&=\mathbb{P}+\mathbb{p}^A+\mathbb{p}^H+\mathbb{p}^{H,A}+\Sigma_{i=L,R}(\mathbb{p}^{\psi_i}+\mathbb{p}^{A,\psi_i})\\
        \mathbb{H}^{SM}&=\mathbb{H}+\mathbb{h}^A+\mathbb{p}^H+\mathbb{h}^{H,A}+\Sigma_{i=L,R}(\mathbb{h}^{\psi_i}+\mathbb{h}^{A,\psi_i})+\mathbb{h}^{H,\psi}\\
        \mathbb{L}^{SM}&=\mathbb{L}+\Sigma_{i=L,R}\mathbb{l}^{\psi_i}\\
        \mathbb{M}^{SM}&=\mathbb{M}^A+\mathbb{m}^{H,A}+\Sigma_{i=L,R}\mathbb{m}^{A,\psi_i}\\
    \end{align*}

One only has to be careful regarding $\mathbb{h}_{\psi_{L/R}}^{H,\psi}$, indeed
    \begin{equation*}
        \gamma\mathbb{h}_{\psi_{L}}^{H,\psi}=-\frac{i}{2}g_Y\lambda e_n<\phi,\psi_R> \qquad \gamma\mathbb{h}_{\psi_{R}}^{H,\psi}=-\frac{i}{2}g_Y\lambda e_n\phi\psi_L
    \end{equation*}
\begin{theorem}\label{thm:first-class-constraints_SM}
    Assume that $g^\partial$ is nondegenerate on $\Sigma$. Then, the zero locus of the functions  $L^{SM}_c$, $P^{SM}_{\xi}$,  $H^{SM}_{\lambda}$ and $M_\mu^{SM}$ defined above is a  coisotropic submanifold with respect to the symplectic structure $\Omega_{SM}$. Their mutual Poisson brackets read
    {\footnotesize
         \begin{align*}
             \{L_c^{SM}, L_c^{SM}\}_{SM} &= - \frac{1}{2} L^{SM}_{[c,c]}  & \{L^{SM}_c, P^{SM}_{\xi}\}_{SM}  &=  L^{SM}_{\mathrm{L}_{\xi}^{\omega_0}c}\\
             \ \{ L^{SM}_c, M^{SM}_\mu \}_{SM} & = 0   & \{L^{SM}_c,  H^{SM}_{\lambda}\}_{SM}  &= - P^{SM}_{X^{(\nu)}} + L^{SM}_{X^{(\nu)}(\omega - \omega_0)_\nu} - H^{SM}_{X^{(n)}} + M^{SM}_{X^{\nu}(A-A_0)_\nu} \\
             \ \{M^{SM}_\mu, P^{SM}_{\xi}\}_{SM}  &=  M^{SM}_{\mathrm{L}_{\xi}^{A_0}\mu  } & \{P^{SM}_{\xi},H^{SM}_{\lambda}\}_{SM}  &=  P^{SM}_{Y^{(\nu)}} -L^{SM}_{ Y^{(\nu)} (\omega - \omega_0)_\nu} +H^{SM}_{ Y^{(n)}} - M^{SM}_{Y^\nu (A-A_0)_\nu} \\
             \ \{M^{SM}_\mu, H^{SM}_\lambda\}_{SM}&= 0  & \{P^{SM}_{\xi}, P^{SM}_{\xi}\}_{SM}  &=  \frac{1}{2}P^{SM}_{[\xi, \xi]}- \frac{1}{2}L^{SM}_{\iota_{\xi}\iota_{\xi}F_{\omega_0}} - \frac{1}{2}M^{SM}_{\iota_{\xi}\iota_{\xi}F_{A_0}} \\
             \ \{M_\mu^{SM},M_\mu^{SM}\}_{SM}&=-\frac{1}{2}M_{[\mu,\mu]}^{SM} & \{H^{SM}_{\lambda},H^{SM}_{\lambda}\}_{SM}  &=0 
         \end{align*}
         }
 with the same notation as in Theorem \ref{thm:first-class-constraints}. 
 
\end{theorem}

\begin{proof}
We make use of Theorem \ref{thm:Poisson_brackets_four} to compute the following brackets. Omitting the vanishing terms, we obtain
    \begin{align*}
        \{L_c^{SM},L_c^{SM}\}_{SM}&=\{L_c+ l_c^{\psi_{L/R}},L_c+l_c^{\psi_{L/R}} \}_{SM}\\
        &=\{L_c+ l_c^{\psi_{L/R}},L_c+l_c^{\psi_{L/R}}\}_{\psi_{L/R}}=-\frac{1}{2}L^{\psi_{L/R}}_{[c,c]}=-\frac{1}{2}L_{[c,c]}^{SM};
    \end{align*}
    \begin{align*}
        \{L_c^{SM},P_\xi^{SM}\}_{SM}&=\{L_c^{A,\psi_{L/R}},P_\xi^{A,\psi_{L/R}}\}_{A,\psi_{L/R}}=L^{A,\psi_{L/R}}_{\mathrm{L}^{\omega_0}c}=L^{SM}_{\mathrm{L}^{\omega_0}c};
    \end{align*}
    \begin{align*}
        \{L_c^{SM},M_\mu^{SM}\}_{SM}&=\{L_c^{A,\psi_{L/R}},M_\mu^{A,\psi_{L/R}}\}_{A,\psi_{L/R}}=0;
    \end{align*}
    \begin{align*}
        \{L_c^{SM},H_\lambda^{SM}\}_{SM}&=\int_\Sigma \lambda e_n \big( \frac{1}{2}[c,e^2<\Pi,d_A\phi>] + \frac{[c,e^3] }{2\cdot3!}(<(\Pi,\Pi)> + 2V_H)  \big)\\
        &\phantom{=}+ \{L_c^{A,\psi_{L/R}},H_\lambda^{A,\psi_{L/R}}\}_{A,\psi_{L/R}} \\
        &=- \int_\Sigma [c,\lambda e_n] \left(  \frac{e^2}{2}<\Pi,d_A\phi> + \frac{e^3}{2\cdot 3!}(2V_H <(\Pi,\Pi)>) \right) \\
        &\phantom{=}+ \{L_c^{A,\psi_{L/R}},H_\lambda^{A,\psi_{L/R}}\}_{A,\psi_{L/R}} \\
        &= - P^{SM}_{X^{(\nu)}} + L^{SM}_{X^{(\nu)}(\omega - \omega_0)_\nu} - H^{SM}_{X^{(n)}} + M^{SM}_{X^{\nu}(A-A_0)_\nu};
    \end{align*}
    \begin{align*}
        \{M_\mu^{SM},P_\xi^{SM}\}_{SM}&=\{M_\mu^{A,\psi_{L/R}},P_\xi^{A,\psi_{L/R}}\}_{A,\psi_{L/R}} \\
        &\phantom{=}- \int_\Sigma <[\mu,p], \mathrm{L}^{\omega_0 +A_0}_\xi\phi> + <p,\mathrm{L}^{\omega_0 +A_0}_\xi[\mu,\phi]> \\
        &=\{M_\mu^{A,\psi_{L/R}},P_\xi^{A,\psi_{L/R}}\}_{A,\psi_{L/R}}\\
        &\phantom{=} - \int_\Sigma <p, [\mu,\mathrm{L}^{\omega_0 +A_0}_\xi\phi] + [\mathrm{L}^{A_0}_\xi \mu,\phi] - [\mu,\mathrm{L}^{\omega_0 +A_0}_\xi\phi]>\\
        &=  M^{A,\psi_{L/R}}_{\mathrm{L}_{\xi}^{\omega_0}\mu  } - \int_\Sigma <p, \mathrm{L}_{\xi}^{\omega_0}\mu ;
    \end{align*}
    \begin{align*}
         \{M_\mu^{SM},H_\lambda^{SM}\}_{SM}&=\int_\Sigma \lambda e_n \big[ \frac{e^2}{2}(<[\mu,\Pi],d_A\phi> + <\Pi, [d_A\mu,\phi]> - <\Pi, d_A[\mu,\phi]>) \\
        &\qquad + g_Y \bar{\psi}_L  \frac{e^3}{3!}( <[\mu,\phi],\psi_L> + <\phi,[\mu,\psi_L]> )  \big]\\
        &\phantom{=}
        +\{M_\mu^{A,\psi_{L/R}},H_\lambda^{A,\psi_{L/R}}\}_{A,\psi_{L/R}}\\
        &=\int_\Sigma \lambda e_n \big[ \frac{e^2}{2}[\mu,<\Pi,d_A\phi>] + g_Y \bar{\psi}_L  \frac{e^3}{3!} [\mu,<\phi,\psi_R>] \big]=0;
    \end{align*}
    \begin{align*}
        \{M_\mu^{SM},M_\mu^{SM}\}_{SM}&= \{M_\mu^{A,\psi_{L/R}},M_\mu^{A,\psi_{L/R}}\}_{A,\psi_{L/R}}+ \int_\Sigma <[\mu,p],[\mu,\phi]> \\
        &=-\frac{1}{2}M_{[\mu,\mu]}^{A,\psi_{L/R}} + \frac{1}{2}\int_\Sigma <p, [\mu,[\mu,\phi]]>\\
        &=-\frac{1}{2}M_{[\mu,\mu]}^{SM};
    \end{align*}
    In the computation of the next couple of brackets we follow verbatim the steps in the analogous brackets in the Yang--Mills-Higgs section.
    \begin{align*}
        \{P_\xi^{SM},P_\xi^{SM}\}_{SM}&=\{P_\xi^{A,\psi_{L/R}},P_\xi^{A,\psi_{L/R}}\}_{A,\psi_{L/R}} \\
        &\phantom{=}+ \frac{1}{2}\int_\Sigma <\mathrm{L}^{\omega_0 +A_0}_\xi p, \mathrm{L}^{\omega_0 +A_0}_\xi\phi> + <p, \mathrm{L}^{\omega_0 +A_0}_\xi \mathrm{L}^{\omega_0 +A_0}_\xi \phi>\\
        &= \frac{1}{2}P^{SM}_{[\xi, \xi]}- \frac{1}{2}L^{SM}_{\iota_{\xi}\iota_{\xi}F_{\omega_0}} - \frac{1}{2}M^{SM}_{\iota_{\xi}\iota_{\xi}F_{A_0}} ;
    \end{align*}
    \begin{align*}
        \{P_\xi^{SM},H_\lambda^{SM}\}_{SM}&= - \int_\Sigma \lambda e_n \left[ \frac{1}{2}\mathrm{L}^{\omega_0 }_\xi(e^2)<\Pi,d_A\phi> + \frac{1}{2\cdot 3!}\mathrm{L}^{\omega_0 }_\xi(e^3)\bar{\psi}_L<\phi,\psi_R> \right]\\
        &\phantom{=}- \int_\Sigma \lambda e_n \left[\frac{1}{2\cdot 3!}\mathrm{L}^{\omega_0 }_\xi(e^3)\left(<(\Pi,\Pi)> + 2V_H\right) + \frac{e^2}{2}(<\mathrm{L}^{\omega_0 +A_0}_\xi \Pi, d_A\phi >\right]\\
        &\phantom{=}- \int_\Sigma \lambda e_n \left[ -\frac{e^2}{2} <\Pi, \left[\mathrm{L}^{A_0}_\xi(A-A_0) +\iota_\xi F_{A_0},\phi\right]> + \frac{e^2}{2}<\Pi,d_A \mathrm{L}^{\omega_0 +A_0}_\xi \phi>  \right] \\
        &\phantom{=}+ \int_\Sigma \lambda e_n  \frac{e^3}{2\cdot 3!}\mathrm{L}^{\omega_0 +A_0}_\xi\left(\bar{\psi}_L<\phi,\psi_R> + <(\Pi,\Pi)> + 2V_H\right)\\
        &\phantom{=}+\{P_\xi^{A,\psi_{L/R}},H_\lambda^{A,\psi_{L/R}}\}_{A,\psi_{L/R}}\\
        &=\int_\Sigma \mathrm{L}^{\omega_0}_\xi(\lambda e_n) \left[ \frac{e^2}{2}<\Pi,d_A\phi> + \frac{e^3}{2\cdot 3!}(\bar{\psi}_L<\phi,\psi_R> + <(\Pi,\Pi)> + 2V_H)\right]\\
        &\phantom{=}+\{P_\xi^{A,\psi_{L/R}},H_\lambda^{A,\psi_{L/R}}\}_{A,\psi_{L/R}}\\
        &= P^{SM}_{Y^{(\nu)}} -L^{SM}_{ Y^{(\nu)} (\omega - \omega_0)_\nu} +H^{SM}_{ Y^{(n)}} - M^{SM}_{Y^\nu (A-A_0)_\nu}.
    \end{align*}
    Lastly, in the computation of $\{H_\lambda^{SM},H_\lambda^{SM}\}$, it is easy to see that all the additional terms are proportional to $\lambda^2=0$, producing the desired result.
\end{proof}
    
\section{BFV formalism}\label{s:BFV_all}
For a precise formulation of the BFV formalism and the construction of a BFV theory out of the classical description of the RPS via symplectic spaces and constraints we refer to \cite[Section 5.1]{CCS2020}. We recall here the definition of BFV theory and its relation to the reduced phase space of a theory. 
\begin{definition}\label{def:BFV_theory}
A BFV theory is a triple $\FF= \left( \mathcal{F}, \mathcal{S}, \varpi \right)$ where $\mathcal{F}$ is a graded manifold (the \emph{space of boundary BFV fields}) endowed with a degree-$0$ exact symplectic form $\varpi= \delta \alpha$ and $\mathcal{S}: \mathcal{F} \rightarrow \mathbb{R}$ is a degree 1 functional
such that $\{\mathcal{S},\mathcal{S}\}=0$, satisfying the Classical Master Equation (CME), where the Poisson brackets $\{\cdot,\cdot\}$ are those inherited from the symplectic form $\varpi$.
\end{definition}

Note that sometimes this is referred as a strict BFV theory. From a BFV theory it is always possible to define a cohomological vector field $\mathcal{Q}$ as the Hamiltonian vector field of $\mathcal{S}$ with respect to $\varpi$. Then the CME is equivalent to $\mathcal{Q}^2=0.$ Under this condition, $\mathcal{Q}$ is a differential graded vector field and it makes sense to consider its cohomology. It turns out that, if we construct a BFV theory in the proper way, the zeroth cohomology of $\mathcal{Q}$ is isomorphic to the space of invariant functions on the coisotropic submanifold generated by the constraints. 

Namely, let us denote by by $X_{c^i}$ the set of constraints, with $c^i$ the corresponding Lagrange multiplier. We can extend the geometric phase space $(F,\varpi_{geom})$ to a graded symplectic manifold $F\times T^*W$ where $W$ is the odd vector space whose coordinates are the $c^i$s. The new symplectic form, or BFV form is then $\varpi_{BFV}= \varpi_{geom} + \int_\Sigma \delta c^\dagger_i\, \delta c^i$
where we denoted by $c^\dagger_i$ the fields in the fiber of $T^*W$.
Then we can define the BFV action, an odd function of degree $1$,
\begin{align*}
    \mathcal{S}= \int_\Sigma c^i\phi_i + \frac12 f_{ij}^k c^\dagger_k c^ic^j + R,
\end{align*}
where $R$ is a function of higher degree in the ghost momenta $c^\dag_i$, chosen such that $\{S,S\}=0$. Such a correction $R$ can always be found \cite{BV1981,Batalin:1983,Stasheff1997}. If we now let $\mathcal{Q}$ be the Hamiltonian vector field of $\mathcal{S}$, then the zeroth cohomology of it is isomorphic to the space of functions on the reduced phase space. 

In the following sections we will use this construction by first applying it to the theories with two gauge/matter fields and using these results to the standard model.

Note that given the results already proven in the previous sections, here we just have to find $R$ in order to get the CME. In particular, in all the cases we will prove that $R=0$. In order to do it we will repeatedly use Theorem \ref{thm:BFV_computation} with the results collected in Appendix \ref{s:appendix_BFV_identities}.

\subsection{Yang--Mills-Spinor}\label{s:BFV_YMS}
\begin{theorem}\label{thm:BFVaction_YMS}
		Let $\mathcal{F}_{YMS}$ be the bundle 
		\begin{equation*}
\mathcal{F}_{YMS} \longrightarrow \Omega_{e_n}^1(\Sigma, \mathcal{V}_{\Sigma})\oplus \mathcal{A}_{\Sigma}^G\oplus S(\Sigma) \oplus \overline{S}(\Sigma),
\end{equation*}
with local trivialisation on an open $\mathcal{U}_{\Sigma}^{YMS}\subset \Omega_{e_n}^1(\Sigma, \mathcal{V}_{\Sigma})\oplus \mathcal{A}_{\Sigma}^G\oplus S(\Sigma) \oplus \overline{S}(\Sigma)$
		\begin{equation*}
			\mathcal{F}_{YMS}\simeq \mathcal{U}_{\Sigma} \times \mathcal{T}_{grav}\times \Omega^{(0,2)}_{\Sigma, red}(\mathfrak{g})\times T^*(\Gamma(\Sigma,\mathfrak{g}[1]))
		\end{equation*}
		where $\mathcal{T}_{grav}$ was defined in \eqref{LoctrivF1} and the fields satisfy \eqref{e:constraintYM} and \eqref{e:omegareprfix2spin}. 
		The symplectic form and the action functional on $\mathcal{F}_{YMS}$ are respectively defined by
		\begin{align}
			\Omega^{BFV}_{YMS} &= \Omega^{BFV}+\varpi_{A}+\varpi^{A}_{ghost}+\varpi_{\psi}, \label{e:BFV_form_YMS}\\
			\mathcal{S}_{YMS} &= \mathcal{S} + p^{A}_\xi + h^{A}_\lambda+ M_{\mu}^A +\mathcal{S}^A_{ghost} + l^{\psi}_{c} + p^{\psi}_{\xi}+ h^{\psi}_{\lambda}+ p^{A, \psi}_\xi +h^{A, \psi}_{\lambda} + m^{A, \psi}_{\mu} \label{e:BFV_action_YMS}
		\end{align}
        where $\varpi^{A}_{ghost}$ and $\mathcal{S}^A_{ghost}$ are defined in Theorem \ref{thm:BFVactionYM}.
		Then the triple $(\mathcal{F}_{YMS}, \Omega^{BFV}_{YMS}, \mathcal{S}_{YMS})$ defines a BFV structure on $\Sigma$.
	\end{theorem}
 
 \begin{proof}
     We apply Theorem \ref{thm:BFV_computation} and consider as the starting BFV theory that of the Yang--Mills field coupled to gravity. Since from Theorem \ref{thm:first-class-constraints_YMS} we already know that the brackets of the interacting theory have the same structure of the one with just gravity coupled to Yang--Mills (given in Theorem \ref{thm:first-class-constraints_YM}), we can use the observation of Remark \ref{rem:BFVandbracketsofconstraints} and we must check just that $\iota_{q_0}\iota_{Q_1}(\varpi + \varpi_A)=0$. We have that
    \begin{align*}
        q_0=\mathbb{l}^{\psi}+\mathbb{p}^{\psi}+\mathbb{p}^{A, \psi}+\mathbb{h}^{\psi}+\mathbb{h}^{A, \psi}+\mathbb{m}^{A, \psi}+ \mathbb{w}
    \end{align*}
    where  $\mathbb{w}$ contains terms in the direction of antifields. In particular, from the expression of $\iota_{Q_1}(\varpi + \varpi_A)$ expressed in Appendix \ref{s:appendix_BFV_identities}, we are interested in the components of $q_0$ in the direction of $e, \omega$ and $\rho$. Using the expression for the Hamiltonian vector fields of Appendix \ref{s:appendix_Hamiltonianvf} and of Section \ref{s:YM+spinor} we get
    \begin{align*}
        (q_0)_{e}&= \mathbb{h}_e^{\psi}= \lambda (\widetilde{\sigma}- \sigma)\\
        (q_0)_{\omega}&=\mathbb{h}_{\omega}^{\psi}+\mathbb{h}_{\omega}^{A, \psi}= - \frac{i}{2}\lambda e_n \overline{\psi}\gamma [A, \psi]+i \frac{\lambda e_n}{4}  (\overline{\psi}\gamma d_\omega \psi - d_\omega\overline{\psi} \gamma \psi) + \mathbb{V}\\
        (q_0)_{\rho}&=\mathbb{h}_{\rho}^{A, \psi}=- \frac{1}{2}\lambda e_n e^2 \overline{\psi}\gamma [\cdot, \psi].
    \end{align*}
    Since all these components are proportional to $\lambda$ and 
    $\iota_{Q_1}(\varpi+\varpi_A)$ is proportional to $\lambda$ as well, we conclude that $\iota_{q_0}\iota_{Q_1}(\varpi + \varpi_A)\sim \lambda^2=0$.
 \end{proof}

\subsection{Yang--Mills-Higgs} 

\begin{theorem}\label{thm:BFVaction_Higgs}
		Let $\mathcal{F}_{YMH}$ be the bundle 
		\begin{equation*}
            \mathcal{F}_{YMH} \longrightarrow \Omega_{e_n}^1(\Sigma, \mathcal{V}_{\Sigma})\oplus \mathcal{A}_\Sigma^G\oplus \Gamma(\Sigma,E_n) ,
        \end{equation*}
with local trivialisation on an open $\mathcal{U}_{\Sigma}^{YMH}\subset \Omega_{e_n}^1(\Sigma, \mathcal{V}_{\Sigma})\oplus \mathcal{A}_\Sigma^G\oplus \Gamma(\Sigma,E_n) $
		\begin{equation*}
			\mathcal{F}_{YMH}\simeq \mathcal{U}_{\Sigma}^{YMH} \times \mathcal{T}_{grav}\times \Omega^1(\Sigma,E_n\otimes \mathcal{V}_{\Sigma})_\mathrm{red} \times \Omega^{(0,2)}_{\Sigma,red}(\mathfrak{g})\times T^*(\Gamma(\Sigma,\mathfrak{g}[1]))
		\end{equation*}
		where $\mathcal{T}_{grav}$ was defined in \eqref{LoctrivF1} and the fields $\Pi\in\Omega^1(\Sigma,E_n\otimes \mathcal{V}_{\Sigma})_\mathrm{red}$ and $B\in\Omega^{(0,2)}_{\Sigma,red}(\mathfrak{g})$ satisfy respectively \eqref{e:constraintscalar} and \eqref{e:constraintYM}. 
		The symplectic form and the action functional on $\mathcal{F}_{YMH}$ are respectively defined by
		\begin{align*}
			\Omega^{BFV}_{YMH}&= \Omega^{BFV}+\varpi_A+\varpi^{A}_{ghost}+\varpi_H, \\
			\mathcal{S}_{YMH} &= \mathcal{S} +M_\mu^{A} + p^{A}_\xi + h_\lambda^A  + \mathcal{S}^A_{ghost} + m_\mu^{H,A} +p^{H}_\xi + p_\xi^{H,A} + h^{H}_\lambda+ h_\lambda^{H,A} .
		\end{align*}
		where $\varpi^{A}_{ghost}$ and $\mathcal{S}^A_{ghost}$ are defined in Theorem \ref{thm:BFVactionYM}.
		Then the triple $(\mathcal{F}_{YMH}, \Omega^{BFV}_{YMH}, \mathcal{S}_{YMH})$ defines a BFV structure on $\Sigma$.
	\end{theorem}

\begin{proof}
    As before we apply Theorem \ref{thm:BFV_computation} and consider as the starting BFV theory that of the Yang--Mills field coupled to gravity. Again, from Theorem \ref{thm:first-class-constraints_Higgs} we know that the brackets of the interacting theory have the same structure of the one with just gravity coupled to Yang--Mills (given in Theorem \ref{thm:first-class-constraints_YM}), hence we can use the observation of Remark \ref{rem:BFVandbracketsofconstraints} and we must check just that $\iota_{q_0}\iota_{Q_1}(\varpi + \varpi_A)$. We have that
    \begin{align*}
        q_0=\mathbb{l}^{H}+\mathbb{p}^{H}+\mathbb{p}^{H,A}+\mathbb{h}^{H}+\mathbb{h}^{H,A}+\mathbb{m}^{H}+ \mathbb{w}
    \end{align*}
    where  $\mathbb{w}$ contains terms in the direction of antifields. In particular, from the expression of $\iota_{Q_1}(\varpi + \varpi_A)$ expressed in Appendix \ref{s:appendix_BFV_identities}, we are interested in the components of $q_0$ in the direction of $e, \omega$ and $\rho$. Using the expression for the Hamiltonian vector fields of Appendix \ref{s:appendix_Hamiltonianvf} and of Section \ref{s:YM+scalar} we get
    \begin{align*}
        (q_0)_{e}&= 0\\
        (q_0)_{\omega}&=\mathbb{h}_{\omega}^{H}+\mathbb{h}_{\omega}^{H,A}= \lambda e_n \left( <\Pi,d\phi> + \frac{e}{4}<(\Pi,\Pi)> + \frac{e}{2}V_H \right) \\ 
        &\quad- \frac{\lambda e_n}{2}<\Pi,[A,\phi]>  -\frac{\lambda}{2}e <\Pi,(\Pi,e_n)>   + \mathbb{V}\\
        (q_0)_{\rho}&=\mathbb{h}_{\rho}^{H,A}= \frac{i}{4}g_H\lambda e_n e^2 \Tr(\Pi \phi^\dagger - \phi \Pi^\dagger).
    \end{align*}
    Since all these components are proportional to $\lambda$ and 
    $\iota_{Q_1}(\varpi+\varpi_A)$ is proportional to $\lambda$ as well, we conclude that $\iota_{q_0}\iota_{Q_1}(\varpi + \varpi_A) \sim \lambda^2=0$.
\end{proof}

\subsection{Yukawa interaction}

\begin{theorem}\label{thm:BFVaction_Yukawa}
		Let $\mathcal{F}_{Y}$ be the bundle 
		\begin{equation*}
\mathcal{F}_{Y} \longrightarrow \Omega_{e_n}^1(\Sigma, \mathcal{V}_{\Sigma})\oplus \mathcal{C^{\infty}}(\Sigma)\oplus S(\Sigma) \oplus \overline{S}(\Sigma),
\end{equation*}
with local trivialisation on an open $\mathcal{U}_{\Sigma} $
		\begin{equation*}
			\mathcal{F}_{Y}\simeq \mathcal{U}_{\Sigma} \times \mathcal{T}_{grav}\times \Omega^{(0,1)}_{\partial,\text{red}}
		\end{equation*}
		where $\mathcal{T}_{grav}$ was defined in \eqref{LoctrivF1} and the fields satisfy \eqref{e:constraintscalar} and \eqref{e:omegareprfix2spin}. 
		The symplectic form and the action functional on $\mathcal{F}_Y$ are respectively defined by
		\begin{align*}
			\Omega^{BFV}_Y &= \Omega^{BFV}+\varpi_{\phi}+\varpi_{\psi}, \\
			\mathcal{S}_{Y} &= \mathcal{S}_{grav} + p^{\phi}_\xi + h^{\phi}_\lambda+ l^{\psi}_{c} + p^{\psi}_{\xi}+ h^{\psi}_{\lambda}+ h^{\phi, \psi}_{\lambda}.
		\end{align*}
		Then the triple $(\mathcal{F}_{Y}, \Omega^{BFV}_{Y}, \mathcal{S}_{Y})$ defines a BFV structure on $\Sigma$.
	\end{theorem}
 
\begin{proof}
    We apply Theorem \ref{thm:BFV_computation} but in this case we consider as the starting theory gravity coupled to the the scalar field. The expression of the brackets of the constraints of the interacting theory (Theorem \ref{thm:first-class-constraints_Yukawa}) and that of gravity coupled with the scalar theory (Section \ref{s:previousResults_scalar}) are the same, hence by Remark \ref{rem:BFVandbracketsofconstraints} we just have to check that $\iota_{q_0}\iota_{Q_1}(\varpi + \varpi_{\phi})=0$. We have
    \begin{align*}
        q_0=\mathbb{l}^{\psi}+\mathbb{l}^{\phi, \psi}+\mathbb{p}^{\psi}+\mathbb{p}^{\phi, \psi}+\mathbb{h}^{\psi}+\mathbb{h}^{\phi, \psi}+ \mathbb{w}
    \end{align*}
    where  $\mathbb{w}$ contains terms in the direction of antifields while the expression of $\iota_{Q_1}(\varpi + \varpi_{\phi})$ is collected in Appendix \ref{s:appendix_BFV_identities}. Hence we are interested in $q_0$ only in the direction of $e$, $\omega$ and $\phi$. From the expression of the Hamiltonian vector fields collected in Appendix \ref{s:appendix_Hamiltonianvf} and in Section \ref{s:scalar+spinor} we get
    \begin{align*}
        (q_0)_{e}&= \mathbb{h}_e^{\psi}= \lambda (\widetilde{\sigma}- \sigma)\\
        (q_0)_{\omega}&=\mathbb{h}_{\omega}^{\psi}= i \frac{\lambda e_n}{4}  (\overline{\psi}\gamma d_\omega \psi - d_\omega\overline{\psi} \gamma \psi) + \mathbb{V}\\
        (q_0)_{\phi}&=0.
    \end{align*}
 These components are all proportional to $\lambda$. Since the components of  $\iota_{Q_1} (\varpi + \varpi_\phi)$ are also all proportional to $\lambda$, we get that $\iota_{q_0}\iota_{Q_1} (\varpi + \varpi_\phi)$ is proportional to $\lambda^2=0$. Hence we conclude that $S_Y$ satisifies the classical master equation.
\end{proof}

\subsection{Standard Model}
Finally in this section we present one of the main results of this paper. Indeed, we characterize, using the BFV formalism, the Reduced Phase Space of the standard model coupled to gravity.  A BFV theory whose zeroth cohomology is isomorphic to it is described in the following theorem. 

\begin{theorem}\label{thm:BFVaction_SM}
		Let $\mathcal{F}_{SM}$ be the bundle 
		\begin{equation*}
\mathcal{F}_{SM} \longrightarrow \Omega_{e_n}^1(\Sigma, \mathcal{V}_{\Sigma})\oplus \mathcal{A}_{\Sigma}^G\oplus S(\Sigma) \oplus \overline{S}(\Sigma)\oplus \Gamma(\Sigma,E_n),
\end{equation*}
with local trivialisation on an open $\mathcal{U}_{\Sigma}^{SM}\subset \Omega_{e_n}^1(\Sigma, \mathcal{V}_{\Sigma})\oplus \mathcal{A}_{\Sigma}^G\oplus S(\Sigma) \oplus \overline{S}(\Sigma)\oplus \Gamma(\Sigma,E_n)$
		\begin{equation*}
			\mathcal{F}_{SM}\simeq \mathcal{U}_{\Sigma} \times \mathcal{T}_{grav}\times \Omega^{(0,2)}_\partial(\mathfrak{g})_\mathrm{red}\times T^*(\Gamma(\Sigma,\mathfrak{g}[1]))\times\Omega^1(\Sigma,E_n\otimes \mathcal{V}_{\Sigma})_\mathrm{red}
		\end{equation*}
		where $\mathcal{T}_{grav}$ was defined in \eqref{LoctrivF1} and the fields satisfy \eqref{e:constraintYM}, \eqref{e:constraintscalar} and \eqref{e:omegareprfix2spin}. 
		The symplectic form and the action functional on $\mathcal{F}_{SM}$ are respectively defined by
		\begin{align*}
			\Omega^{BFV}_{SM} &= \Omega^{BFV}+\varpi_{A}+\varpi^{A}_{ghost}+\varpi_{\psi_L}+\varpi_{\psi_R}+\varpi_H, \\
			\mathcal{S}_{SM} &= \mathcal{S}_{grav} + p^{A}_\xi + h^{A}_\lambda+ M_{\mu}^A +\mathcal{S}^A_{ghost} + l_c^{\psi_{L/R}} + p_\xi^{\psi_{L/R}}+ h_\lambda^{\psi_{L/R}}+ p_\xi^H + h_\lambda^H\\
            &\phantom{=}+ p_\xi^{A,{\psi_{L/R}}} + h_\lambda^{A,{\psi_{L/R}}} + m_\mu^{A,{\psi_{L/R}}} + p_\xi^{H,A}+ h_\lambda^{H,A} + m_\mu^{H,A}  +h_\lambda^{H,{\psi_{L/R}}}
		\end{align*}
        where $\varpi^{A}_{ghost}$ and $\mathcal{S}^A_{ghost}$ are defined in Theorem \ref{thm:BFVactionYM}.
		Then the triple $(\mathcal{F}_{SM}, \Omega^{BFV}_{SM}, \mathcal{S}_{SM})$ defines a BFV structure on $\Sigma$.
	\end{theorem}

 \begin{proof}
     We apply Theorem \ref{thm:BFV_computation} once again but this time we consider as the starting BFV theory the building block given in Section \ref{s:BFV_YMS}, i.e. the BFV theory given by Yang--Mills and spinor fields coupled to gravity. Once again we can apply the observation in Remark \ref{rem:BFVandbracketsofconstraints} since the brackets of the two theories considered are the same (see Theorems \ref{thm:first-class-constraints_YMS} and \ref{thm:first-class-constraints_SM}). Hence we just have to check that $\iota_{q_0}\iota_{Q_1}(\varpi+\varpi_A+\varpi_{\psi_L}+\varpi_{\psi_R})$, where $Q_1$ satisfies $\iota_{Q_1}\Omega^{BFV}_{YMS}= \delta \mathcal{S}^{YMS}_1$ with $\mathcal{S}^{YMS}_1$ being the part linear in the antifields of $\mathcal{S}^{YMS}$ defined in \eqref{e:BFV_action_YMS} and $\Omega^{BFV}_{YMS}$ being defined in \eqref{e:BFV_form_YMS}. On the other hand we have 
    \begin{align*}
        q_0=\mathbb{p}^{H}+\mathbb{h}^{H}+\mathbb{p}^{H,A}+\mathbb{h}^{H,A}+\mathbb{m}^{H, A}+\mathbb{h}^{H, \psi_{L/R}}+ \mathbb{w}
    \end{align*}
    with $\mathbb{w}$ containing terms in the direction of antifields. In particular, from the expression of $\iota_{Q_1}(\varpi+\varpi_A+\varpi_{\psi_L}+\varpi_{\psi_R})$ expressed in Appendix \ref{s:appendix_BFV_identities}, we are interested in the components of $q_0$ in the direction of $e, \omega, \psi_L, \overline{\psi}_L,\psi_R, \overline{\psi}_R$ and $\rho$. Using the expression for the Hamiltonian vector fields of Appendix \ref{s:appendix_Hamiltonianvf} and of Sections \ref{s:YM+scalar} and \ref{s:scalar+spinor} we get
    \begin{align*}
        (q_0)_{e}&=0\\
        (q_0)_{\omega}&=\mathbb{h}_{\omega}^{H}+\mathbb{h}_{\omega}^{A, H}=  \lambda e_n \left( <\Pi,d\phi> + \frac{e}{4}<(\Pi,\Pi)> + \frac{e}{2}V_H \right) -\frac{\lambda}{2}e\Pi(\Pi,e_n) \\
        & \phantom{=\mathbb{h}_{\omega}^{H}+\mathbb{h}_{\omega}^{A, H}=}-\frac{\lambda e_n}{2}<\Pi,[A,\phi]>+\mathbb{V}_{h^H}\\
        \gamma(q_0)_{\psi_L}&=\gamma\mathbb{h}^{H, \psi_L}_{\psi_L}=-\frac{i}{2}g_Y \lambda e_n\phi \psi_L\\
        (q_0)_{\overline{\psi}_L}\gamma&=\mathbb{h}^{H, \psi_L}_{\overline{\psi}_L}\gamma=\frac{i}{2}g_Y \lambda e_n\overline{\psi}_L\phi\\
        \gamma(q_0)_{\psi_R}&=\gamma\mathbb{h}^{H, \psi_R}_{\psi_R}=-\frac{i}{2}g_Y \lambda e_n\phi \psi_R\\
        (q_0)_{\overline{\psi}_R}\gamma&=\mathbb{h}^{H, \psi_R}_{\overline{\psi}_R}\gamma=\frac{i}{2}g_Y \lambda e_n\overline{\psi}_R\phi\\
        (q_0)_{\rho}&=\mathbb{h}^{H, A}_{\rho}=\frac{i g_H}{4}\lambda e_n e^2 \Tr(\Pi \phi^\dagger - \phi \Pi^\dagger)
    \end{align*}
    Since all these components are proportional to $\lambda$ and 
    $\iota_{Q_1}(\varpi+\varpi_A+\varpi_{\psi_L}+\varpi_{\psi_R})$ is proportional to $\lambda$ as well, we conclude that $\iota_{q_0}\iota_{Q_1}\varpi_f \sim \lambda^2=0$.
 \end{proof}

\section{Light-like boundary}\label{s:degenerate_section}
In the preceding sections, we have consistently assumed that the pulled-back boundary metric $g^\partial=\iota^*g$ was non-degenerate. In the degenerate case some adaptations need to be taken into account.
\subsection{Gravity, light-like boundary}\label{s:degenerate_gravity}

 The analysis of the symplectic reduction to obtain the geometric phase space is independent of the nature of the boundary metric. Specifically, in both the degenerate and non-degenerate cases, the symplectic form remains unchanged.
However, when dealing with a light-like boundary, the geometric phase space has to be described differently, since it is not always possible to find a couple $(e,\omega)$ satisfying \eqref{e:structural_constraint_grav}.
Hence, it must be weakened and this leads to a new constraint. The details of this construction for gravity have been worked out in \cite{CCT2020} and the most important ones are recalled in Appendix \ref{a:def_maps_and_spaces} together with the definition of some of the spaces used. Here we just recall the results.

The Geometric phase space of gravity on a light-like boundary is a bundle 
\begin{align*}
    F_d^{\partial} \rightarrow \Omega^1_{e_n}(\Sigma,\mathcal{V}_{\Sigma}) 
\end{align*}
with fiber $\mathcal{A}_{red}(\Sigma)$ where the fields $\omega \in \mathcal{A}_{red}(\Sigma)$ satisfy the structural constraint
\begin{align}\label{e:structural_constraint_grav_deg}
    e_n (d_\omega e- p_{\mathcal{T}}(d_{\omega} e)) = e \sigma_d
\end{align}
for some $\sigma_d\in \Omega^{1}(\Sigma, \mathcal{V}_{\Sigma})$ and 
\begin{align}\label{e:structural_constraint2_grav_deg}
    p_{\mathcal{K}} \omega = 0
\end{align}

where the space $\mathcal{K}$ is defined in Definition \ref{d:spaces_and_maps} in Appendix \ref{a:def_maps_and_spaces}.
The corresponding symplectic form on  $F_d^{\partial} $ is given by $\varpi$ as in the non-degenerate case:
\begin{align*}
    \varpi = \int_{\Sigma} e \delta e \delta \omega .
\end{align*}

On this geometric phase space the constraints are the same of the non-degenerate theory with an additional one.
\begin{definition}\label{constraints_gravity-deg}
    The set of integral functionals defining the constraints of the theory are $L_c$, $H_\lambda$ and $P_\xi$ as defined in Section \ref{s:previousResults_gravity}, with the additional constraint
       \begin{align}\label{e:constraint_R}
     R_\tau&=\int_{\Sigma}\tau d_\omega e,
    \end{align}
with $\tau\in\mathcal S$ where the space $\mathcal{S}$ is defined in Definition \ref{d:spaces_and_maps} in Appendix \ref{a:def_maps_and_spaces}.
\end{definition}

The presence of the additional constraint spoils the coisotropicity of the zero set, as proved by the following.

\begin{theorem}[\cite{CCT2020}]
    Let $g^\partial$ be degenerate on $\Sigma$. Then, the set of constraints $L_c$, $H_\lambda$, $P_\xi$ and $R_\tau$ given in Definition \ref{constraints_gravity-deg} does not form a first class system.
\end{theorem}
\begin{remark}
    The presence of the constraint $R_\tau$ is causing the constraint set to be second class. We interpret this fact as the essential peculiarity of the degenerate theory.
    In particular the number of local degrees of freedom is only \emph{one} instead of \emph{two} as in the case of the non-degenerate theory.
\end{remark}
As we have done for a space or time-like boundary, in the following sections, we recall the description of the reduced phase space for gravity coupled with some gauge and matter fields.

\subsection{Scalar field, light-like boundary}\label{s:degenerate_scalar}

In this section, we will briefly summarize the findings shown in \cite{CFHT23} about the degenerate structure of gravity coupled  with a scalar field. We use the same notation as in the non-degenerate case and we refer to Section \ref{s:previousResults_scalar} for the definition of the quantities mentioned here.

The geometric phase space is the bundle
\begin{align*}
    {F}^{\partial}_{d,\phi} \rightarrow \Omega_{e_n}^1(\Sigma, \mathcal{V}_{\Sigma})\oplus \mathcal{C^{\infty}}(\Sigma)
\end{align*}
with fiber $\mathcal{A}_{red}(\Sigma)\oplus\Omega^{(0,1)}_{\partial,\text{red}}$ such that the structural constraints \eqref{e:structural_constraint_grav_deg}, \eqref{e:structural_constraint2_grav_deg} and \eqref{e:constraintscalar} are satisfied.
  ${F}^{\partial}_{d,\phi}$ is a symplectic space with symplectic form
  $\Omega_{\phi}= \varpi+\varpi_{\phi}$.
  The set of constraints in this space is given by 
  \begin{align*}
      L_c, \qquad P^{\phi}_{\xi}=P_{\xi}+p^{\phi}_\xi, \qquad H^{\phi}_{\lambda}=H_{\lambda}+h^{\phi}_\lambda \quad\text{and}\quad R_{\tau}
  \end{align*}
 Then, it has been proven that these functions are not of first class and hence they do not define a coisotropic submanifold.

\subsection{Yang--Mills, light-like boundary}\label{s:degenerate_YM}

Let us now consider gravity coupled with a Yang--Mills field. We refer again to \cite{CFHT23} for more details.
The geometric phase space is the bundle 
\begin{align*}
    F^{\partial}_{d,A} \rightarrow \Omega_{e_n}^1(\Sigma, \mathcal{V}_{\Sigma})\oplus \mathcal{A}^{G}_{\Sigma}
\end{align*}
with fiber $\mathcal{A}_{red}(\Sigma)\oplus \Omega^{(0,2)}_{\Sigma, red}(\mathfrak{g})$ such that, the structural constraints \eqref{e:structural_constraint_grav_deg}, \eqref{e:structural_constraint2_grav_deg}  and \eqref{e:constraintYM}
  are satisfied.
  The symplectic form on ${F}^{\partial}_{d,A}$ reads $\Omega_{A}= \varpi + \varpi_{A}$.

  The constraints of the theory are $L^A_c$, $H^A_\lambda$, $P^A_\xi$ and $M_\mu^A$ as defined in Section \ref{s:YM}, with the additional $R_\tau$ defined in \ref{s:previousResults_gravity}.

The algebra of constraints has been computed, leading to the result that this set of constraints do not form a first class system.

\subsection{Spinor, light-like boundary}\label{s:degenerate_spinor}
We now consider the theory of gravity coupled with a Dirac spinor. We refer once more to \cite{CFHT23} for more details.
The geometric phase space is the bundle 
\begin{align*}
    {F}^{\partial}_{d,\psi} \rightarrow \Omega_{e_n}^1(\Sigma, \mathcal{V}_{\Sigma})\times S(\Sigma)\times \overline{S}(\Sigma)
\end{align*}
with fiber $\mathcal{A}_{red}(\Sigma)$ such that \eqref{e:structural_constraint2_grav_deg} and
\begin{align}\label{e:structural_constraint_spinor_deg}
         e_n\left[\alpha(e,\omega,\psi,\bar\psi)-p_\mathcal{T}(\alpha(e,\omega,\psi,\bar\psi))\right]= e \widetilde{\sigma}_d
    \end{align}
are satisfied for some $\widetilde{\sigma}_d \in \Omega^{1}(\Sigma, \mathcal{V}_{\Sigma})$, where $\alpha(e,\omega,\psi,\bar\psi)=d_\omega e+\frac{i}{4}(\bar\psi\gamma[e^2,\psi]-[e^2,\bar\psi]\gamma\psi)$.
 The symplectic form on this space is given by $\Omega_{\psi}= \varpi + \varpi_{\psi}$.

The constraints on this space are given are $L^\psi_c$, $H^\psi_\lambda$ and $P^\psi_\xi$ as defined in Section \ref{s:spinor}, with the additional constraint     \begin{align*}
        R^\psi_\tau = R_\tau+r^\psi_\tau
    \end{align*}
    where 
    \begin{align}\label{e:constraint_R_psi}
         r^\psi_\tau&=\int_{\Sigma}\frac{i}{4}\tau(\bar\psi\gamma[e^2,\psi]-[e^2,\bar\psi]\gamma\psi)\\
         &=\int_{\Sigma}+\frac{i}{4}e^2(\bar\psi\gamma[\tau,\psi]+[\tau,\bar\psi]\gamma\psi) \nonumber
    \end{align}
with $\tau\in\mathcal S$.
As for the other theories, also this set of constraints does not form a first class system.

\subsection{Yang--Mills--Spinor, light-like boundary}\label{s:degenerate_YMS}
In this section, our focus will be on studying the scenario of gravity coupled with a Yang--Mills field and a Dirac spinor. 
Since the KT construction does not depend on the degeneracy of the boundary metric, the geometric phase space, as a quotient is the same as in \ref{s:YM+spinor}. However the presentation with the structural constraints fixing the representatives is different in the two cases, as in the previous cases. 
In particular, we can present the geometric phase space as the bundle

\begin{align*}
    {F}^{\partial}_{d, YMS} \rightarrow \Omega_{e_n}^1(\Sigma, \mathcal{V}_{\Sigma})\oplus \mathcal{A}^{G}_{\Sigma} \times S(\Sigma)\times \overline{S}(\Sigma)
\end{align*}
with fiber $\mathcal{A}_{red}(\Sigma) \oplus \Omega^{(0,2)}_{\Sigma, red}(\mathfrak{g})$ such that  \eqref{e:constraintYM}, \eqref{e:structural_constraint_spinor_deg} and  \eqref{e:structural_constraint2_grav_deg} are satisfied. The symplectic form on this space reads 
\begin{align*}
    \Omega_{YMS}= \varpi + \varpi_A + \varpi_\psi 
\end{align*}
exactly as in the non-degenerate case.
On this geometric phase space we can then define the following constraints:
\begin{align*}
    L^{A, \psi}_c &= L_c + l^{\psi}_{c}; \\
    P^{A, \psi}_{\xi} &= P_{\xi}+ p^{A}_{\xi} + p^{\psi}_{\xi}+p^{A, \psi}_{\xi}; \\
    H^{A, \psi}_{\lambda} &= H_{\lambda}+ h^{A}_{\lambda} + h^{\psi}_{\lambda}+ h^{A, \psi}_{\lambda};\\
    M^{A, \psi}_\mu &= M^{A}_\mu + m^{A, \psi}_\mu\\
    R^{A,\psi}_\tau &= R_\tau+r^{\psi}_\tau.
\end{align*}
\begin{remark}\label{rem:R_for_YMS}
    Note that the constraint $R^{A,\psi}_\tau$ takes the very same form of the one of the free degenerate spinor $R^\psi_\tau$ defined in Section \ref{s:degenerate_spinor}.
\end{remark}
\begin{theorem}\label{thm:no-first-class-constraints_YM+spinor_deg}
    Let $g^\partial$ be degenerate on $\Sigma$. Then, the set of integral functionals $L_c^{A,\psi}$, $H^{A,\psi}_\lambda$, $P^{A,\psi}_\xi$, $M_\mu^{A,\psi}$ and $R^{A,\psi}_\tau$  does not form a first class system.
\end{theorem}
\begin{proof}
    In order to prove the result it is just sufficient to show that one bracket is not proportional to any constraint (and not zero). As for the previous cases (see \cite{CCT2020,CFHT23}) we can show that 
    \begin{align*}
        \{R^{A,\psi}_\tau, R^{A,\psi}_\tau\}_{YMS} \neq 0.
    \end{align*}
    This is a consequence of the fact that, using Theorem \ref{thm:Poisson_Brackets_couples} and Remark \ref{rem:R_for_YMS}, it is possible to show that
    \begin{align*}
        \{R^{A,\psi}_\tau, R^{A,\psi}_\tau\}_{YMS} = \{R^{\psi}_\tau, R^{\psi}_\tau\}_{\psi}
    \end{align*}
    and the last term is not proportional to any constraint (nor zero) by the results in \cite{CFHT23} recalled in Section \ref{s:degenerate_spinor}.

\end{proof}

\subsection{Yang--Mills--Higgs, light-like boundary}\label{s:degenerate_YMH}

In this section, we will examine the case of gravity coupled with a Yang--Mills field and a Higgs fields and their interactions. We proceed as in the previous section, using the results of the corresponding non-degenerate Section \ref{s:YM+scalar}.

The geometric phase space  is defined as the bundle
    \begin{equation*}
        {F}^\partial_{d, YMH}\rightarrow\Omega_{e_n}^1(\Sigma, \mathcal{V}_{\Sigma})\oplus \mathcal{A}_\Sigma^{SU(n)}\oplus \Gamma(\Sigma,E_n\vert_\Sigma)\times\Gamma(\Sigma,E_{\bar{n}}\vert_\Sigma),
    \end{equation*}
with fiber 
    \begin{equation*}
        \mathcal{A}_{red}(\Sigma)\oplus \Omega^{(0,2)}_{\Sigma, red}(\mathfrak{g})\oplus\Omega^{0,1}_\partial(E_{n}\vert_\Sigma)\times\Omega^{0,1}_\partial(E_{\bar{n}}\vert_\Sigma ).
    \end{equation*}
    where $\omega,B,\Pi,\Pi^\dagger$ satisfy
    \eqref{e:structural_constraint_grav_deg}, \eqref{e:structural_constraint2_grav_deg} \eqref{e:constraintYM} and \eqref{e:structural_Higgs}.
${F}^\partial_{d, YMH}$ is symplectic with symplectic form
    \begin{equation*}
        \Omega_{YMH}=\varpi + \varpi_{A}+\varpi_H.
    \end{equation*}

We now consider  the constraints. As in the non-degenerate case we have 
    \begin{align*}
    L_c^{H,A}&=L_c; &  P_{\xi}^{H,A}&=P_\xi + p_\xi^A + p_\xi^H + p_\xi^{H,A};\\
    M_\mu^{H,A}&=M_\mu^A + m_\mu^{H,A}; &   H_{\lambda}^{H,A}&=H_\lambda + h_\lambda^A + h_\lambda^{H} + h_\lambda^{H,A}   \\
    R_\tau^{H,A}&=R_\tau.
    \end{align*}

\begin{remark}
    Note that, as it was for the case of the Yang--Mills coupling of Section \ref{s:YM}, the constraint $R^{H,A}_\tau$ takes the very same form of the purely gravitational one $R_\tau$.
\end{remark}
\begin{theorem}\label{thm:no-first-class-constraints_Higgs}
    Let $g^\partial$ be degenerate on $\Sigma$. Then, the set of integral functionals $L^{H,A}_c$, $P^{H,A}_{\xi}$,  $H^{H,A}_{\lambda}$, $M_\mu^{H,A}$ and $R^{H,A}_\tau$ does not form a first class system.
\end{theorem}
\begin{proof}
The proof is completely analogue to the one of Theorem \ref{thm:no-first-class-constraints_YM+spinor_deg} where in this case we use that 
\begin{align*}
        \{R^{H,A}_\tau, R^{H,A}_\tau\}_{YMH} = \{R_\tau, R_\tau\}
    \end{align*}
    and the last term is not proportional to any constraint (nor zero) by the results in \cite{CCT2020} recalled in Section \ref{s:degenerate_gravity}.

\end{proof}

\subsection{Yukawa interaction, light-like boundary}\label{s:degenerate_yuk}
 As a third interaction in this section, we will investigate the scenario of gravity coupled with a scalar and a spinor field and the Yukawa interaction. We proceed as in the previous sections, using the results of the corresponding non-degenerate Section \ref{s:scalar+spinor}.
The geometric phase space is once again the bundle
\begin{align*}
    {F}^{\partial}_{d, Y} \rightarrow \Omega_{e_n}^1(\Sigma, \mathcal{V}_{\Sigma})\oplus \mathcal{C^{\infty}}(\Sigma)\times S(\Sigma)\times \overline{S}(\Sigma)
\end{align*}
with fiber $\mathcal{A}_{red}(\Sigma)\oplus \Omega^{(0,1)}_{\Sigma,red}$ such that  \eqref{e:constraintscalar} and  \eqref{e:structural_constraint2_grav_deg} and \eqref{e:structural_constraint_spinor_deg} are satisfied.

The corresponding symplectic form is again just the sum of the symplectic form of the building blocks:
\begin{align*}
\Omega_{Y} = \varpi + \varpi_{\psi}+ \varpi_{\phi}.
\end{align*}
The constraints for the degenerate Yukawa theory are:
\begin{align*}
    L^{\phi, \psi}_c&= L_c + l^{\psi}_{c}; &
    P^{\phi, \psi}_{\xi}&= P_{\xi}+ p^{\phi}_{\xi} + p^{\psi}_{\xi}; \\
    H^{\phi, \psi}_{\lambda}&= H_{\lambda}+ h^{\phi}_{\lambda} + h^{\psi}_{\lambda}+ h^{\phi, \psi}_{\lambda} &
    R^{\phi,\psi}_\tau&=R_\tau+r^{\psi}_\tau
\end{align*}
\begin{remark}
    Note again that the constraint $R^{\phi,\psi}_\tau$ is exactly  $R^\psi_\tau$.
\end{remark}

We are now able to give the Poisson brackets of the constraints.

\begin{theorem}\label{thm:no-first-class-constraints_Yukawa}
    Let $g^\partial$ be degenerate on $\Sigma$. Then, the constraints $L^{\phi,\psi}_c$, $P^{\phi,\psi}_{\xi}$,  $H^{\phi,\psi}_{\lambda}$ and $R^{\phi,\psi}_\tau$ do not form a first class system.
\end{theorem}
\begin{proof}
    The proof of Theorem \ref{thm:no-first-class-constraints_YM+spinor_deg} holds verbatim also in this case (with the appropriate substitution of indices).

\end{proof}

\subsection{Standard model, light-like boundary}\label{s:degenerate_SM}
Finally in this section we describe the reduced phase space of the standard model on a light-like boundary. We use here the same notation as in Section \ref{s:standard_model} and we refer to it for the definition of the quantities used here. As before we consider a theory of gravity coupled to a Yang--Mills, a Higgs scalar and a spinor fields. Then the usual structure of the standard model can be recovered by choosing the appropriate gauge group and its representation on the various fields. 
 
On a light-like boundary $\Sigma$, the geometric phase space is given by 
    \begin{equation*}
        F^\partial_{d, SM}\rightarrow \Omega^1_{e_n}(\Sigma,\mathcal{V}_{\Sigma}) \oplus \mathcal{A}_\Sigma^G \oplus \Gamma(\Sigma, (S_R\otimes E_n)) \oplus \Gamma(\Sigma, (\bar{S}_R\otimes E_n)) \oplus \Gamma(\Sigma,{S}_L) \oplus\Gamma(\Sigma,\bar{S}_L) \oplus \Gamma(\Sigma,E_n),
    \end{equation*}
with fiber given by $\mathcal{A}_\Sigma \oplus \Omega^{(0,2)}_{\partial,\mathrm{red}}(\mathfrak{g}) \oplus \Omega^{(0,1)}_{\partial,\mathrm{red}}(E_n) $ satisfying \eqref{e:structural_constraint_spinor_deg}, \eqref{e:structural_constraint2_grav_deg},\eqref{e:constraintYM} and \eqref{e:constraintscalar}.

The symplectic structure on $F^\partial_{d, SM}$ is 
    \begin{equation*}
        \Omega_{SM}=\varpi+\varpi_A+\varpi_{\psi_L}+\varpi_{\psi_R}+\varpi_H.
    \end{equation*}
    On this space the KT construction leads to the following constraints:
    \begin{align*}
        L_c^{SM}&=L_c+l_c^{\psi_{L/R}}, \\
        P_\xi^{SM}&=P_\xi + p_\xi^A + p_\xi^{\psi_{L/R}} + p_\xi^H + p_\xi^{A,{\psi_{L/R}}} + p_\xi^{H,A}, \\
        M_\mu^{SM}&= M_\mu^A + m_\mu^{H,A} + m_\mu^{A,{\psi_{L/R}}}, \\
        H_\lambda^{SM}&=H_\lambda + h_\lambda^A + h_\lambda^{\psi_{L/R}} + h_\lambda^H + h_\lambda^{A,{\psi_{L/R}}} + h_\lambda^{H,A} +h_\lambda^{H,{\psi_{L/R}}}\\
        R_\tau^{SM}&=R_\tau +  r_\tau^{\psi_{L/R}} 
    \end{align*}
\begin{theorem}\label{thm:no-first-class-constraints_SM}
    Let $g^\partial$ be degenerate on $\Sigma$. Then, the constraints $L^{SM}_c$, $P^{SM}_{\xi}$,  $H^{SM}_{\lambda}$ and $R^{SM}_\tau$ do not form a first class system.
\end{theorem}
\begin{proof}
    Once again, since $R_\tau^{SM}= R_\tau^{\psi}$ we can resort to the proof of Theorem \ref{thm:no-first-class-constraints_YM+spinor_deg} and conclude as well that
    \begin{align*}
        \{R^{SM}_\tau, R^{SM}_\tau\}_{SM} \neq 0
    \end{align*}
    and it is not proportional to any other constraint.
\end{proof}
Since these constraints are not defining a coisotropic submanifold of the geometric phase space, in the case of a light-like boundary we cannot resort to the BFV formalism to describe this quotient cohomologically.

\appendix

\section{Some useful identities for BFV structures}
\label{s:appendix_BFV_identities}

From Theorem \ref{thm:BFV_computation} we deduce that one of the quantities that will need to be computed is $\iota_{q_0}\iota_{Q_1}\varpi_f$ for the various theories that we consider here. Let us now recap the expressions of $\iota_{Q_1}\varpi_f$ for gravity coupled to a scalar field, a Yang--Mills field and a spinor field respectively. From the definition of $Q_1$, $\iota_{Q_1}\varpi_{BFV}= \delta S_1$, we deduce that it only depends on the expression of $S_1$ which is the same for the first and the third theory. In particular the quantity  $\iota_{Q_1}\varpi_f$ was already computed in \cite[Proof of Theorem 30]{CCS2020} from which we get:
    \begin{align*}
       \iota_{Q_1}\varpi= \int_{\Sigma}\iota_{Q_1}(e \delta e \delta \omega) =\int_{\Sigma} &  - [c, \lambda e_n ]^{(b)}\delta e_b^{(a)}(\xi_a^{\dag}- (\omega - \omega_0)_{\nu} c^\dag) - [c, \lambda e_n ]^{(a)}\delta (\omega - \omega_0)_{\nu} c^\dag \\
 & - [c, \lambda e_n ]^{(b)}\delta e_b^{(n)}\lambda^\dag +\mathrm{L}_{\xi}^{\omega_0} (\lambda e_n)^{(b)}\delta e_b^{(a)}(\xi_a^{\dag}- (\omega - \omega_0)_{\nu} c^\dag) \\
 &+ \mathrm{L}_{\xi}^{\omega_0} (\lambda e_n)^{(a)}\delta (\omega - \omega_0)_{\nu} c^\dag +\mathrm{L}_{\xi}^{\omega_0} (\lambda e_n)^{(b)}\delta e_b^{(n)}\lambda^\dag.
\end{align*}
Furthermore we get
\begin{align*}
    \iota_{Q_1}\varpi_{\phi}=\int_{\Sigma}\iota_{Q_1}\frac{1}{3!}\delta(e^3 \Pi)\delta \phi)= \int_{\Sigma} \frac{1}{2}e^2 \Pi \left([c, \lambda e_n ]^{(a)}c_a^\dag- \mathrm{L}_{\xi}^{\omega_0} (\lambda e_n)^{(a)}c_a^\dag\right) \delta \phi,
\end{align*}
\begin{align*}
    \iota_{Q_1} \varpi_{\psi}= \int_{\Sigma} i\frac{e^2}{4}\left( \overline{\psi}\gamma \delta \psi - \delta \overline{\psi}	\gamma \psi \right)\left([c, \lambda e_n ]^{(a)}c_a^\dag- \mathrm{L}_{\xi}^{\omega_0} (\lambda e_n)^{(a)}c_a^\dag\right)
\end{align*}
and
\begin{align*}
    \iota_{Q_1} \varpi_{A}= \int_{\Sigma}\Tr\left[\delta \rho \left([c, \lambda e_n ]^{(a)}\mu_a^\dag- \mathrm{L}_{\xi}^{\omega_0} (\lambda e_n)^{(a)}\mu_a^\dag\right)\right].
\end{align*}
Since $[c, \lambda e_n ]^{(a)}= \lambda[c,  e_n ]^{(a)}$ and $\mathrm{L}_{\xi}^{\omega_0} (\lambda e_n)^{(a)}= \mathrm{L}_{\xi}^{\omega_0} (\lambda) e_n^{(a)}-\lambda\mathrm{L}_{\xi}^{\omega_0} ( e_n)^{(a)}= -\lambda\mathrm{L}_{\xi}^{\omega_0} ( e_n)^{(a)}$, it is straightforward to check that $\iota_{Q_1}\varpi$, $\iota_{Q_1}\varpi_{\phi}$, $\iota_{Q_1}\varpi_{A}$ and $\iota_{Q_1}\varpi_{\psi}$ are all proportional to $\lambda$.

\section{Hamiltonian vector fields}
\label{s:appendix_Hamiltonianvf}
In this section we list the Hamiltonian vector fields of the constraints computed in \cite{CCS2020,CCF2022}. We list 
the expressions of $\mathbb{L}$, $\mathbb{l}^{ \psi}$, $\mathbb{M}^{A}$, $\mathbb{P}$, $\mathbb{p}^{\phi}$, $\mathbb{p}^{A}$, $\mathbb{p}^{\psi}$, $\mathbb{H}$, $\mathbb{h}^{\phi}$, $\mathbb{h}^{A}$ and $\mathbb{h}^{\psi}$
satisfying
\begin{align*}
    \iota_{\mathbb{L}} \varpi &= \delta L_c, & 
    \iota_{\mathbb{L}}\varpi_{\phi}+ \iota_{\mathbb{l}^{\phi}} (\varpi+\varpi_{\phi}) &= \delta l^{\phi}_c,\\
    \iota_{\mathbb{L}}\varpi_{A}+ \iota_{\mathbb{l}^{A}} (\varpi+\varpi_{A}) &= \delta l^{A}_c, &
    \iota_{\mathbb{L}}\varpi_{\psi}+ \iota_{\mathbb{l}^{\psi}} (\varpi+\varpi_{\psi}) &= \delta l^{\psi}_c,\\
    \iota_{\mathbb{M}^{A}} (\varpi+\varpi_{A}) &= \delta M^{A}_{\mu}\\
    \iota_{\mathbb{P}} \varpi &= \delta P_{\xi}, & 
    \iota_{\mathbb{P}}\varpi_{\phi}+ \iota_{\mathbb{p}^{\phi}} (\varpi+\varpi_{\phi}) &= \delta p^{\phi}_{\xi},\\
    \iota_{\mathbb{P}}\varpi_{A}+ \iota_{\mathbb{p}^{A}} (\varpi+\varpi_{A}) &= \delta p^{A}_{\xi}, &
    \iota_{\mathbb{P}}\varpi_{\psi}+ \iota_{\mathbb{p}^{\psi}} (\varpi+\varpi_{\psi}) &= \delta p^{\psi}_{\xi},\\
    \iota_{\mathbb{H}} \varpi &= \delta H_{\lambda}, & 
    \iota_{\mathbb{H}}\varpi_{\phi}+ \iota_{\mathbb{h}^{\phi}} (\varpi+\varpi_{\phi}) &= \delta h^{\phi}_{\lambda},\\
    \iota_{\mathbb{H}}\varpi_{A}+ \iota_{\mathbb{h}^{A}} (\varpi+\varpi_{A}) &= \delta h^{A}_{\lambda}, &
    \iota_{\mathbb{H}}\varpi_{\psi}+ \iota_{\mathbb{h}^{\psi}} (\varpi+\varpi_{\psi}) &= \delta h^{\psi}_{\lambda}.
\end{align*}
Note that we have $l^{\phi}_c=l^{A}_c=0$. We express the vector fields in components, i.e. $\mathbb{X}=\mathbb{X_{\phi}}\frac{\delta}{\delta \phi}$. We start with the vaccum ones.
\begin{align*}
    \mathbb{L}_e &= [c,e] & \mathbb{L}_{\omega}&=d_{\omega} c + \mathbb{V}_L \\
    \mathbb{P}_e &= - \mathrm{L}_{\xi}^{\omega_0} e & \mathbb{P}_{\omega}&=  - \mathrm{L}_{\xi}^{\omega_0} (\omega-\omega_0) - \iota_ {\xi}F_{\omega_0} + \mathbb{V}_P\\
    \mathbb{H}_e &=d_{\omega}(\lambda e_n) + \lambda \sigma & e \mathbb{H}_{\omega} &=  \lambda e_n F_{\omega}+\frac{1}{2}\Lambda  \lambda e_n e^2
\end{align*}
where $\mathbb{V}_L, \mathbb{V}_P, \mathbb{H}_{\omega}\in \mathrm{ker}(W_1^{\partial,(1,2)})$ are such that the vector fields 
$\mathbb{L}$, $\mathbb{P}$ and $\mathbb{H}$ are tangent to the structural constraint \eqref{e:structural_constraint_grav} (see \cite[Remark 26]{CCS2020}).

Let us now list the ones related to the scalar field theory:
    \begin{align*}
        \mathbb{l}^{\phi}_e &= 0  & \mathbb{l}^{\phi}_{\omega} &= \mathbb{V}_{l^{\phi}} \\
        \mathbb{l}^{\phi}_{\Pi}& = [c, \Pi]+\mathbb{W}_{l^{\phi}} &\mathbb{l}^{\phi}_\phi&=0\\
		\mathbb{p}^{\phi}_e &= 0 & \mathbb{p}^{\phi}_{\omega}&=\mathbb{V}_{p^{\phi}}\\
        \mathbb{p}^{\phi}_{\Pi}&= - \mathrm{L}_{\xi}^{\omega_0} \Pi +\mathbb{W}_{p^{\phi}} & \mathbb{p}^{\phi}_{\phi}&= - \mathrm{L}_{\xi} \phi\\
        \mathbb{h}^{\phi}_e &= 0 & \mathbb{h}^{\phi}_{\omega}&=\lambda  e_n\left(e\Pi d\phi + \frac{1}{4}e^2(\Pi,\Pi) \right) - \frac{\lambda }{2}e^2\Pi(\Pi,e_n)+\mathbb{V}_{h^{\phi}}\\
        \frac{1}{3!}e^3\mathbb{h}^{\phi}_{\Pi}&= \frac{1}{2}  \lambda e_n e d_{\omega}(e\Pi)  & \mathbb{h}^{\phi}_{\phi}&=  (\lambda e_n , \Pi).
		\end{align*}
        where $\mathbb{V}_{l^{\phi}}$, $\mathbb{V}_{p^{\phi}}$, $\mathbb{V}_{h^{\phi}}$ $\mathbb{W}_{l^{\phi}}$, $\mathbb{W}_{p^{\phi}}$, and 
		$\mathbb{h}^{\phi}_{\Pi}$ are such that the vector fields 
        $\mathbb{L}+\mathbb{l}^{\phi}$, $\mathbb{P}+\mathbb{p}^{\phi}$ and $\mathbb{H}+\mathbb{h}^{\phi}$ are tangent to the structural constraints \eqref{e:structural_constraint_grav} and \eqref{e:constraintscalar}.

Let us now list the ones related to the Yang--Mills theory:
    \begin{align*}
        \mathbb{l}^{A}_e &= 0  & \mathbb{l}^{A}_{\omega} &= \mathbb{V}_{l^{A}} \\
        \mathbb{l}^{A}_{B}& = [c, B]+\mathbb{W}_{l^{A}} &\mathbb{l}^{A}_A&=0\\
        \mathbb{M}^{A}_e &= 0  & \mathbb{M}^{A}_{\omega} &= \mathbb{V}_{M^{A}} \\
        e^2 \mathbb{M}^{A}_{B}& = [\mu, e^2 B] &\mathbb{M}^{A}_A&=d_A \mu\\
		\mathbb{p}^{A}_e &= 0 & \mathbb{p}^{A}_{\omega}&=\mathbb{V}_{p^{A}}\\
        \mathbb{p}^{A}_{B}&= - \mathrm{L}_{\xi}^{A_0+\omega_0} B +\mathbb{W}_{p^{A}} & \mathbb{p}^{A}_{A}&= -\mathrm{L}_\xi^{A_0}(A-A_0) -\iota_\xi(F_{A_0})\\
        \mathbb{h}^{A}_e &= 0 & e\mathbb{h}^{A}_{\omega}&= \lambda e_n \Tr(B F_A )+ \frac{1}{4}  \lambda e_n e^2 \Tr(B, B)-\lambda e \Tr(B(B,e_ne)) +\mathbb{V}_{h^{A}}\\
        e^2\mathbb{h}^{A}_{B}&=   2\lambda e_n e d_{A+ \omega}B  & e^2\mathbb{h}^{A}_{A}&=  2\lambda e_n e F_A + \frac{2}{3!} (\lambda e_n e^3, B).
		\end{align*}
        where $\mathbb{V}_{l^{A}}$, $\mathbb{V}_{M^{A}}$, $\mathbb{V}_{p^{A}}$, $\mathbb{V}_{h^{A}}$ $\mathbb{W}_{l^{A}}$, $\mathbb{M}^{A}_{B}$, $\mathbb{W}_{p^{A}}$, and 
		$\mathbb{h}^{A}_{B}$ are such that the vector fields 
        $\mathbb{L}+\mathbb{l}^{A}$, $\mathbb{M}^{A}$ $\mathbb{P}+\mathbb{p}^{A}$ and $\mathbb{H}+\mathbb{h}^{A}$ are tangent to the structural constraints \eqref{e:structural_constraint_grav} and \eqref{e:constraintYM}.

Finally here are the Hamiltonian vector field in presence of a spinor field:
\begin{align*}
        \mathbb{l}^{\psi}_e &= 0  & \mathbb{l}^{\psi}_{\omega} &= \mathbb{V}_{l^{\psi}} \\
        \mathbb{l}^{\psi}_{\psi}& = [c,\psi]  &\mathbb{l}^{\psi}_{\overline{\psi}} &= [c,\overline{\psi}] \\
		\mathbb{p}^{\psi}_e &= 0 & \mathbb{p}^{\psi}_{\omega}&=\mathbb{V}_{p^{\psi}}\\
        \mathbb{p}^{\psi}_{\psi}&= - \mathrm{L}_{\xi}^{\omega_0} \psi & \mathbb{p}^{\psi}_{\overline{\psi}}&= - \mathrm{L}^{\omega_0}_{\xi} \overline{\psi}\\
	    \mathbb{h}^{\psi}_e &=  -\lambda (\sigma-\widetilde{\sigma}) & 
	    e \mathbb{h}^{\psi}_\omega &=  i \frac{\lambda e_n}{4} e (\overline{\psi}\gamma d_\omega \psi - d_\omega\overline{\psi} \gamma \psi)\\
	    \frac{e^3}{3!}\gamma \mathbb{h}^{\psi}_{\psi}&= \frac{\lambda e_n}{2} e^2 \gamma d_\omega \psi + \frac{\lambda e^2 e_n}{4}(\widetilde{\sigma}-2\sigma)\gamma \psi &
	    \frac{e^3}{3!}\mathbb{h}^{\psi}_{\overline{\psi}} \gamma &= \frac{\lambda e_n}{2} e^2  d_\omega \overline{\psi} \gamma- \frac{\lambda e^2 e_n}{4}\overline{\psi}\gamma(\widetilde{\sigma}-2\sigma)
	    \end{align*}

        where $\mathbb{V}_{l^{\psi}}$, $\mathbb{V}_{p^{\psi}}$ and $\mathbb{h}^{\psi}_{\omega}$  are such that the vector fields 
        $\mathbb{L}+\mathbb{l}^{\psi}$, $\mathbb{P}+\mathbb{p}^{\psi}$ and $\mathbb{H}+\mathbb{h}^{\psi}$ are tangent to the structural constraints \eqref{e:omegareprfix2spin}.

Let us also summarize here the Hamiltonian vector fields of the constraints $R_{\tau}$ and $r_\tau^\psi$. These were computed in \cite{CCT2020} and \cite{CFHT23} respectively.
        The Hamiltonian vector fields of $R_\tau$  are given by
    \begin{align*}
        e\mathbb R_{e}&=[e,\tau] &
        e\mathbb R_{\omega}&=g(\tau,e,\omega)+d_\omega \tau.
    \end{align*}
    where $g=g(\tau,e,\omega)$ is a form arising from the variation of $\tau$ along $e$. 
        The Hamiltonian vector fields of $R^\psi_\tau$  are given by
    \begin{align*}
            e\mathbb r^{\psi}_{e}&=0 & 
            e\mathbb r^{\psi}_{\omega}&=0 \\
            e\mathbb r^{\psi}_{\psi}&=3[\tau,\psi]&
            e\mathbb r^{\psi}_{\bar\psi}&=3[\tau,\bar\psi].
    \end{align*} 

\section{Definition of some maps and spaces}\label{a:def_maps_and_spaces}
One of the main differences between the degenerate and the non-degenerate case is given by the structural constraint of the geometric phase space of gravity and in the presence of an additional constraint. In order to understand why this happens and the objects involved, we need to define some maps and some spaces. 
\begin{definition}\label{d:def_W}
Let $e\in\Omega_{e_n}(\Sigma, \mathcal{V})$. Then, we define the following map:
    \begin{align*}
                W_1^{\partial, (i,j)}\colon \Omega_\partial^{i,j}  & \longrightarrow \Omega_\partial^{i+1,j+1} \\
                \alpha  & \longmapsto    e\wedge \alpha. \nonumber
    \end{align*}
\end{definition}
\begin{lemma}\label{lem:strconstr_free}
Let $g^\partial$ be non-degenerate. Then, for $\alpha\in\Omega_\partial^{2,1}$ 
$\alpha = 0$
if and only if 
	\begin{equation}\label{e:nd_str_constraint}
    e_n \alpha \in \Ima W_{1}^{\partial,(1,1)}
    \end{equation}
and
	\begin{equation*}
        e\alpha=0.
	\end{equation*}
 \end{lemma}
 For a proof we refer to \cite[Lemma 13]{CCS2020}.

 If we apply this lemma to $\alpha=d_\omega e$, \eqref{e:nd_str_constraint} becomes the structural constraint \eqref{e:structural_constraint_grav} and guarantees the equivalence between the equation $d_\omega e=0$ and the one generated by the constraint $L_c$, $ed_\omega e=0$. Furthermore it has been proved that this structural constraint is exactly what is needed for $F^{\partial}$ to be a symplectic space with symplectic form $\varpi$ (\cite[Theorem 17]{CCS2020}).

In the degenerate case the results of Lemma \ref{lem:strconstr_free}
are no longer true and need to be adapted. In order to do so, we must define some maps and spaces.

\begin{definition}\label{d:spaces_and_maps}
    Let $e\in\Omega_{e_n}(\Sigma, \mathcal{V})$. Then, we define the following maps:
    \begin{align*}
                    \varrho^{(i,j)} \colon \Omega_{\partial}^{i,j}  & \longrightarrow \Omega_{\partial}^{i+1,j-1} \\
                    \alpha & \longmapsto [e,\alpha] \nonumber
    \end{align*}
    \begin{align*}
                    \widetilde{\varrho}^{(i,j)} \colon \Omega_{\partial}^{i,j}  & \longrightarrow \Omega_{\partial}^{i+1,j-1} \\
                    \alpha & \longmapsto [\widetilde{e},\alpha], \nonumber
    \end{align*}
    with $\widetilde{e}\in \widetilde\Omega_\partial^{1,1}$ such that $\widetilde{e}^*\eta=0$.\footnote{Note that such $\widetilde{e}$ exists only for degenerate boundary metrics.}
    Furthermore, let $J$ be a complement\footnote{To obtain an explicit expression for the complement, one can follow these steps. Start by selecting an arbitrary Riemannian metric on the boundary manifold $\Sigma$ and extend it to the space $\Omega^{2,1}$. Then, the orthogonal complement of the image of the map $\varrho^{(1,2)} |{\mathrm{Ker} W_{1}^{\partial, (1,2)}}$ in $\Omega_{\partial}^{2,1}$ can be identified as the space $J$, with respect to the chosen Riemannian metric.} in $\Omega_{\partial}^{2,1}$ of the space $\Ima \varrho^{(1,2)} |_{\mathrm{Ker} W_{1}^{\partial, (1,2)}}$. Then, we define the following subspaces:
    \begin{align*}
    \mathcal{T}&\coloneqq \mathrm{Ker}W_{1}^{\partial (2,1)} \cap J \subset \Omega_{\partial}^{2,1}\\
    \mathcal S&\coloneqq\mathrm{Ker}W_{1}^{\partial, (1,3)} \cap \mathrm{Ker} \widetilde{\varrho}^{(1,3)}  \subset \Omega^{1,3}_{\partial}\\
    \mathcal K&\coloneqq \mathrm{Ker}W_{1}^{\partial, (1,2)} \cap \mathrm{Ker} \varrho^{(1,2)} \subset \Omega_{\partial}^{1,2}.
    \end{align*}
\end{definition}
\begin{remark}
   It is worth noting that in the non-degenerate case, the subspaces $\mathcal{T}$, $\mathcal{S}$, and $\mathcal{K}$ defined above are all trivial. This observation allows us to consider the approach used in the degenerate case as a generalization of the methodology employed in the non-degenerate scenario.
\end{remark}
For a degenerate boundary metric, the result corresponding to Lemma \ref{lem:strconstr_free} is as follows.
\begin{lemma}
    Let $\iota^*g$ be degenerate and $\alpha\in\Omega_\partial^{2,1}$. Then, we have that
    $$\alpha = 0$$
    if and only if 
    \begin{align}\label{e:tau_strcontr}
        e_n (\alpha- p_{\mathcal{T}}\alpha) \in \Ima W_{1}^{\Sigma,(1,1)}\\
            p_{\mathcal T}\alpha=0\label{e:degeneracy}\\
            e\alpha=0, \nonumber
    \end{align}
    where $p_\mathcal{T}$ is the projector onto $\mathcal T$.
\end{lemma}
When we apply this lemma to $\alpha=d_\omega e$ we call the condition \eqref{e:tau_strcontr} the \emph{degenerate structural constraint} and the condition \eqref{e:degeneracy} the \emph{degeneracy constraint}. This last constraint can be rewritten using the following equivalence
    	\begin{equation*}
        p_{\mathcal T}\alpha=0\quad\Longleftrightarrow\quad\int_\Sigma \tau\alpha=0\quad\forall\tau\in\mathcal{S},
	\end{equation*}
 where we observe the presence of the integral condition on the right-hand side, which is subsequently treated as an additional constraint of the theory and called $R_\tau$.

Since \eqref{e:tau_strcontr} is weaker than \eqref{e:nd_str_constraint}, in order to still get a symplectic space we need an additional condition given by the following Lemma.
\begin{lemma}
 Let $\iota^*g$ be degenerate. Then, given $\omega\in\Omega_\partial^{1,2}$, the conditions
    \begin{equation}\label{e:structural_constraint_gravity_degenerate}
        \begin{cases}
             e_n (d_\omega e- p_{\mathcal{T}}(d_{\omega} e)) \in \Ima W_{1}^{\Sigma,(1,1)}\\[6pt]
           p_{\mathcal{K}} \omega = 0
        \end{cases}
    \end{equation}
    uniquely define a representative of the equivalence class $[\omega]_e$, where $\omega \sim \omega'$ if $\omega=\omega'+v$ with $ev=0$.
\end{lemma}

Strictly speaking, on a light-like boundary, the combination of the structural and degeneracy constraints, along with the additional equation $p_\mathcal{K}\omega=0$, is required to ensure the equivalence between $d_\omega e=0$ and $e d_\omega e=0$ on the boundary and to uniquely determine the representative of the equivalence class $[\omega]_e$. Specifically, the structural constraint, along with the constraint $R_\tau$, guarantees the aforementioned equivalence condition, whereas the structural constraint, together with $p_\mathcal{K}\omega=0$, ensures the uniqueness of the representatives. For further details on this discussion we refer to \cite{CCT2020} and \cite{CFHT23}.

\emergencystretch=2em
\newrefcontext[sorting=nty]
\printbibliography

\end{document}

%% file: Packages.tex
    \usepackage[english]{babel} 

    \usepackage{graphicx}
    \usepackage[]{color}  
    \usepackage[dvipsnames]{xcolor} 
    \usepackage{csquotes} 
    \usepackage{appendix} 
    \usepackage[a4paper]{geometry} 
    \usepackage{cancel} 
    \usepackage{tikz}
        \usetikzlibrary{backgrounds}

    \usepackage{amsmath,amssymb,amsfonts,amsthm} 
    \usepackage{mathtools,todonotes} 
    \usepackage{tikz-cd}  
    \usepackage[all,cmtip]{xy} 
    \usepackage[bbgreekl]{mathbbol} 
    \usepackage{array} 

    \usepackage{hyperref} 

    \usepackage{authblk} 

%% file: technical.tex
\theoremstyle{definition} 
    \newtheorem{definition}{Definition}

\theoremstyle{plain} 
    \newtheorem{theorem}[definition]{Theorem}
    
    \newtheorem{lemma}[definition]{Lemma}

\theoremstyle{remark} 
    \newtheorem{remark}[definition]{Remark}

\usepackage[bibstyle=alphabetic,citestyle=alphabetic,useprefix,giveninits=true, sorting=ynt, sortcites, minbibnames=99,maxbibnames=99,backend=biber]{biblatex}  
\renewbibmacro{in:}{} 
\bibliography{bibliography}  
\DeclareSourcemap{
  \maps[datatype=bibtex]{
    \map[overwrite]{ 
      \step[fieldsource=doi, final]
      \step[fieldset=url, null]
      \step[fieldset=eprint, null]
      }
     \map[overwrite]{ 
      \step[fieldsource=eprint, final]
      \step[fieldset=pages, null]
      \step[fieldset=eid, null]
      \step[fieldset=journal, null]
    }  
  }
}

    \newcommand{\Tr}[0]{\text{Tr}}
    
    \DeclareMathOperator{\Ima}{Im}

    \newcommand{\qsp}[2]{\,\ensuremath{\raise.5ex\hbox{$#1$}\big\slash\raise-.5ex\hbox{$#2$}}}

    \newcommand{\FF}{\mathfrak{F}}

    \makeatletter
        \newcommand{\zzlabel}[1]{\ifmeasuring@\else\ltx@label{#1}\fi} 
    \makeatother

    \newcounter{terms}[equation] 

